\documentclass{jfp1}
\newif\ifdraft\draftfalse

\usepackage{amssymb,amsmath}
\usepackage{mathbbol}
\usepackage{cmll}
\usepackage{stmaryrd}
\usepackage{listings,color}
\usepackage{flushend}
\usepackage[utf8]{inputenc} % for proper diacritics: Universität, Hüttel, Luís
\usepackage[T1]{fontenc}
\usepackage{hyperref}
\usepackage{etoolbox}
\appto\UrlBreaks{\do\-}

\usepackage{mathpartir}
\usepackage{bcprules}
\bcprulessavespacetrue

\newcommand{\rtext}[1]{\makebox[0em][r]{#1}}
\newcommand{\boxedrel}[1]{\hfill \fbox{$#1$}}

\newcommand{\spacer}{\vspace{0ex}}
\newcommand{\nextline}{\\[1ex]} % was 1.5
\newcommand{\gap}{\qquad}
\newcommand{\hcaption}[1]{\vspace{2.5ex}\hrule\vspace{-2.5ex}\caption{#1}}

\usepackage{xcolor}

\newcommand{\mygray}[1]{{\color{gray}#1}}

\newcommand{\GV}{GV}
\newcommand{\GGV}{Gradual GV}
\newcommand{\UGV}{Uni GV}
\newcommand{\GGVe}{GGV${}_{\!e}$}
\newcommand{\GGVi}{GGV${}_{\!i}$}

\newcommand{\plingk}{\texttt{!}}
\newcommand{\queryk}{\texttt{?}}
\newcommand{\pl}[1]{\plingk{#1}.\,}
\newcommand{\qu}[1]{\queryk{#1}.\,}
\newcommand{\Endpl}{\End_\plingk}
\newcommand{\Endqu}{\End_\queryk}
\newcommand{\unit}{\unitk}

\newcommand{\PAR}{\mid}

\newcommand{\reduces}{\longrightarrow}
\newcommand{\nreduce}{\:\not\!\!\reduces}
\newcommand{\subs}[2]{[{#1}/{#2}]}
\newcommand{\context}[2][E]{{#1}[{#2}]}
\newcommand{\contextE}[1]{E[#1]}
\newcommand{\contextF}[1]{F[#1]}

\newcommand{\subj}{\meta{subj}}
\newcommand{\agree}{\meta{agree}}
\newcommand{\un}{\meta{un}}
\newcommand{\lin}{\meta{lin}}
\newcommand{\dual}[1]{\overline{#1}}

\newcommand\Return{\ensuremath{\meta{return}}}

\newcommand{\eqdef}{\triangleq} % equal by definition

\usepackage{wasysym}

\newcommand\DYN{\mathord{\star}}%\ast
\newcommand\DC{\mathord{\text{\textcircled{$\star$}}}} %\DC{\mathord{\logof}}%\circledast
\renewcommand\DC{\mathord{\text{\textcircled{$\star$}\rule{0pt}{0.5\baselineskip}}}} %\DC{\mathord{\logof}}%\circledast

\newcommand{\meta}[1]{\mathsf{#1}}
\newcommand{\keyword}[1]{\mathsf{#1}}

\newcommand{\boolk}{\keyword{bool}}
\newcommand{\intk}{\keyword{int}}

\newcommand{\stringk}{\keyword{string}}
\newcommand{\CCk}{\keyword{CC}}
\newcommand{\URLk}{\keyword{URL}}
\newcommand{\unitk}{\keyword{unit}} % for fork
\newcommand{\sendk}{\keyword{send}}
\newcommand{\recvk}{\keyword{receive}}
\newcommand{\newk}{\keyword{new}}
\newcommand{\forkk}{\keyword{fork}}
\newcommand{\closek}{\keyword{close}}
\newcommand{\waitk}{\keyword{wait}}

\newcommand{\letk}{\keyword{let}}
\newcommand{\ink}{\keyword{in}}

\newcommand{\selectk}{\keyword{select}}
\newcommand{\casek}{\keyword{case}}
\newcommand{\ofk}{\keyword{of}}

\newcommand{\End}{\keyword{end}}
\newcommand{\blamek}{\keyword{blame}}

\newcommand{\new}{\newk}
\newcommand{\fork}[1]{\forkk\,{#1}}
\newcommand{\send}[2]{\sendk\,{#1}\,{#2}}
\newcommand{\recv}[1]{\recvk\,{#1}}
\newcommand{\close}[1]{\closek\,{#1}}
\newcommand{\wait}[1]{\waitk\,{#1}}
\newcommand{\letin}[3]{\letk\,{#1}={#2}\,\ink\,{#3}}

\newcommand{\select}[2]{\selectk{\,{#1}\,{#2}}}

\newcommand{\casefst}[1]{\casek\,{#1}\,}
\newcommand{\casesnd}[1]{\ofk\,{#1}}
\newcommand{\case}[2]{\casek\,{#1}\,\ofk\,{#2}}
\newcommand{\casebranch}[3]{#1 \colon #2.#3}

\newcommand{\blamegc}[2]{\blamek\:#1\:#2}
\newcommand{\blameboth}[3]{\blamek\:#1\:#2\:#3}

\newcommand{\proc}[1]{\langle{#1}\rangle}

\newcommand{\br}[1]{\{#1\}}

\newcommand{\grmeq}{\;::= \;}
\newcommand{\grmor}{\:\mid\:}

\lstdefinelanguage{gst}{
  language=haskell,
  basicstyle=\sffamily\small,
  extendedchars=true,
  breaklines=true,
  morekeywords={new,send,receive,skip,select,dualof,int,string,bool,unit,fork,close,wait},
  tabsize=8,
  literate=
    {end!}{{$\textbf{end}_{!}$}}4 
    {end?}{{$\textbf{end}_{?}$}}4 
    {oplus}{$\oplus$}1 
    {otimes}{$\otimes$}1 
    {forall}{$\forall$}1
    {alpha}{$\alpha$}1 
    {beta}{$\beta$}1 
    {DYN}{$\DYN$}1
    {DC}{$\DC$}1
    {=>}{$\CastInc$}2 
    {=9=>}{$\CastNine$}2 
    {lambda}{$\lambda$}1 
    {->}{$\rightarrow$}2
    {-o}{$\multimap$}1
}
\newcounter{labelcounter}
\newcommand{\CastInc}{\stepcounter{labelcounter}\Cast{\ell_{\arabic{labelcounter}}}}
\newcommand{\CastNine}{\Cast{\ell_9}}

\newcommand\cgNothing\cdot
\newcommand\tcBase{B}

\newcommand\tcLolli\multimap
\newcommand\tcFun\to

\newcommand{\tcpair}{\times} % linear or unrestricted 
\newcommand\tBase\tcBase

\newcommand\dom{\mathsf{dom}}

\newcommand{\cast}{\Rightarrow}
\newcommand{\Cast}[1]{\stackrel{#1}{\cast}}

\newcommand\lockedbl[1]{\keyword{locked}\:#1}
\newcommand{\wrap}[3][{}]{{#2}:{#3}\Cast{#1}\DYN}

\newcommand{\flv}[1]{\meta{flv}(#1)}
\newcommand{\fn}[1]{\meta{fn}(#1)}

\newcommand{\name}{z}

\newcommand\Storea{a}
\newcommand\Storeb{b}

\newcommand\Chanc{c}
\newcommand\Chand{d}

\newcommand\Multm{m}
\newcommand\Multn{n}
\newcommand{\Multsub}{\Sub}

\newcommand{\VarsX}{X}  % sets of variables

\newcommand{\ground}[1]{\textbf{#1}}
\newcommand\GroundS{\ground{S}}
\newcommand\GroundR{\ground{R}}

\newcommand\GroundT{\ground{T}}
\newcommand\GroundU{\ground{U}}

\newcommand\Sub{\mathrel{<:}}
\newcommand\super[2]{#1^{:>}(#2)}
\newcommand\notSub{\mathrel{\not\!\Sub}}
\newcommand\SubGV{\Sub}
\newcommand\SafeForP{\ensuremath{\keyword{safe\ for}}} % the predicate
\newcommand\SafeFor{\ensuremath{\ \SafeForP\ }}
\newcommand{\Subn}{\sqsubseteq}

\newcommand\Embed[1]{\lceil#1\rceil}

\newcommand{\citet}{\cite}
\newcommand{\citep}{\cite}
\newtheorem{theorem}{Theorem}
\newtheorem{proposition}{Proposition}
\newtheorem{lemma}{Lemma}
\newtheorem{corollary}{Corollary}

\usepackage{algpseudocode}      % for keywords
\usepackage{algorithm}
\usepackage{algorithmicx}

\newcommand{\ottnt}[1]{\mathit{#1}}
\newcommand{\ottmv}[1]{\mathit{#1}}

\newcommand{\ottsym}[1]{#1}

\algblockdefx{Case}{EndCase}[1]{\textbf{case} #1 \textbf{of}}{\textbf{end case}}
\algblockdefx{CaseItem}{EndCaseItem}[1]{$\mid$ #1 $\Rightarrow$}{\textbf{end case item}}
\algtext*{EndCaseItem}     % not display EndCaseItem
\algtext*{EndCase}
\algtext*{EndFor}
\algtext*{EndIf}
\algtext*{EndFunction}
\newcommand\Assert{\textbf{assert}\;}
\newcommand{\Error}{\textbf{error}}
\newcommand\asgn{\mathrel{:=}} % assignment
\newcommand\WC{\_}             % Wildcard pattern
\usepackage{letltxmacro}
\renewcommand\Return{\textbf{return}\;}
\LetLtxMacro{\oldElse}{\Else}
\renewcommand\Else{\oldElse\;} % add space

\newcommand\tyexp{CheckExpr}

\newcommand\remove{\meta{rm}}%{\meta{remove}}

\newcommand\set[1]{\{#1\}}
\newcommand\join{\vee}

\newcommand\bigjoin{\bigvee}

\newcommand\ciarrow{\rightsquigarrow}

\newcommand\vdashG{\vdash_{\!e}} % external; G for Gradual
\newcommand\vdashC{\vdash_{\!i}} % internal; C for Cast
\newcommand\vdashS{\vdash_{\!\mathrm{GV}}} % Static
\newcommand\matchingsym{\mathrel{\triangleright}}
\newcommand\matching[2]{#1 \matchingsym #2}

\newcommand\Erase[1]{|#1|}

\title{Gradual Session Types}

\author[Igarashi et al.]{ATSUSHI IGARASHI \\
  Kyoto University, Japan \\
  PETER THIEMANN \\
  University of Freiburg, Germany \\
  YUYA TSUDA \\
  Kyoto University, Japan \\
  VASCO T.\ VASCONCELOS \\
  LASIGE, Faculty of Sciences, University of Lisbon, Portugal \\
  PHILIP WADLER \\
  University of Edinburgh, Scotland \\
  \email{igarashi@kuis.kyoto-u.ac.jp,
    thiemann@informatik.uni-freiburg.de,
    tsuda@fos.kuis.kyoto-u.ac.jp,
  vv@di.fc.ul.pt,
  wadler@inf.ed.ac.uk}
}

\begin{document}

\maketitle

\begin{abstract}
Session types are a rich type discipline, based on linear types, that
lifts the sort of safety claims that come with type systems to
communications.  However, web-based applications and microservices
are often written in a mix of languages, with type disciplines in a
spectrum between static and dynamic typing.  Gradual session types
address this mixed setting by providing a framework which grants
seamless transition between statically typed handling of sessions and
any required degree of dynamic typing.

We propose \GGV{} as a gradually typed extension of the functional
session type system GV.  Following a standard framework of gradual
typing, \GGV{} consists of an external language, which relaxes the
type system of GV using dynamic types, and an internal language with
casts, for which operational semantics is given, and a cast-insertion
translation from the former to the latter.  We demonstrate type and
communication safety as well as blame safety, thus extending previous
results to functional languages with session-based communication.  The
interplay of linearity and dynamic types requires a novel approach to
specifying the dynamics of the language.
\end{abstract}

\section{Introduction}

It was the best of types, it was the worst of types.

A survey of the top-20 programming languages to learn for open source
projects\footnote{\url{https://www.ubuntupit.com/top-20-most-popular-programming-languages-to-learn-for-your-open-source-project/}
accessed in April 2019}
lists eight dynamically-typed languages (JavaScript, Python, Ruby, R,
PHP, Perl, Scheme, Erlang) and states that developer salaries for
these languages are among the highest in the industry.
The survey also suggests to learn languages with 
elaborate static type systems like Rust, Scala, and Haskell, with developers
earning even higher salaries. These languages derive their
expressiveness from advanced type system features like 
linearity; uniqueness; effects; dependent types as embodied in
research languages like Agda \citep{Norell-2009}, Coq \cite{the_coq_development_team_2019_2554024}, and Idris \citep{Brady-2013}; and session types as
in Links \citep{Lindley-and-Morris-2016}.
This data indicates two opposing trends in current industrial
practice, one asking for dynamically-typed
programming and another asking for expressive statically-typed
programming.

Gradually-typed languages reconcile these two trends. They permit one
to assemble programs with some components written in a
statically-typed language and some in a dynamically-typed
language. Gradually-typed languages have been widely explored in both
theory and practice, beginning with contracts in Racket
\cite{Findler-Felleisen-2002} and their interfacing with
TypedRacket \cite{DBLP:conf/popl/Tobin-HochstadtF08} and then
popularized by Siek and others
\citet{Siek-Taha-2006,DBLP:conf/ecoop/SiekT07,Siek-et-al-2015-criteria}. They
are geared towards safely interconnecting dynamically-typed parts with
statically-typed parts of a program by ensuring that type mismatches
only occur in the dynamically-typed parts \cite{Wadler-Findler-2009}.

Dynamics in  C\#  \citep{Bierman-et-al-2010},
Microsoft's TypeScript\footnote{\url{https://www.typescriptlang.org/}
  accessed in April 2019} \cite{Bierman-et-al-2014}, Google's Dart
\cite{Dart-2014,DBLP:journals/scp/ErnstMSS17}, and Facebook's Hack \cite{Verlaguet-2013} and Flow
\cite{DBLP:journals/pacmpl/ChaudhuriVGRL17} are industrial systems
inspired by gradual typing, but focusing on enhancing programmer
productivity and bug finding rather than containing type mismatches. Systems
such as Racket \cite{Findler-Felleisen-2002} and Reticulated Python
\cite{DBLP:conf/popl/VitousekSS17} rely on contracts or similar 
constructs to ensure that dynamically-typed values adhere to
statically-typed constraints when values pass from one world to the
other.

At first blush, one might consider gradual types as largely a response
to the former trend: they provide a way for developers using
dynamically-typed languages to evolve their code toward
statically-typed languages that are deemed easier to maintain. But on
second thought, one might consider gradual types as even more helpful
in light of the latter trend.  Suitably generalized, gradual typing
can mediate between simple type systems and type systems that feature
dependent types, effect types, or session types, for example.  Gradual
typing in this sense can help in evolving 
software development toward languages with more precise type systems.

Hence, an important line of research is to extend gradual typing so
that it not only relates dynamically-typed and statically-typed
languages, but also relates less-precisely-typed and more-precisely-typed
languages.  There is already some research on doing so for
dependent types~\cite{Ou-et-al-2004,Flanagan-2006,Greenberg-et-al-2010,DBLP:conf/popl/LehmannT17}, effect
types~\cite{Banados-et-al-2014}, typestate \cite{Wolf-et-al-2011}, and
several others which we review in the section on
related work.  This paper presents the first system
that extends gradual typing to session types.

Session types were introduced by Honda \shortcite{Honda-1993}, drawing
on Milner's $\pi$-calculus \cite{Milner-et-al-1992} and Girard's
linear logic \cite{Girard-1987}, and further developed by many others
\cite{Honda-et-al-1998,Yoshida-Vasconcelos-2007}.  Gay and Hole \shortcite{Gay-Hole-2005}
introduced subtyping for session types, and
session types were embedded into a functional language with linear
types, similar to the one used in this paper, by Gay and Vasconcelos
\shortcite{Gay-Vasconcelos-2010}. Caires, Pfenning, Toninho, and
Wadler introduced propositions-as-types interpretations of session
types in linear logic
\cite{Caires-Pfenning-2010,Caires-et-al-2014,Wadler-2012,Wadler-2014}.
One important line of research is multiparty session types
\cite{Honda-et-al-2008,Honda-et-al-2016} but we confine our attention
here to dyadic session types.

Session types have been adapted to a variety of languages,
either statically or dynamically checked, and
using either libraries or additions to the toolchain;
implementations include C, Erlang, Go, Haskell, Java, Python, Rust, and Scala.
New languages incorporating session types include 
C0~\cite{Willsey-et-al-2016}, Links~\cite{Cooper-et-al-2007}, SePi~\cite{franco.vasconcelos:sepi}, SILL~\cite{Pfenning-Griffith-2015}, and Singularity~\citep{Fahndrich-et-al-2006}.
Industrial uses of session types include:
Red Hat's support of the Scribble specification language \citep{Yoshida-et-al-2013},
which has been used as a common interface for several systems based on
session types; Estafet's use of session types to manage
microservices\footnote{\url{http://estafet.com/scribble/} Accessed in April 2019}; and the Ocean Observatories Initiative's use
of dynamically-checked session types in Python~\cite{DBLP:journals/fmsd/DemangeonHHNY15}.
Session types inspired an entire line of research on
what has come to be called behavioural types, the subject of
EU COST action BETTY, a recent Shonan meeting,
and a recent Dagstuhl seminar.

Here is a simple session type encoding of a protocol to purchase
an online video:
\[
S_{\textit{video}} =
\pl{\stringk}\qu{\intk}
  \oplus\br{\mathit{buy}: \pl{\CCk}\qu{\URLk}\Endqu, ~\mathit{quit}: \Endpl}.
\]
It describes a channel endpoint along which a client sends the name
of a video as a string, receives its cost as an integer, and then
selects either to buy the video, in which case one sends a credit card
number, receives a URL from which the video may be downloaded, and
waits for an indication that the channel has been closed, or selects
to quit and closes the channel.  There is a dual session type for
server at the other end of the channel, where $!$ (write) is swapped
with $?$ (read), $\oplus$ (select from a choice) is swapped with
$\with$ (offer a choice), and $\End_!$ (close a channel) is swapped
with $\End_?$ (wait for a channel to close).

Session types are necessarily linear.  Let $x$ be bound to a string
and let $c$ be bound to a channel endpoint
of type $S_{\textit{video}}$.  Performing
\[
\letk~d=\send{x}{c}~\ink~\ldots
\]
binds $d$ to a channel endpoint of type $R$, where $S_{\textit{video}} = {!\stringk.R}$.
To avoid sending a string to the same channel twice, it is essential that
$c$ must be bound to the only reference to the channel endpoint before the
operation, and for similar reasons $d$ must be bound to the only reference
to the channel endpoint after.  Such restrictions can easily be enforced
in a statically-typed language with an affine type
discipline. Linearity is required to guarantee that channels are not
abandoned before they are closed.

But how is one to ensure linearity in a dynamically-typed language?
Following Tov and Pucella \shortcite{Tov-Pucella-2010}, we require that each dynamically-typed
reference to a channel endpoint is equipped with a lock. That
reference is locked after the channel is used once to ensure it cannot be
used again.  To ensure that each channel is appropriately terminated,
with either a wait or a close operation, garbage collection flags
an error if a dynamically-typed reference to a channel becomes
inaccessible.

Our system is the first to integrate static and dynamic session types
via gradual typing.
It preserves the safety properties of statically-typed sessions,
namely progress (for expressions), preservation, and absence of run-time errors.
The latter includes session fidelity: every send is
matched with a receive, every select is matched with an offer,
and every wait is matched with close.  Many, but not all, systems
with session types support recursive session types, and many,
but not all, systems with session types ensure deadlock freedom;
we leave such developments for future work.

Previous systems that perform dynamic monitoring on session types
include the work on Scribble \cite{Yoshida-et-al-2013} which applies
the ideas developed for distributed monitoring of protocols to
multiparty session types
\cite{Bocchi-et-al-2013-conference,Bocchi-et-al-2013}.
Gommerstadt and others \shortcite{DBLP:conf/popl/JiaGP16} consider dynamic
monitoring of higher-order session typed processes in the presence of unreliable
communication and malicious communication partners. Their focus is on
assigning blame correctly in this setting. 
The same authors \shortcite{DBLP:conf/esop/GommerstadtJP18}
develop a theory of contracts that translate into processes that serve
as proxies between the original communication partners. Proxies ensure
adherence to the session protocol with dynamic tests.
A similar proxy-based monitoring scheme was also proposed by one of
the authors \cite{DBLP:conf/tgc/Thiemann14} where gradual typing was
restricted to the transmitted values.
Melgratti and Padovani \shortcite{DBLP:journals/pacmpl/MelgrattiP17}
propose a contract system that mediates between (simply-typed)
sessions and contract-refined sessions. Enforcement is done with an
inline monitor.

In contrast to these approaches, our work applies to the mediation
between dynamically-typed and statically-typed code and it relies on gradual principles
that enable a pay-as-you-go approach: a protocol is checked
statically as much as possible, dynamic checks are only employed if
they cannot be avoided; full gradualization including the
communication channel; no forced introduction
of proxies that may affect efficiency. 

We give our system a compact formulation along the lines of the
blame calculus \cite{Wadler-Findler-2009}, based on the notion
of a \emph{cast} to mediate interactions between more-precisely typed
(e.g., statically typed) and less-precisely typed (e.g., dynamically typed)
components of a program.  We define the four subtyping relations
exhibited by the blame calculus,
ordinary, positive, negative, and naive, and
show the corresponding results, including a tangram theorem
relating the four forms of subtyping and blame safety.
A corollary of our results is that
in any interaction between more-precisely typed and less-precisely
typed components of a program, any cast error is due to the
less-precisely typed component.

Our paper makes the following contributions.
\begin{itemize}
\item Section~\ref{sec:motivation}
  provides an overview of the novel techniques in our work,
  and how we dynamically enforce linearity and session types. 
\item Section~\ref{sec:language}
  describes a complete formal calculus, including syntax of both an external language, in which programs are written and run-time checking is implicit, and an internal language, in which programs are executed after run-time checking in the form of casts is made explicit;
  typing rules of the two languages; reduction rules for the internal language; 
  cast-insertion translation from the external to the internal language;
  and embedding of a
  dynamically typed language with channel-based communication
  into our calculus.
\item Section~\ref{sec:results}
  presents standard results for our calculus, including
  progress (for expressions) and preservation, session fidelity,
  the tangram theorem, blame safety, conservativity of the external language typing over fully static typing, and type preservation of the cast insertion translation.  We also discuss the gradual guarantee property for the external language.  It turns out that it fails to hold---we will analyze counterexamples
  and discuss why.
\end{itemize}
Section~\ref{sec:related} describes related work and
Section~\ref{sec:conclusions} concludes.

Compared to the previous paper
\cite{DBLP:journals/pacmpl/IgarashiTVW17}, we extend the development
with the external language, the cast-insertion translation, a type
checking algorithm, proofs of their properties, and analysis of the
failure of the gradual guarantee, as well as more detailed proofs for
the earlier results.  These extensions make gradual session types
accessible for the programmer, who works in the external language.

\section{Motivation}
\label{sec:motivation}

Sy and Rob collaborate on a project whose design is based on
microservices.
Sy is a strong advocate of static typing and relies on an
implementation language that supports session types out of the
box. Rob, on the other hand, is a strong advocate of dynamically
typed languages. 
One of the credos of microservice architectures
is that the implementation of a service endpoint is language-agnostic,
which means it can be implemented in any programming language
whatsoever as long as it adheres to its protocol.
However, Sy does not want to compromise
the strong guarantees (e.g., type safety, session fidelity) of the statically
typed code by communicating with Rob's client. Rob is also
keen on having strong guarantees, but does not mind if they are
enforced at run time. 
Here is the story how they can collaborate safely using
\GGV \footnote{GV is our name for the functional session type calculus
  of Gay and Vasconcelos \shortcite{Gay-Vasconcelos-2010}, which is
  the statically-typed baseline for our gradual system.}, our
proposal for a gradually typed functional language with synchronous
binary session types.

\subsection{A Compute Service}
\label{sec:compute-service}

The compute service is a simplified version of one of the
protocols in Sy and Rob's project. The service involves two peers, a
server and a client, connected via a communication link.  The server
runs a protocol that first offers a choice of two arithmetic
operations, negation or addition, then reads one or two numbers
depending on the operation, outputs the result of applying the
selected operation to its operand(s), and finally closes the
connection. The client chooses an operation by sending the server a
label, which is either \textit{neg} or \textit{add} indicating the
choice of negation or addition, respectively. In session-type
notation, the server's view of the compute protocol reads as follows.
\begin{eqnarray*}
  \text{\lstinline{Compute}} &=
                           \& \{ \mathit{neg}:{?\intk}.{!\intk}.\End_!
                     ,~ \mathit{add}: {?\intk}.{?\intk}.{!\intk}.\End_! \}
\end{eqnarray*}
Sy chooses to implement the server in the language GV that is inspired
by previous work \cite{Gay-Vasconcelos-2010} and that we will describe
formally in Section~\ref{sec:language}.
\begin{lstlisting}
computeServer : Compute -> unit
computeServer c =
  case c of {
    neg: c. let v1,c = receive c in
            let c = send (-v1) c in
            close c;
    add: c. let v1,c = receive c in
            let v2,c = receive c in
            let c = send (v1+v2) c in
            close c
  }
\end{lstlisting}
The parameter \lstinline{c} of type \lstinline{Compute} is the server's
endpoint of the communication link to the client (when unambiguous, we
often just say endpoint or channel). The \lstinline{case c of ...}
expression receives the client's choice on channel \lstinline{c} in
the form of a label \textit{neg} or \textit{add} and branches
accordingly.  The notation ``\lstinline{c. }'' in each branch
(re-)binds the variable \lstinline{c} to the channel in the state
after the transmission has happened. The type of \lstinline{c} is
updated to the session type corresponding to the respective branch in
the \lstinline{Compute} type.  The \lstinline{receive c} operation
receives a value on channel \lstinline{c} and returns a pair of the
received value and the depleted channel with a correspondingly
depleted session type. Analogously, the \lstinline{send v c} operation
sends value \lstinline{v} on channel \lstinline{c} and returns the
depleted channel. The final \lstinline{close c} disconnects the
communication link by closing the channel.

\subsection{The View from the Client Side}
\label{sec:view-from-client}

A client of the \lstinline{Compute} protocol communicates on a channel
with the protocol \lstinline{ComputeD}\relax defined below.  This protocol
is \textit{dual} to \lstinline{Compute}:
sending and receiving operations are swapped.
\[
  \text{\lstinline{ComputeD}} =
                      \oplus \{ \mathit{neg}:{!\intk}.{?\intk}.\End_?
                     , \mathit{add}: {!\intk}.{!\intk}.{?\intk}.\End_? \}
\]
A client of the compute service may always select the same operation
and then proceed linearly according the corresponding branch. Such a client can use 
a simpler supertype of \lstinline{ComputeD} with a unary internal choice. For example, a client that only ever asks for negation can implement \lstinline{ComputeDneg}. 
\[
  \text{\lstinline{ComputeDneg}} =
                      \oplus \{ \mathit{neg}:{!\intk}.{?\intk}.\End_? \}  
\]
Here is Sy's implementation of a typed client for \lstinline{ComputeDneg}.
\begin{lstlisting}
negationClient : int -> ComputeDneg -> int
negationClient v c =
  let c = select neg c in
  let c = send v c in
  let y,c = receive c in
  let _ = wait c in
  y
\end{lstlisting}
There are two new operations in the client code.
The \lstinline{select neg c} operation selects the \textit{neg} branch in the protocol by
sending the \textit{neg} label to the server. It returns a channel to
run the selected branch of the protocol with type ${!\intk}.{?\intk}.\End_?$.
The \lstinline{wait c} operation matches the \lstinline{close}
operation on the server and disconnects the client. 

\subsection{A Unityped Server}
\label{sec:going-unityped}

To test some new features, Rob also implements the \lstinline{Compute}
protocol, but does so in the unityped language \UGV, which
is safe but does not impose a static
typing discipline.
Here is Rob's implementation of the server.
\begin{lstlisting}
-- unityped
dynServer c = 
  case c of {
    neg: c. serveOp 1 (lambdax.-x) c;
    add: c. serveOp 2 (lambdax.lambday.x+y) c
  }

serveOp n op c =
  if n==0 then
    close (send op c)
  else
    let v,c = receive c in
    serveOp (n-1) (op v) c
\end{lstlisting}

The main function \lstinline{dynServer} takes a channel \lstinline{c}
on which it receives the client's selection. It delegates to an
auxiliary function \lstinline{serveOp} that takes the arity of a
function, the function itself, and the channel end on which to receive
the arguments and to send the result. The \lstinline{serveOp} function
counts down the number of remaining function applications in the first
argument, accumulates partial function applications in the second
argument, and propagates the channel end in the third argument.

It is easy to see that the \lstinline{dynServer} function implements
the \lstinline{Compute} protocol.
Rob chose this style of implementation because it is amenable to
experimentation with protocol extensions: the function \lstinline{dynServer} is trivially extensible to new
operations and types by adding new lines to the \lstinline{case}
dispatch. 

\subsection{The Gradual Way}
\label{sec:gradual-way}

How can we embed Rob's server with  other program fragments in the typed
language (e.g., Sy's client) while retaining as many typing guarantees
as possible?

One answer would be to use a dependently typed system that can describe
the type of the \lstinline{serveOp} function adequately. In an
extension of a recently proposed system
\cite{DBLP:conf/fossacs/ToninhoY18} with iteration on natural numbers
and large elimination, we might write that code as follows.
\begin{lstlisting}
  Op : nat -> Type
  Op 0 = int
  Op (n+1) = int -> Op n

  Ch : nat -> Session
  Ch 0 = !int.end!
  Ch (n+1) = ?int.Ch n
  
  serveOpDep : (n : nat) (op : Op n) (c : Ch n) -> unit
  serveOpDep 0 op c = close (send op c)
  serveOpDep (n+1) op c = let v,c = receive c in
                          serveOpDep n (op v) c
\end{lstlisting}
However, we are not aware of a fully developed theory of a session-type
system that would be able to process this definition.

An alternative that is immediately available is to resort to gradual
typing. For this particular program it will insert casts to make the
program type check, but all those casts are semantically guaranteed to
succeed because it would have a dependent type.
To this end, we rewrite the function \lstinline{dynServer} in
a gradually typed external language analogous to the gradually typed
lambda calculus GTLC \cite{Siek-et-al-2015-criteria}, but extended
with GV's communication operations.

In our example, the rewrite to the external language boils
down to providing suitable type signatures for \lstinline{dynServer}
and \lstinline{serveOp}:
\begin{lstlisting}
dynServer : Compute -> unit
serveOp : int -> DYN -> DC -> unit
\end{lstlisting}
The first argument \lstinline{n} of \lstinline{dynServer} is consistently
handled as an integer, so its type is \lstinline{int}.
The second argument \lstinline{op} is invoked with values of type
\lstinline{int -> int -> int}, \lstinline{int -> int}, and
\lstinline{int}: these types are subsumed to
the dynamic type $\DYN$.  Similarly to other gradual type systems,
an expression of type $\DYN$ can be used in any context, e.g.,
addition, function application, or even communication, 
and any value can be passed where $\DYN$ is expected.  
The third argument \lstinline{c} is invoked with channels of different types:
\lstinline{?int.?int.!int.end!}, \lstinline{?int.!int.end!}, and
\lstinline{!int.end!}. These types are subsumed to a type that is novel to this work, the \emph{dynamic session type}, $\DC$, a \emph{linear  
  type} which subsumes all session types.  
It is important to see that the channel \lstinline{c} is handled
linearly in functions \lstinline{dynServer} and
\lstinline{serveOp}. For that reason, the role and handling of the linear
dynamic session type with respect to the set of session types is
analogous to the role and handling of $\DYN$ with respect to general
types, as shown in earlier work \cite{DBLP:conf/sfp/FennellT12,DBLP:conf/tgc/Thiemann14}.
Aside from the type annotation, the code remains exactly the same as in the unityped case.

The external language comes with a
translation into a blame calculus with explicit casts. This
translation inserts just the casts that are necessary to make typing of the code go
through.
Here is the output of this translation (suffix \lstinline|Cast| is appended
to the names of the functions to distinguish different versions):
\begin{lstlisting}
dynServerCast : Compute -> unit
dynServerCast c =
  case c of {
    neg: c. serveOpCast 1 ((lambdax.-x) : int -> int => DYN) 
                        (c : ?int.!int.end! => DC);
    add: c. serveOpCast 2 ((lambdax.lambday.x+y) : int -> int -> int => DYN)
                        (c : ?int.?int.!int.end! => DC)
  }

serveOpCast : int -> DYN -> DC -> unit
serveOpCast n op c =
  if n==0 then
    close ((send op (c : DC => !DYN.DC)) : DC => end!)
  else
    let v,c = receive (c : DC => ?DYN.DC) in
    serveOpCast (n-1) ((op : DYN => DYN -> DYN) v) c
\end{lstlisting}
Casts of the form $e : T_1 \Cast{p} T_2$---meaning that $e$ of type $T_1$ is cast to $T_2$---are inserted where values are converted from/to $\DYN$ or $\DC$, similarly
to the translation from GTLC.  The \emph{blame labels} $\ell_1$, $\ell_2$, \ldots{} (ranged over by $p$ and $q$) on the arrow identify casts, when they fail.
The resulting casts in \lstinline{dynServerCast} and \lstinline{serveOpCast}
look fairly involved, but we should keep in mind that
the programmer does \emph{not} have to write them as they result
from the translation.
In practice, blame labels may contain information on program locations to help identify how a program fails.  For example, if Rob made the following mistake in writing his \lstinline{dynServer}
\begin{lstlisting}
  neg: c. serveOp 2 (lambdax.-x) c;
     -- The first argument to serveOp should be 1!
\end{lstlisting}
then a call to \lstinline{negationClient} would fail after the server
receives the first integer from the client.  More specifically, the failure would
identify the cast labeled $\ell_7$ failed because a channel endpoint
whose session type is ${!\intk}.\End_!$ had been flown from
$\ell_2$.

\subsection{Dynamic Linearity}
\label{sec:dynamic-linearity}

The refined criteria for gradual typing
\citet{Siek-et-al-2015-criteria} postulate that a gradual type system
should come with a full embedding of a unityped calculus.  This
embedding (which we indicate by ceiling brackets $\Embed{\dots}$)
extends the embedding given for the simply-typed lambda calculus
\citet{Wadler-Findler-2009} to handle the operations on sessions (see
Figure~\ref{fig:embedding-unityped} for its definition).

For example, (the unityped version of) the \lstinline{dynServer} as written by Rob is compiled
and embedded into the gradually typed language as a value
\lstinline{dynServer : DYN}. To directly incorporate Rob's code,
the gradual type checker enables Sy to write a function
\lstinline{callDynServer} that accepts a channel of type
\lstinline{Compute} and returns a value of type \lstinline{unit}, but internally just calls \lstinline{dynServer}.
\begin{lstlisting}
callDynServer : Compute -> unit
callDynServer c =
  dynServer c
\end{lstlisting}
The gradual type checker translates the definition of
\lstinline{callDynServer} by inserting the appropriate casts: it casts
the embedded \lstinline{dynServer} (of type $\DYN$) to the function
type \lstinline|DYN -> DYN|, it casts the channel argument to this
function to \lstinline{DYN}, and it casts the result to
\lstinline{unit}.
\begin{lstlisting}
callDynServer : Compute -> unit
callDynServer c =
  ((dynServer : DYN => DYN -> DYN) (c : Compute => DYN)) : DYN => unit
\end{lstlisting}
The casts inserted in this code make Sy's expectations completely obvious:
\lstinline{dynServer} must be a function and it is expected to use
\lstinline{c} as a channel of type \lstinline{Compute}. Any
misuse will allocate blame to the respective cast in
\lstinline{dynServer}.  % (We omit blame labels in the examples for simplicity.)

One kind of misuse that we have not discussed, yet, is compromising
linearity: Sy has no guarantee that Rob's code does not accidentally
duplicate or drop the communication channel. Both actions can lead to
protocol violations, which should be detected at run time. \GGV{}
takes care of linearity by factoring the cast
\lstinline{(c : Compute =9=>  DYN)} through the dynamic session type $\DC$:
\begin{lstlisting}
  ((c : Compute =9=> DC) : DC =9=> DYN)
\end{lstlisting}
The first part is a cast among linear (session) types and it can be
handled as outlined in Section~\ref{sec:gradual-way}.
The second part is a cast from a \emph{linear type} (which could be a session type, a linear function type,
or a linear product) to the \emph{unrestricted dynamic type} $\DYN$. 

A cast from a linear type to unrestricted $\DYN$ is a novelty of
{\GGV}.  Operationally, the cast introduces an indirection through a
store: it takes a linear value as an argument, allocates a new cell in
the store, moves the linear value along with a representation of its
type into the cell, and returns a handle $\Storea$ to the cell as an
unrestricted value of type $\DYN$.  {\GGV} represents the cell by a
process and creates handles by introducing an appropriate binder so
that a process of the form $\context{v \colon \DC \Cast{p} \DYN}$
reduces to
$(\nu\Storea)(\context{a} \PAR \Storea \mapsto v \colon \DC \Cast{p} \DYN) $.
Here, $(\nu\Storea)P$ represents the scope of a fresh reference to a linear value and the process $\Storea \mapsto v \colon \DC \Cast{p} \DYN$ represents
the cell storing $v$ at $a$.
Linear
use of this cell is controlled at run time using ideas for run-time monitoring of affine types
\cite{Tov-Pucella-2010,Padovani-2017}.

Any access to a cell comes in the guise of a cast $\Storea : \DYN
\Cast{p} T$ from $\DYN$ to
another type applied to a handle $\Storea$. 
If the first access to the cell is a cast from $\DYN$ to a
linear type consistent with the type representation stored in the cell, then
the cast returns the linear value and empties the cell.
Any subsequent access to the same cell results in a linearity violation which allocates blame to
the label on the cast from $\DYN$.
If the first cast attempts to convert to an inconsistent type, then
blame is allocated to that cast.
In addition, there is a garbage collection rule that fires when the
handle of a full cell is no longer reachable from any process. It
allocates blame to the context of the cast to $\DYN$ because that cast
violated the linearity protocol by dismissing the handle.

\subsection{End-to-end Dynamicity}
\label{sec:end-end-dynamicity}

The examples so far tacitly assume that channels are created with a
fully specified session type that provides a ``ground truth'' for the protocol on this channel.  Later
on, channels may be cast to $\DC$ and on to $\DYN$, but essentially they adhere to
the ground truth established at their creation.

Unfortunately, this view cannot be upheld in a calculus that is able
to embed a unityped language like \UGV. When writing \lstinline{new} in a
unityped program to create a channel, Rob (hopefully) has some session type in mind, but
it is not manifest in the code. 

In the typed setting, \lstinline{new} returns a linear pair of session
endpoints of type $S \tcpair_\lin \dual S$ where $S$ is the server
session type and $\dual S$ its dual client counterpart (cf.\ the
\lstinline{Compute} and \lstinline{ComputeD} types in
Sections~\ref{sec:compute-service} and~\ref{sec:view-from-client}).
When embedding the unityped \lstinline{new}, the session type $S$ is
unknown. Hence, the embedding needs to create
a channel without an inherent ground truth
session type. It does so by assigning both channel ends type
$\DC$ and casting it to $\DYN$ as in
\lstinline{new}$ \colon \DC \tcpair_\lin \DC \Cast{} \DYN $.
To make this work, the dynamic session type $\DC$ is considered \emph{self-dual}, that is $\dual\DC = \DC$.  
{\GGV} offers \emph{no static guarantees} for
either end of such a channel.

To see what run-time guarantees {\GGV} can offer for a
channel of unknown session type, let's consider the embedding of the dynamic send and receive
operations that may be applied to it.
The embedded send operation takes two arguments of type
$\DYN$, for the value and the channel, and returns the updated
channel wrapped in type $\DYN$.  The embedded receive operation takes
a wrapped channel of type $\DYN$ and
returns a ($\DYN$-wrapped) pair of the received value and the updated channel.
\begin{eqnarray*}
  \Embed{\send ef} &=&
                     (\send{\Embed{e}}{(\Embed{f} \colon \DYN \Cast{p} {!{\DYN}.\DC})})
                     \colon \DC \Cast{} \DYN
\\ %                     &&\quad&
  \Embed{\recv e} &=&
                    (\recv{(\Embed{e} \colon \DYN \Cast{q} {?{\DYN}.\DC})})
        \colon \DYN \times_\lin \DC \Cast{} \DYN
\end{eqnarray*}
(Here, $p$ and $q$ are metavariables ranging over blame labels.)
Now consider running the following unityped program with entry point  \lstinline{main}.
\begin{lstlisting}[numbers=left,escapechar=|]
client cc = 
  let v,cc = receive cc in wait cc
server cs = 
  let cs = send 42 cs in close cs
main =
  let cs,cc = new in|\label{line:new}|
  let _ = fork (client cc) in
  server cs
\end{lstlisting}
After a few computation steps, it reaches a configuration where the
client and the server have reduced to
$(\nu cc,cs)(\text{\lstinline{client}} \mid \text{\lstinline{server}})$ where
\begin{align*}
  \text{\lstinline{client}}  &= \proc{\context{
  (\recv{(cc \colon \DC \Cast{q} {?{\DYN}.\DC})})
        \colon \DYN \times_\lin \DC \Cast{} \DYN}}
  \\
  \text{\lstinline{server}} &= \proc{\context[F]{
(\send{ (42 : \intk\Cast{}\DYN)}{
  (cs \colon \DC \Cast{p} {!{\DYN}.\DC})})
                     \colon \DC \Cast{} \DYN}}
\end{align*}
for some contexts $E$ and $F$.
The channel ends $cc : \DC$ and $cs:\DC$ are the
two ends of the channel created in line~\ref{line:new}. Fortunately, the two
processes use the channel consistently as the cast target ${?{\DYN}.\DC}$ on one end
 is dual to the cast target ${!{\DYN}.\DC}$ at the
other end. Hence, {\GGV} has a reduction that drops the casts at both
ends in this situation, and retypes the ends to $cc : {?{\DYN}.\DC}$
and $cs : {!{\DYN}.\DC}$, respectively.
\[
  \proc{\context{
  (\recv{cc})
        \colon \DYN \times_\lin \DC \Cast{} \DYN}}
~~\mid~~
\proc{\context[F]{
(\send{ (42 : \intk\Cast{}\DYN)}{
  cs})
                     \colon \DC \Cast{} \DYN}}
\]
Implementing this reduction requires communication between the two
processes to check the cast targets for consistency. While our formal
presentation abstracts over this implementation issue, we observe that
a single \emph{asynchronous} message exchange is sufficient: Each
cast first sends its target type and then receives the target type of
the cast at the other end. Then both processes check locally whether the target
types are duals of one another. If they are, then both processes continue; otherwise they allocate blame. As 
both ends perform the same comparison, the outcome is the same
in both processes.

\section{GV and Gradual GV}
\label{sec:language}

\subsection{\GV}
\label{sec:gv}

We begin by discussing a language \GV{} with session types but without
gradual types.  The language is inspired by both the Gay and Vasconcelos'
functional session type calculus \cite{Gay-Vasconcelos-2010} and
Wadler's `good variant' of the language
\cite{Wadler-2012,Wadler-2014}.  A main difference from the former
is the introduction of communication primitives and session types to
close a session explicitly.  Unlike the latter,
 types are ``stratified'' into two levels---sessions types are just a subgrammer of types---and deadlock freedom is not guaranteed.

\subsubsection{Types and subtyping}

\begin{figure}[t]
  \begin{align*}
    &\text{Multiplicities}&
    \Multm,\Multn \grmeq& \lin \grmor \un
    \\
    & \text{Types}&
    T,U \grmeq& \unitk
    \grmor S 
    \grmor T \rightarrow_\Multm U 
    \grmor T \tcpair_\Multm U
    \\
    & \text{Session types}&
    S,R \grmeq& \pl{T}S
    \grmor \qu{T}S
    \grmor \oplus\br{l_i\colon S_i}_{i\in I}
    \grmor \with\br{l_i\colon S_i}_{i\in I}
    \grmor \Endpl
    \grmor \Endqu
  \end{align*}
  \spacer

  Duality \hfill \fbox{$\dual{S}=R$}
  \begin{align*}
    \dual{\pl{T}S} &= \qu{T}\dual{S} &
    \dual{\oplus\br{l_i\colon S_i}_{i \in I}} &= \with\br{l_i\colon \dual S_i}_{i \in I} &
    \dual\Endpl &= \Endqu \\
    \dual{\qu{T}S} &= \pl{T}\dual S &
    \dual{\with\br{l_i\colon S_i}}_{i \in I} &= \oplus\br{l_i\colon \dual S_i}_{i \in I} &
    \dual\Endqu &= \Endpl
  \end{align*}
  \spacer

  Multiplicity ordering \hfill \fbox{$m \Multsub n$}
  \begin{gather*}
      \un \Multsub \un
    \gap
      \un \Multsub \lin
    \gap
      \lin \Multsub \lin
  \end{gather*}
  \spacer

  Multiplicity of a type \hfill \fbox{$m(T) \quad \super{n}{T}$}
  \begin{gather*}
      \un(\unit)
    \gap
      \lin(S)
    \gap
      m(T \tcpair_m U)
    \gap
      m(T \to_m U)
    \gap
      \frac
      {m(T) \quad m \Multsub n}
      {\super{n}{T}}
  \end{gather*}
  \spacer

  Subtyping\hfill \fbox{$T \SubGV U$}
  \begin{gather*}
      \unit \SubGV \unit
    \gap
      \frac
      {T' \SubGV T \quad U \SubGV U' \quad \Multm \Multsub \Multn}
      {T \to_\Multm U \SubGV T' \to_\Multn U'}
    \gap
      \frac
      {T \SubGV T' \quad U \SubGV U' \quad \Multm \Multsub \Multn}
      {T \tcpair_\Multm U \SubGV T' \tcpair_\Multn U'}
    \nextline
      \frac
      {T' \SubGV T \quad S \SubGV S' }
      {\pl{T}S \SubGV \pl{T'}S'}
    \gap
      \frac
      {T \SubGV T' \quad S \SubGV S'}
      {\qu{T}S \SubGV  \qu{T'}S'}
    \gap
      \frac
      {J \subseteq I \quad (S_j \SubGV R_j)_{j \in J}}
      {\oplus\br{l_i \colon S_i}_{i\in I} \SubGV \oplus\br{l_j \colon R_j}_{j\in J}}
    \nextline
      \frac
      {I \subseteq J \quad (S_i \SubGV R_i)_{i \in I}}
      {\with\br{l_i \colon S_i}_{i\in I} \SubGV \with\br{l_j \colon R_j}_{i\in J}}
    \gap
      \Endpl \SubGV \Endpl
    \gap
      \Endqu \SubGV \Endqu
  \end{gather*}

  \hcaption{Types and subtyping in \GV.}
  \label{fig:gv-types}
\end{figure}

\begin{figure}[t]

  \begin{align*}
    & \text{Names} &
    \name \grmeq & x
    \grmor c
    \\
    & \text{Expressions} &
    e,f \grmeq & \name
    \grmor ()
    \grmor \lambda_m x.e
    \grmor e\,f 
    \grmor (e,f)_m
    \grmor \letin{x,y} e f
    \grmor \fork{e}
    \\ 
    &&
    \grmor & \new 
    \grmor \send{e}{f}
    \grmor \recv{e}
    \grmor \select{l}{e}
    \grmor \case{e}{\br{\casebranch{l_i}{x_i}{e_i}}_{i \in I}}
    \\ 
    &&
    \grmor & \close{e}
    \grmor \wait{e}
    \\
    & \text{Processes} &
    P,Q \grmeq & \proc e
    \grmor (P \PAR Q)
    \grmor (\nu \Chanc,\Chand)P
    \\
    & \text{Type environments} &
    \Gamma,\Delta \grmeq& \cdot
    \grmor \Gamma, \name \colon T
  \end{align*}
  \spacer

  Environment splitting \hfill \fbox{$\Gamma = \Gamma_1 \circ \Gamma_2$}
  \begin{gather*}
      \cdot = \cdot \circ \cdot
    \quad
      \frac
      {\Gamma = \Gamma_1 \circ \Gamma_2 \quad \un(T)}
      {\Gamma, \name \colon T = (\Gamma_1,\name\colon T) \circ (\Gamma_2,\name\colon T)}
    \quad
    \frac
      {\Gamma = \Gamma_1 \circ \Gamma_2 \quad \lin(T)}
      {\Gamma, \name \colon T = (\Gamma_1,\name\colon T) \circ \Gamma_2}
    \quad
      \frac
      {\Gamma = \Gamma_1 \circ \Gamma_2 \quad \lin(T)}
      {\Gamma, \name \colon T = \Gamma_1 \circ (\Gamma_2,\name\colon T)}
  \end{gather*}
  \spacer

  Typing expressions \hfill \fbox{$\Gamma \vdash e: T$}
  \begin{gather*}
      \frac
      {\un(\Gamma)}
      {\Gamma,\name \colon T \vdash \name: T}
    \gap 
      \frac
      {\un(\Gamma)}
      {\Gamma \vdash (): \unit}
    \gap
      \frac
      {\Gamma, x \colon T \vdash e : U \quad \super{\Multm}{\Gamma}}
      {\Gamma \vdash \lambda_\Multm x.e : T \to_\Multm U}
    \gap
      \frac
      {\Gamma \vdash e : T \to_\Multm U \quad \Delta \vdash f : T}
      {\Gamma \circ  \Delta \vdash e\,f : U}
    \nextline
      \frac
      {\Gamma \vdash e : T \quad \Delta \vdash f : U \quad \super{\Multm}{T} \quad \super{\Multm}{U}}
      {\Gamma\circ \Delta \vdash (e,f)_\Multm : T \times_\Multm U}
    \gap
      \frac
      {\Gamma \vdash e : T_1 \times_\Multm T_2 \quad \Delta, x\colon T_1, y\colon T_2 \vdash f: U}
      {\Gamma \circ \Delta \vdash \letin{x,y}{e}{f} : U}
    \nextline
      \frac
      {\Gamma \vdash e : \unit}
      {\Gamma \vdash \fork{e} : \unit}
    \gap
      \frac
      {\un(\Gamma)}
      {\Gamma \vdash \new : S\times_\lin \dual{S}}
    \gap
      \frac
      {\Gamma\vdash e : T \quad \Delta\vdash f : \pl{T}S}
      {\Gamma \circ \Delta \vdash \send{e}{f} : S}
    \gap
      \frac
      {\Gamma\vdash e : \qu{T}S}
      {\Gamma \vdash \recv{e} :  T \tcpair_\lin S}
    \nextline
      \frac
      {\Gamma \vdash e : \oplus\{l_i\colon S_i\}_{i\in I} \quad j\in I}
      {\Gamma \vdash \select {l_j} e : S_j}
    \gap
      \frac
      {\Gamma \vdash e : \&\{l_i\colon S_i\}_{i\in I} \quad (\Delta,x_i:S_i \vdash e_i : T) _{i\in I}}
      {\Gamma \circ \Delta \vdash \case{e}{\br{\casebranch{l_i}{x_i}{ e_i}}_{i\in I}} : T}
    \nextline
      \frac
      {\Gamma \vdash e : \Endpl}
      {\Gamma \vdash \close{e} :  \unit}
    \gap
      \frac
      {\Gamma \vdash e : \Endqu}
      {\Gamma \vdash \wait{e} :  \unit}
    \gap
      \frac
      {\Gamma \vdash e : T \quad T \Sub U}
      {\Gamma \vdash e : U}
  \end{gather*}
  \spacer

  Typing processes \hfill \fbox{$\Gamma \vdash P$}
  \begin{gather*}
    \frac
      {\Gamma \vdash e : T \quad \un (T)}
      {\Gamma \vdash \proc{e}}
    \gap
      \frac
      {\Gamma \vdash P \quad \Delta \vdash Q}
      {\Gamma \circ \Delta \vdash P \PAR Q}
    \gap
      \frac
      {\Gamma, \Chanc \colon S, \Chand \colon \dual{S} \vdash P}
      {\Gamma \vdash (\nu \Chanc,\Chand)P}
  \end{gather*}

  \hcaption{Expressions, processes, and typing in GV.}
  \label{fig:gv-typing}
\end{figure}

\newcommand{\grmorevalctx}{\,\mid}

\begin{figure}[t]
  \begin{align*}
    & \text{Values} &
    v,w \grmeq & ()
    \grmor \lambda_m x.e
    \grmor (v,w)_m
    \grmor \Chanc
    \\
    & \text{Eval contexts} &
    E,F \grmeq& [\,]
    \grmorevalctx E\,e
    \grmorevalctx v\,E 
    \grmorevalctx (E,e)_\Multm
    \grmorevalctx (v,E)_\Multm
    \grmorevalctx \letin{x,y}{E}{e}
    \grmorevalctx \send{E}{e}
    \grmorevalctx \send{v}{E}
    \\ &&
    \grmor & \recv{E}
    \grmor \select{l}{E }
    \grmor \case{E}{\br{\casebranch{l_i}{x_i}{ e_i}}_{i \in I}}
    \grmor \close{E}
    \grmor \wait{E}
  \end{align*}
  \spacer

Expression reduction \hfill \fbox{$e \reduces f$}
  \begin{align*}
    (\lambda_\Multm x.e)v &\reduces e\subs{v}{x} \\
    \letin{x,y}{(v,w)_\Multm}{e} &\reduces e\subs{v}{x}\subs{w}{y}
  \end{align*}
  \spacer

Structural congruence \hfill \fbox{$P \equiv Q$} 
  \begin{align*}
    P \PAR Q &\equiv Q \PAR P
    &
    P \PAR (Q \PAR P') &\equiv (P \PAR Q) \PAR P'
    &
      P \PAR \proc{()} &\equiv P
    &  (\nu \Chanc,\Chand)P &\equiv (\nu \Chand,\Chanc)P
  \end{align*}
  \[
  \begin{array}{r@{\;}ll}
      ((\nu \Chanc,\Chand)P) \PAR Q &\equiv (\nu \Chanc,\Chand)(P \PAR Q)
      &
      \text{if $\{c,d\}\cap \fn{Q} = \emptyset$}
      \\[1ex]
      (\nu \Chanc,\Chand)(\nu \Chanc',\Chand')P &\equiv (\nu \Chanc',\Chand')(\nu \Chanc,\Chand)P
      &
           \text{if $\{c,d\} \cap \{c',d'\} = \emptyset$}
  \end{array}
  \]
  \spacer

Process reduction  \hfill \fbox{$P \reduces Q$}
  \begin{align*}
    \proc{\contextE{\fork e}}
      &\reduces \proc{\contextE{()}} \PAR \proc{e} \\
    \proc{\contextE{\newk}}
      &\reduces (\nu \Chanc,\Chand)\proc{\contextE{(\Chanc,\Chand)_\lin}} \\
    (\nu \Chanc,\Chand)(\proc{\contextE{\send{v}{\Chanc}}} \PAR \proc{\contextF{\recv{\Chand}}})
      &\reduces (\nu \Chanc,\Chand)(\proc{\contextE{\Chanc}} \PAR \proc{\contextF{(v,\Chand)_\lin}}) \\
    (\nu \Chanc,\Chand) (\proc{\contextE{\select{l_j}{\Chanc}}} \PAR
      \proc{\contextF{\case{\Chand}{\br{\casebranch{l_i}{x_i}{e_i}}_{i \in I}}}})
      &\reduces (\nu \Chanc,\Chand)(\proc{\contextE{\Chanc}} \PAR \proc{\contextF{e_j[\Chand/x_j]}})
      &\text{if $j \in I$} \\
    (\nu \Chanc,\Chand)(\proc{\contextE{\close{\Chanc}}} \PAR \proc{\contextF{\wait{\Chand}}})
      &\reduces \proc{\contextE{()}} \PAR \proc{\contextF{()}}
  \end{align*}
  \begin{gather*}
    \frac{P \reduces P'}{P\PAR Q \reduces P'\PAR Q}
    \quad
    \frac{P \reduces Q}{(\nu \Chanc,\Chand)P \reduces (\nu \Chanc,\Chand)Q}
    \quad
    % \frac{P \reduces Q}{(\nu \alpha)P \reduces (\nu \alpha)Q}
    % \qquad
    \frac{P' \equiv P \quad P \reduces Q \quad Q \equiv Q'}{P' \reduces Q'}
    \quad
    \frac{e \reduces f}{\proc{\contextE{e}} \reduces \proc{\contextE{f}}}
  \end{gather*}
  %  (In rules $\newk$/$\sendk$/$\selectk$, context $E$/$E$/$F$ does not bind~$c,d$/$c$/$d$.)

 \hcaption{Reduction in \GV.}
 \label{fig:gv-reductions}
\end{figure}

Figure~\ref{fig:gv-types} summarises types of \GV.
Let $\Multm,\Multn$ range over \emph{multiplicities} for types whose
use is either unrestricted, $\un$, or must be linear, $\lin$.  

Let $T,U$ range over types, which include: unit type, $\unit$;
unrestricted and linear function types, $T \to_\Multm U$;
unrestricted and linear product types, $T \times_\Multm U$;
and session types.  One might also wish to include booleans
or base types, but we omit these as they can be dealt with analogously to $\unit$.

Let $l$ range over labels used for selection and case choices.
Let $S,R$ range over session types that describe communication
protocols for channel endpoints, which include: send $\pl{T}S$, to
send a value of type $T$ and then behave as $S$; receive $\qu{T}S$, to
receive a value of type $T$ and then behave as $S$, select
$\oplus\br{l_i \colon S_i}_{i \in I}$, to send one of the labels $l_i$
and then behave as $S_i$; case $\with\br{l_i \colon S_i}_{i \in I}$ to receive
any of the labels $l_i$ and then behave as $S_i$; close $\Endpl$, to
close a channel endpoint; and wait $\Endqu$, to wait for the other end
of the channel to close.
In $\oplus\br{l_i \colon S_i}_{i \in I}$ and
 $\with\br{l_i \colon S_i}_{i \in I}$, the label set must be non-empty.
We will call the session type that describes
the behaviour after send, receive, select, or case the \emph{residual}.

We define the usual notion of the dual of a
session type $S$, written as $\dual{S}$.  Send is dual to receive,
select is dual to case, and close is dual to wait.
Duality is an involution, so that $\dual{\dual{S}} = S$.

Multiplicities are ordered by $\un \Sub \lin$, indicating that an
unrestricted value may be used where a linear value is expected, but
not conversely.  The unit type is unrestricted, session types are
linear, while function types $T \to_\Multm U$ and product types
$T \times_\Multm U$ are unrestricted or linear depending on the
multiplicity $\Multm$ that decorates the type constructor.
To ensure that linear objects are used exactly once our type system
imposes the invariant that unrestricted data structures do not contain
linear data structures. As an example, type
$\unitk \tcpair_\un \Endpl$ cannot be introduced in any derivation.
We also write $\super{n}{T}$ if $m(T)$ holds for some $m$ such that
$m \Sub n$, thus $\super{\un}{T}$ holds only if $\un(T)$, while
$\super{\lin}{T}$ holds if either $\lin(T)$ or $\un(T)$, and hence
holds for any type.

We define subtyping as usual for functional-program like systems \cite{Gay-Vasconcelos-2010}.  Function types are 
contravariant in their domain, covariant in their range, and covariant in their multiplicity,
and send types are contravariant in the value sent and
covariant in the residual session type.  All other types
and session types are covariant in all components.
Width subtyping resembles record subtyping for select,
and variant subtyping for case.  That is, on an endpoint
where one may select among labels with an index in $I$
one may instead select among labels with indexes in $J$,
so long as $J \subseteq I$, while on an endpoint where
one must be able to receive any label with an index in $I$
one may instead receive any label with an index in $J$,
so long as $I \subseteq J$. 
(Beware that the subtyping on endpoints is exactly the reverse for process-calculus like systems, such as Wadler's CP \cite{Wadler-2012,Wadler-2014}!)

Subtyping is reflexive, transitive, and antisymmetric.
Duality inverts subtyping, in that $S \Sub R$ if and
only if $\dual{R} \Sub \dual{S}$.

\subsubsection{Expressions, processes, and typing}

Expressions, processes, and typing for \GV\
are summarised in Figure~\ref{fig:gv-typing}.
We let $x,y$ range over variables,
$c,d$ range over channel endpoints,
and $z$ range over names, which are
either variables or channel endpoints.

We let $e,f$ range over expressions, which include names, unit value,
function abstraction and application, pair creation and destruction,
fork a process, create a new pair of channel endpoints, send, receive,
select, case, close, and wait.  Function abstraction and pair creation
are labelled with the multiplicity of the value created.  We sometimes
abbreviate expressions of the form $(\lambda_\lin x.e)f$ to
$\letin xef$, as usual.  A \GV{} program is always given as an
expression, but as it executes it may fork new processes.

We let $P,Q$ range over processes, which
include expressions, parallel composition,
and a binder that introduces a pair of channel endpoints.
The initial process will consist of a single expression,
corresponding to a given \GV{} program.

The \emph{bindings} in the language are as follows:
variable $x$ is bound in subexpression $e$ of $\lambda_\Multm x.e$,
variables $x,y$ are bound in subexpression $f$ of $\letin{x,y} e f$,
variables $x_i$ are are bound in subexpressions $e_i$ of
$\case{e}{\br{\casebranch{l_i}{x_i}{ e_i}}_{i \in I}}$,
channel endpoints $\Chanc,\Chand$ are bound in subprocess $P$ of
$(\nu \Chanc,\Chand)P$.
We assume that $\Chanc$ and $\Chand$ in $(\nu \Chanc,\Chand)P$ are different.
The notions of free and bound names/variables as
well that substitution %(of $x$ by $v$ in $e$, notation $e\subs vx$)
are defined accordingly.
The set of the free names in $P$ is denoted by $\fn{P}$.
We follow Barendregt's variable convention, whereby all names
in binding occurrences in any mathematical context are pairwise
distinct and distinct from the free names~\citep{barendregt:lambda-calculus}.

We let $\Gamma,\Delta$ range over environments, which are used for
typing. An environment consists of zero or more associations of
names with types.  Environment splitting $\Gamma = \Gamma_1 \circ
\Gamma_2$ is standard.  It breaks an environment $\Gamma$ for an
expression or process into environments $\Gamma_1$ and $\Gamma_2$ for
its components; a name of unrestricted type may be used in both
environments, while a name of linear type must be used in one
environment or the other but not both.
We write $m(\Gamma)$ if $m(T)$ holds for each $T$ in $\Gamma$,
and similarly for $\super{m}{\Gamma}$.

Write $\Gamma \vdash e : T$ if under environment $\Gamma$ expression $e$ has type $T$.
The typing rules for expressions are standard.
In the rules for names, unit, and $\new$
the remaining environment must be unrestricted,
to enforce the invariant that linear variables are used exactly once.
A function abstraction that is unrestricted must have
only unrestricted variables bound in its closure, and a pair that is
unrestricted may only contain components that are unrestricted.
Thus, it is never possible to construct a pair of type, e.g., $S \tcpair_\un T$,
which contains a linear type $S$ under the unrestricted pair type constructor $\tcpair_\un$, even though such a type is syntactically allowed for simplicity. 
The rules for send, receive, select, case, close, and wait
match the corresponding session types.
For example, the following type judgment
\[
  o \colon \intk, c \colon !\intk.\Endpl \vdash \close{(\send{o}{c})} : \unitk
\]
can be derived.
The typing system supports
subsumption: if $e$ has type $T$ and $T$ is a subtype of $U$ then
$e$ also has type $U$.

Write $\Gamma \vdash P$ if under environment $\Gamma$ process $P$ is
well typed.  The typing rules for processes are also standard.  If
expression $e$ has unrestricted type $T$ then process $\proc{e}$ is well-typed.  If
processes $P$ and $Q$ are well-typed, then so is process $P \PAR Q$,
where the environment of the latter can be split to yield the
environments for the former.  And if process $P$ is well-typed under
an environment that includes channel endpoints $c$ and $d$ with session
types $S$ and $\dual{S}$, then process $(\nu c,d)P$ is well-typed
under the same environment without $c$ and $d$.

\subsubsection{Reduction}

Values, evaluation contexts, reduction for expressions, structural
congruence, and reduction for processes for \GV{} are summarised in
Figure~\ref{fig:gv-reductions}.

Let $v,w$ range over values, which include
unit, function abstractions, pairs of values, and channel endpoints.
Let $E,F$ range over evaluation contexts,
which are standard.

Write $e \reduces f$ to indicate that expression $e$ reduces to expression $f$.
Reduction is standard, consisting of beta reduction for functions
and pairs.

Write $P \equiv Q$ for structural congruence of processes.
It is standard, with composition being commutative and
associative. A process returning the unit is the
identity of parallel composition, so $P \PAR \proc{()} \equiv P$.
The order in which the
endpoints are written in a $\nu$-binder is irrelevant. Distinct prefixes commute,
and satisfy scope extrusion.
The Barendregt convention ensures that $c,d$ are not free in $Q$ in the rule for scope extrusion.  Similarly for the rule to swap prefixes.

Write $P \reduces Q$ if process $P$ reduces to process $Q$.
Evaluating $\fork{e}$ returns $()$ and creates a new process $\proc{e}$.
Evaluating $\newk$ introduces a new binder $(\nu c,d)$ and returns
a pair $(c,d)_\lin$ of channel endpoints.  Evaluating $\send{v}{c}$
on one endpoint of a channel and $\recv{d}$ on the other,
causes the send to return $c$ and the receive to return $(v,d)_\lin$.
Similarly for select on one endpoint of a channel and case on the other,
or close on one endpoint of a channel and wait on the other.

Process reduction is a congruence with regard to parallel composition
and binding for channel endpoints, it is closed under structural
congruence, and supports expression reduction under evaluation
contexts.

\subsection{\GGV}
\label{sec:ggv}

We now introduce \GGV.  Following standard frameworks of gradual typing~\cite{Siek-Taha-2006,Siek-et-al-2015-criteria},
\GGV{} consists of two sublanguages: an external language
\GGVe, in which source programs are written, and an internal language
\GGVi, to which \GGVe{} is elaborated by cast-inserting translation to
make necessary run-time checks explicit.  The operational semantics of
a program is given as reduction of processes in \GGVi.  We first
introduce \GGVi{} by outlining its differences to \GV{}
(Sections~\ref{sec:GGV-types}--\ref{sec:GGVi-reduction}).  Next, we
introduce the syntax of \GGVe, which has only expressions, because it
is the language in which source programs are written, its type system,
and cast-inserting translation from \GGVe{} to \GGVi{}
(Sections~\ref{sec:GGVe}--\ref{sec:cast-insertion}).  Finally, we discuss how an untyped variant
of \GV{} can be embedded into \GGVi{} (Section~\ref{sec:embedding}).

\subsubsection{Types and subtyping}
\label{sec:GGV-types}

\begin{figure}[t]
  \begin{align*}
    & \text{Types}&
    T,U &\grmeq \mygray{\unitk
    \grmor S 
    \grmor T \rightarrow_\Multm U 
    \grmor T \tcpair_\Multm U}
    \grmor \DYN 
    \\
    & \text{Session types}&
    S,R &\grmeq \mygray{\pl{T}S
    \grmor \qu{T}S
    \grmor \oplus\br{l_i\colon S_i}_{i\in I}
    \grmor \with\br{l_i\colon S_i}_{i\in I}
    \grmor \Endpl
    \grmor \Endqu}
    \grmor \DC
    \\
    &\text{Ground types}&
    \GroundT,\GroundU &\grmeq \unit
    \grmor \DC
    \grmor \DYN \to_\Multm \DYN 
    \grmor \DYN \tcpair_\Multm \DYN
    \\
    &\text{Ground session types}&
    \GroundS,\GroundR &\grmeq \pl{\DYN}\DC
    \grmor \qu{\DYN}\DC
    \grmor \oplus\br{l_i \colon \DC}_{i \in I}
    \grmor \with\br{l_i \colon \DC}_{i \in I}
    \grmor \Endpl
    \grmor \Endqu
  \end{align*}
  \spacer

  \begin{minipage}[t]{0.45\textwidth}
  Duality \hfill \fbox{$\dual{S}=R$}
  \[
  \dual{\DC} = \DC
  \]
  \end{minipage}
  \hfill
  \begin{minipage}[t]{0.45\textwidth}
  Multiplicity of a type \hfill \fbox{$m(T)$}
  \[
  \un(\DYN) \gap \lin(\DC)
  \]
  \end{minipage}

  \spacer

  %% Duality \hfill \fbox{$\dual{S}=R$}
  %% \vspace{-1ex}
  %% \begin{center}
  %% \mygray{
  %% \begin{minipage}[c]{0.7\textwidth}
  %%   \begin{align*}
  %%     \dual{\pl{T}S} &= \qu{T}\dual{S} &
  %%     \dual{\oplus\br{l_i\colon S_i}_{i \in I}} &= \with\br{l_i\colon \dual S_i}_{i \in I} &
  %%     \dual\Endpl &= \Endqu \\
  %%     \dual{\qu{T}S} &= \pl{T}\dual S &
  %%     \dual{\with\br{l_i\colon S_i}}_{i \in I} &= \oplus\br{l_i\colon \dual S_i}_{i \in I} &
  %%     \dual\Endqu &= \Endpl
  %%   \end{align*}
  %% \end{minipage}
  %% }
  %% \begin{minipage}[c]{0.2\textwidth}
  %%   \begin{align*}
  %%      \dual{\DC} &= \DC
  %%   \end{align*}
  %% \end{minipage}
  %% \end{center}
  %% \spacer

  %% Multiplicity of a type \hfill \fbox{$m(T)$}
  %% \begin{gather*}
  %%   \mygray{
  %%     \un(\unit)
  %%   \gap
  %%     \lin(S)
  %%   \gap
  %%     m(T \tcpair_m U)
  %%   \gap
  %%     m(T \to_m U)
  %%   \gap
  %%     \frac
  %%     {m(T) \quad m \Multsub n}
  %%     {n(T)}
  %%   }
  %%   \gap
  %%     \un(\DYN)
  %%   \gap
  %%     \lin(\DC)
  %% \end{gather*}
  %% \spacer

Subtyping \hfill \fbox{$T \Sub U$}
  \begin{gather*}
    %% \mygray{
    %% \unit \Sub \unit
    %% \gap
    %% \frac{
    %%   T' \Sub T \quad
    %%   U \Sub U' \quad
    %%   m \Multsub n
    %% }{
    %%   T \to_m U \Sub T' \to_n U'
    %% }
    %% \gap
    %% \frac{
    %%   T \Sub T' \quad
    %%   U \Sub U' \quad
    %%   m \Multsub n
    %% }{
    %%   T \times_m U \Sub T' \times_n U'
    %% }
    %% }
    %% \nextline
    %% \mygray{
    %% \frac{
    %%   T \Sub T' \quad
    %%   S \Sub S'
    %% }{
    %%   \pl{T}S \Sub \pl{T'}S' 
    %% }
    %% \gap 
    %% \frac{
    %%   T' \Sub T \quad
    %%   S \Sub S'
    %% }{
    %%   \qu{T}S \Sub \qu{T'}S' 
    %% }
    %% \gap
    %% \frac{
    %%   I \subseteq J \quad
    %%   (S_i \Sub S_i')_{i \in I}
    %% }{
    %%   \oplus\br{l_i\colon S_i}_{i\in I} \Sub
    %%   \oplus\br{l_i\colon S'_i}_{i\in J}
    %% }
    %% \gap
    %% \frac{
    %%   J \subseteq I \quad
    %%   (S_i \Sub S_i')_{i \in I}
    %% }{
    %%   \with\br{l_i\colon S_i}_{i\in I} \Sub
    %%   \with\br{l_i\colon S'_i}_{i\in J}
    %% }
    %% }
    %% \nextline
    %% \mygray{
    %% \Endpl \Sub \Endpl
    %% \gap
    %% \Endqu \Sub \Endqu
    %% }
    %% \nextline
    \DYN \Sub \DYN
    \gap
    \DC \Sub \DC
    %% \gap
    %% \frac
    %%   {T \Sub \GroundT}
    %%   {T \Sub \DYN}
  \end{gather*}
  \spacer

  Consistent subtyping \hfill \fbox{$T \lesssim U$}
  \begin{gather*}
      \unit \lesssim \unit
    \gap
      \frac
      {T' \lesssim T \quad U \lesssim U' \quad \Multm \Multsub \Multn}
      {T \to_\Multm U \lesssim T' \to_\Multn U'}
    \gap
      \frac
      {T \lesssim T' \quad U \lesssim U' \quad \Multm \Multsub \Multn}
      {T \tcpair_\Multm U \lesssim T' \tcpair_\Multn U'}
    \gap
      \DYN \lesssim T
    \gap
      T \lesssim \DYN 
    \nextline
      \frac
      {T' \lesssim T \quad S \lesssim S' }
      {\pl{T}S \lesssim \pl{T'}S'}
    \gap
      \frac
      {T \lesssim T' \quad S \lesssim S'}
      {\qu{T}S \lesssim  \qu{T'}S'}
    \gap
      \frac
      {J \subseteq I \quad (S_j \lesssim S_j')_{j \in J}}
      {\oplus\br{l_i \colon S_i}_{i\in I} \lesssim \oplus\br{l_j \colon S'_j}_{j\in J}}
    \nextline
      \frac
      {I \subseteq J \quad (S_i \lesssim S_i')_{i\in I}}
      {\with\{l_i\colon S_i\}_{i\in I} \lesssim \with\{l_j\colon S'_j\}_{j\in J}}
    \gap
      \Endpl \lesssim \Endpl
    \gap
      \Endqu \lesssim \Endqu
    \gap 
      \DC \lesssim S 
    \gap 
      S \lesssim \DC
  \end{gather*}

  \hcaption{Types and subtyping in \GGV.}
  \label{fig:ggv-types}
\end{figure}

Following the usual approach to gradual types, we extend the grammar
of types with a \emph{dynamic} type (sometimes also called the
\emph{unknown} type), written $\DYN$.  Similarly, we extend session
types with the dynamic session type, written $\DC$.  The extended
grammar of types is given in Figure~\ref{fig:ggv-types}, where types
carried over from Figure~\ref{fig:gv-types} are typeset in gray.

As before, we let $T$, $U$ range over types and $S$, $R$ range over
session types.  We also distinguish a subset of types which we call
\emph{ground types}, ranged over by $\GroundT,\GroundU$, and a subset
of session types which we call \emph{ground session types}, ranged
over by $\GroundS,\GroundR$, consisting of all the type constructors
applied only to arguments which are either the dynamic type or the
dynamic session type, as appropriate.

We define $\DC$ to be self-dual: $\dual{\DC} = \DC$.
We define the multiplicity of the new types by setting
$\DYN$ to be $\un$ and $\DC$ to be $\lin$.  The remaining
definitions of multiplicity of types carries over
unchanged from Figure~\ref{fig:gv-types}.  Type $\DYN$
is labelled unrestricted although (as we will see below) it corresponds
to all possible types, both unrestricted and linear, and therefore
we will need to take special care when handling values of type $\DYN$
that correspond to values of a linear type.

Consistent subtyping is defined over types of \GGV{} also in
Figure~\ref{fig:ggv-types}.  It is identical to the
definition of subtyping from Figure~\ref{fig:gv-types},
with each occurrence of $\Sub$ replaced by $\lesssim$,
and with the addition of four rules for the new types
\begin{gather*}
  \DYN \lesssim T  \gap
  T \lesssim \DYN  \gap
  \DC \lesssim S   \gap
  S \lesssim \DC\ .
\end{gather*}
For example, we have
(a) $\oplus\br{l_1: \pl{\DYN}\DC, l_2: \qu{\DYN}\DC} \lesssim \oplus\br{l_1: \DC}$
and (b) $\with\br{l_1: \DC} \lesssim \with\br{l_1: \pl{\DYN}\DC, l_2: \qu{\DYN}\DC}$.
Consistent subtyping is reflexive, but neither symmetric nor transitive.
As with subtyping, we have $\dual{S} \lesssim R$ iff $\dual{R} \lesssim S$.
In \GGV, we will be permitted to attempt to cast a value of type
$T$ to a value of type $U$ exactly when $T \lesssim U$. A cast
may fail at run time: while a cast using (a)
will not fail, a cast using (b) may fail because an expression of
type $\with\br{l_1: \DC}$ may evaluate to a value of type, say,
$\with\br{l_1: \Endpl}$.

Two types are consistent, written $T \sim U$, if $T \lesssim U$ and
$U \lesssim T$.  Consistency is reflexive and symmetric but not
transitive.  The standard example of the failure of transitivity is
that for any function type we have $T \to_\Multm U \sim \DYN$ and for
any product type we have $\DYN \sim T' \times_\Multn U'$, but
$T \to_\Multm U \not\sim T' \times_\Multn U'$.
In the setting of session types one has for example $\qu TS \sim \DC$
and $\DC \sim \Endpl$, but $\qu TS \not\sim \Endpl$.

Subtyping $T \Sub U$ for \GGV\ essentially
carries over from \GV. Its
definition is exactly as in Figure~\ref{fig:gv-types},
with the addition of two rules that ensure subtyping is reflexive
for the dynamic type and the dynamic session type. In contrast to consistent subtyping, subtyping
$T \Sub U$ guarantees that we may always treat a value of the first type as if it belongs
to the second type without casting.

\subsubsection{Expressions, processes, and typing of \GGVi}
\label{sec:GGVi}

\begin{figure}[t]

  \begin{align*}
    & \text{Blame labels} &
    p, q & 
    \\
    & \text{References} &
    \Storea, \Storeb & 
    \\
    & \text{Names} &
    \name \grmeq & \cdots
    \grmor \Storea
    \\
    & \text{Expressions} &
    e,f \grmeq & \cdots
    \grmor e \colon T \Cast{p} U
    \\
    & \text{Processes} &
    P,Q \grmeq & \cdots
    \grmor (\nu a)P
    \grmor \Storea \mapsto w : \GroundT \Cast{p} \DYN
    \grmor \Storea \mapsto \lockedbl{p}
    \grmor \blameboth{p}{q}{\VarsX}
    \grmor \blamegc{p}{\VarsX}
  \end{align*}
  \spacer

  Typing expressions \hfill \fbox{$\Gamma \vdash e: T$}
  \begin{gather*}
      \frac
      {\Gamma \vdash e : T \quad T \lesssim U}
      {\Gamma \vdash (e \colon T \Cast{p} U) : U}
  \end{gather*}
  \spacer

  Typing processes \hfill \fbox{$\Gamma \vdash P$}
  \begin{gather*}
    \frac
      {\Gamma, \Storea \colon \DYN \vdash P}
      {\Gamma \vdash (\nu\Storea)P}
    \gap
    \frac
    {\Gamma \vdash \Storea \colon \DYN \quad
      \Delta \vdash w : \GroundT \quad \lin(\GroundT)}
      {\Gamma \circ \Delta \vdash \Storea \mapsto w : \GroundT \Cast p \DYN}
    \gap
    \frac
      {\Gamma \vdash \Storea \colon \DYN}
      {\Gamma \vdash \Storea \mapsto \lockedbl{p}}
    \nextline
      \frac
      {\flv{\Gamma} = \VarsX}
      {\Gamma \vdash \blameboth{p}{q}{\VarsX}}
    \gap
      \frac
      % {\un(\Gamma) \quad \lin(\Delta) \quad \flv{\Delta} = \VarsX}
      % {\Gamma \circ \Delta \vdash \blamegc{p}{\VarsX} : T}
      {\flv{\Gamma} = \VarsX}
      {\Gamma \vdash \blamegc{p}{\VarsX}}
  \end{gather*}

  \hcaption{Expressions, processes, and typing in \GGVi.}
  \label{fig:ggv-typing}
\end{figure}

Expressions, processes, and type rules of \GGVi{}
are summarised in Figure~\ref{fig:ggv-typing}.
The expressions of \GGVi{} are those of \GV, plus an additional
form for casts.  A cast is written
\begin{gather}\label{eq:1}
e : T \Cast{p} U
\end{gather}
where $e$ is an expression of type $T$, and $p,q$ range over blame
labels such as \(\ell_1, \ell_2, \ldots\).  For example, the following
term
\[
  \textit{SOC} = \lambda_\un o. \lambda_\un c. \close{((\send{o}{(c \colon \DC \Cast{\ell_1} !\DYN. \DC)}) \colon \DC \Cast{\ell_2} \Endpl)},
\]
which represents a simplified version of \lstinline{serveOpCast} in
Section~\ref{sec:motivation}, can be given type
$\DYN \to_\un \DC \to_\un \unitk$.% \footnote{The order of abstractions is important:

Blame labels carry a polarity, which is either positive or
negative. The complement operation, $\dual{p}$, takes a positive label
into a negative one and vice versa; complement is an involution, so
that $\dual{\dual{p}} = p$. By convention, we assume that all blame
labels in a source program are positive, but negative blame labels may
arise during evaluation of casts at a function type or a send type.  A
cast raises \emph{positive blame} if the fault lies with the
expression \emph{contained} in the cast (for instance, because it
returns an integer where a character is expected), while it raises
\emph{negative blame} if the fault lies with the context
\emph{containing} the cast (for instance, because it passes an
argument or sends a value that is an integer where a character is
expected).

In a valid cast $e : T \Cast{p} U$, the type $T$ must be a consistent subtype of
$U$ ($T \lesssim U$), the type of the entire expression.  If a cast in
a program fails, it evaluates to $\blameboth{p}{q} X$ or
$\blamegc{p} X$ (which, as we see later, are treated as processes)
where the blame label $p$ and $q$ indicate the root cause of the
failure (we will explain $X$ shortly).  If the cast in~(\ref{eq:1})
fails, it means that the value returned by $e$ has type $T$, but not
type $U$.  For example, let $e = 4711 : \intk \Cast {q} \DYN$,
$T= \DYN$, and $U=\boolk$. As $\DYN \lesssim \boolk$, the resulting
expression $(4711 : \intk \Cast {q} \DYN) : \DYN \Cast p \boolk$ is
well-typed.  However, at run time it raises blame by reducing to
$\blameboth {\dual q} p \emptyset$, which flags the error that $\intk$ is
not a subtype of $\boolk$: that is, $\intk \not\Sub \boolk$.

Blame is indicated by processes of the form
\begin{gather*}
  \blameboth{p}{q}{X}  \gap\mbox{or}\gap  \blamegc{p}{X}
\end{gather*}
where $p$ and $q$ are blame labels, and $X$ is a set of variables
of linear type.  As we will see, most instances that yield blame
involve two casts, hence the form with two blame labels, although
blame can arise for a single cast, hence the form with one blame
label.  The set $X$ records all linear variables in scope when
blame is raised, and is used to maintain the invariant that as a
program executes each variable of linear type appears linearly
(only once, or once in each branch of a case).
Discarding linear variables when raising blame would break the
invariant.  Blame corresponds to raising an exception, and
the list of linear variables corresponds to cleaning up after
linear resources when raising an exception (for instance, closing
an open file or channel).  In the typing rules, the notation
$\flv{\Gamma}$ refers the set of free variables of linear
type that appear in $\Gamma$.  We also write $\flv{E}$ and
$\flv{v}$ for the free linear variables appearing in an
evaluation context $E$ or a value $v$.  In a running program,
only free linear variables are channel endpoints, so $\flv{E}$
and $\flv{v}$ can be defined without type information.

The processes of \GGVi{} are those of \GV, plus three additional
forms for references to linear values (as well as blame, described above).  Recall that a value of type
$\DYN$ may contain a linear value, in which case
dynamic checking must ensure that it is used exactly once.
The mechanism for doing so is to allocate a reference to a linear value.
We let $\Storea,\Storeb$ range over references.
A reference is of
type $\DYN$, and contains a value $w$ of ground type
$\GroundT$, where $\GroundT$ is linear (either
$\DYN \to_\lin \DYN$ or $\DYN \times_\lin \DYN$ or
the dynamic session type $\DC$). References are
allocated by the binding form $(\nu \Storea)P$, and the value contained in store $\Storea$ is indicated by a process which is either of the form
\begin{gather*}
  \Storea \mapsto w : \GroundT \Cast{p} \DYN
  \gap\mbox{or}\gap
  \Storea \mapsto \lockedbl{p}
\end{gather*}
where $w$ is a value of type $\GroundT$ and $p$ is a blame label.
Bindings for references initially take the first form, but change to
the second form after the reference has been accessed once;
any subsequent attempt to access the reference a second time will
cause an error.

\subsubsection{Reduction}
\label{sec:GGVi-reduction}

\newcommand{\grmorvalues}{\,\mid}

\begin{figure}[t]

  \[
  \begin{array}{lllll}
    \mbox{Values}&
    v,w &\grmeq \cdots
    \grmorvalues v \colon \GroundT \Cast{p} \DYN     \grmorvalues v \colon \GroundS \Cast{p} \DC
      \grmorvalues v \colon T \to_\Multm U \Cast{p} T' \to_\Multn U'
      \grmorvalues v \colon S \Cast{p} R
      \grmorvalues \Storea
    \\
    && \qquad \mbox{where }\un (\GroundT), S\ne \DC,R\ne\DC
    \\
    \mbox{Eval contexts}&
    E,F &\grmeq \cdots
                    \grmor E \colon T \Cast{p} U 
  \end{array}
  \]
  % \begin{align*}
  %   &\text{\Values}&
  %   v,w \grmeq \cdots
  %                   &\grmorvalues v \colon \GroundT \Cast{p} \DYN & \text{if  $\un (\GroundT)$}  \\
  %   &&&    \grmorvalues v \colon \GroundS \Cast{p} \DC
  %   \\
  %   &&& \grmorvalues v \colon T \to_\Multm U \Cast{p} T' \to_\Multn U' & \mbox{if $T \to_\Multm
  %                                                                        U \lesssim T' \to_\Multn U'
  %                                                                        $} \\
  %   &&& \grmorvalues v \colon S \Cast{p} R & \mbox{if $S, R \neq \DC$ and $S \lesssim R$} \\
  %   &&& \grmorvalues \Storea
  %   \\
  %   &\text{Eval contexts}&
  %   E,F \grmeq \cdots
  %   &\grmor E \colon T \Cast{p} U
  % \end{align*}
  \spacer

  Expression reduction \boxedrel{e \reduces f}
  \[
  \begin{array}{rcll}
    v \colon \DYN \Cast{p} \DYN &\reduces& v \hspace{6cm}
    \\
    v \colon \DC \Cast{p} \DC &\reduces &v \\
    v \colon \unitk \Cast{p} \unitk &\reduces & v \\
    (v \colon T \to_\Multm U \Cast{p} T' \to_\Multn U')\,w &\reduces&
      (v\, (w \colon T' \Cast{\dual p} T)) \colon U \Cast{p} U' \\
    (v,w)_\Multm \colon T \tcpair_\Multm U \Cast{p} T' \tcpair_\Multn U' &\reduces&
      (v \colon T \Cast{p} T', w \colon U \Cast{p} U')_\Multn \\
    \send{v}{(w \colon \pl{T}S \Cast{p} \pl{T'}S')} &\reduces&
      (\send{(v \colon T' \Cast{\dual p} T)}{w}) \colon S \Cast{p} S' \\
    \recv (w \colon \qu{T}S \Cast{p} \qu{T'}S') &\reduces&
       (\recv w) \colon T \tcpair_\lin S \Cast{p} T' \tcpair_\lin S' \\
    \multicolumn{3}{l}{
    \select{l_k}{(w \colon \oplus\br{l_i\colon R_i}_{i\in I} \Cast{p}
    \oplus\br{l_j\colon S_j}_{j\in J})}} \\ & \reduces &
       (\select{l_k}{w}) \colon R_k \Cast{p} S_k &
       \rtext{if $k \in J$, $J \subseteq I$} \\
    \multicolumn{3}{l}{
    \casefst{(w \colon \with\br{l_i\colon R_i}_{i\in I} \Cast{p}
    \with\br{l_i\colon S_i}_{i\in J})}
    \casesnd{\br{\casebranch{l_j}{x_j}{e_j}}_{j \in J}}
    }
    \\
    \multicolumn{4}{r}{
    \reduces \quad \casefst{w}
    \casesnd{\br{\casebranch{l_i}{x_i}{\letin{x_i}{(x_i \colon R_i \Cast{p} S_i)}{e_i}}}_{i \in I}}
    \qquad \text{if $I \subseteq J$}}
\\
    \close{(v \colon \Endpl \Cast{p} \Endpl)} &\reduces& \close{v} \\
    \wait{(v \colon \Endqu \Cast{p} \Endqu)} &\reduces& \wait{v} \\
    v \colon T \Cast{p} \DYN &\reduces&
      (v \colon T \Cast{p} \GroundT) \colon \GroundT \Cast{p} \DYN&
      \rtext{if $T \neq \DYN$, $T \neq \GroundT$, $T \sim \GroundT$} \\
    v \colon \DYN \Cast{p} T &\reduces&
      (v \colon \DYN \Cast{p} \GroundT) \colon \GroundT \Cast{p} T&
      \rtext{if $T \neq \DYN$, $T \neq \GroundT$, $T \sim \GroundT$} \\
    v \colon S \Cast{p} \DC &\reduces&
      (v \colon S \Cast{p} \GroundS) \colon \GroundS \Cast{p} \DC &
      \rtext{if $S \ne \DC$, $S \ne \GroundS$, $S \sim \GroundS$} \\
    v \colon \DC \Cast{p} S &\reduces&
      (v \colon \DC \Cast{p} \GroundS ) \colon \GroundS \Cast{p} S &
      \rtext{if $S \ne \DC$, $S \ne \GroundS$, $S \sim \GroundS$}
  \end{array}
  \]

 \hcaption{Reduction in \GGVi, expressions.}
 \label{fig:ggv-reductions-expressions}
\end{figure}

\begin{figure}[t]

Structural congruence \hfill \fbox{$P \equiv Q$} 
\[
  \begin{array}{r@{\;}ll}
    ((\nu \Storea)P) \PAR Q &\equiv (\nu \Storea)(P \PAR Q)&
           \text{if $\{c,d\} \cap \fn{Q} = \emptyset$}
  \\[1ex]
    (\nu \Storea)(\nu \Storeb)P &\equiv (\nu \Storeb)(\nu \Storea)P&
           \text{if $a \neq b$}
  \\[1ex]
    (\nu c,d)(\nu \Storea)P &\equiv (\nu \Storea)(\nu c,d)P       &
           \text{if $a \neq c$ and $a \neq d$}
  \end{array}
    % &
    % (\nu \Storea)\Storea \mapsto \lockedbl{p} &\equiv \proc{()}
\]
  \spacer

  Process reduction  \hfill \fbox{$P \reduces Q$}
  \[
  \begin{array}{r@{\;}c@{\;}lr}
    \proc{ \contextE{v \colon \GroundT \Cast{p} \DYN} } &\reduces&
            (\nu \Storea)(\proc{\context{\Storea}} \PAR
    \Storea \mapsto \wrap[p]{v}{\GroundT}) \\
     &&& \rtext{if $\lin(\GroundT)$ and $E \neq \contextF{[\,]\colon \DYN \Cast{q} \GroundU}$}
    \\
    \proc{ \contextE{\Storea \colon \DYN \Cast{q} \GroundU} }
      \PAR \Storea \mapsto \wrap[p]{v}{\GroundT} &\reduces& \multicolumn{2}{@{}l}{
                \proc{ \contextE{(\wrap[p]{v}{\GroundT}) \colon \DYN \Cast{q} \GroundU} }
                \PAR \Storea\mapsto\lockedbl{p} }
    \\
    \proc{ \contextE{\Storea \colon \DYN \Cast{q} \GroundU} }
      \PAR \Storea \mapsto \lockedbl{p} &\reduces & \multicolumn{2}{@{}l}{
                \blameboth{\dual{p}}{q}{(\flv{E})}
                \PAR \Storea \mapsto \lockedbl{p}} \\
    (\nu \Storea) ( \Storea \mapsto \lockedbl{p} ) &\reduces&
      \proc{()} \\
    (\nu \Storea) (\Storea \mapsto \wrap[p]{w}{\GroundT}) &\reduces&
      \blamegc{\dual{p}}{(\flv{w})} &
\\
    \proc{\contextE{(v \colon \GroundT \Cast{p} \DYN) \colon \DYN \Cast{q} \GroundU}} &\reduces&
       \proc{\contextE{v}} &
       \rtext{if $\GroundT \Sub \GroundU$} \\
    \proc{\contextE{(v \colon \GroundT \Cast{p} \DYN) \colon \DYN \Cast{q} \GroundU}} &\reduces&
       \blameboth{\dual{p}}{q}{(\flv{E} \cup \flv{v})} &
       \rtext{if $\GroundT \not\Sub \GroundU$} \\
    \proc{ \contextE{(v \colon \GroundS \Cast{p} \DC)  \colon \DC \Cast{q} \GroundR} } &\reduces&
      \proc{ \contextE{v} } &
      \rtext{if $\GroundS \Sub  \GroundR$} \\
    \proc{ \contextE{(v \colon \GroundS \Cast{p} \DC)  \colon \DC \Cast{q} \GroundR} } &\reduces&
      \blameboth{\dual{p}}{q}{(\flv{E} \cup \flv{v})} &
      \text{if $\GroundS \not\Sub \GroundR$} \\
    (\nu c,d)( \proc{ \contextE{\Chanc \colon \DC \Cast{p} \GroundS} }
               \PAR \proc{ \contextF{\Chand \colon \DC \Cast{q} {\GroundR}} }) &\reduces&
      (\nu c,d)( \proc{ \contextE{\Chanc} } \PAR \proc{ \contextF{\Chand } })
      &\rtext{if $\dual{\GroundS} \Sub {\GroundR}$}
      %\label{redn:channel-cancel}
      \\
    (\nu c,d)( \proc{ \contextE{\Chanc \colon \DC \Cast{p} \GroundS} }
               \PAR \proc{ \contextF{\Chand \colon \DC \Cast{q} \GroundR} }) &\reduces&
      \blameboth{p}{q}{(\flv{E} \cup \flv{F} \cup \br{c, d})}&
      \text{if $\dual{\GroundS} \notSub \GroundR$}
  \end{array}
  \]
  \begin{gather*}
        \frac{P \reduces Q}{(\nu \Storea)P \reduces (\nu \Storea)Q}
  \end{gather*}
 \hcaption{Reduction in \GGVi, processes.}
 \label{fig:ggv-reductions-processes}
\end{figure}

Values, evaluation contexts, reductions for expressions, structural
congruence, and reductions for processes for \GGVi{} are summarised in
Figures~\ref{fig:ggv-reductions-expressions}
and~\ref{fig:ggv-reductions-processes}.

The values of \GGVi{} are those of \GV, plus five additional forms.  Values of dynamic type 
have the form either $v \colon \GroundT \Cast{p} \DYN$ as in other blame calculi, if $\GroundT$ is
unrestricted, or $\Storea$, which is a reference to a linear value, if the dynamic type wraps a linear
value. Additionally, there are values of dynamic session type which take the form
$ v \colon \GroundS \Cast{p} \DC $.

Following standard practice for blame
calculus, we take a cast of a value between function types
to be a value, and for similar reasons a cast from a session type
to a session type is a value unless one end of the cast is the dynamic session type:
\begin{gather*}
v \colon T \to_\Multm U \Cast{p} T' \to_\Multn U'
\gap\mbox{or}\gap
v \colon S \Cast{p} R
\end{gather*}
where % $T \to_\Multm U \lesssim T' \to_\Multn U'$  and $S \lesssim R$ with
$S, R \neq \DC$.

Additional reductions for expressions appear in
Figure~\ref{fig:ggv-reductions-expressions}.
Typical of blame calculus is the reduction for a
cast between function types, often called the \emph{wrap} rule:
\begin{equation*}
(v \colon T \to_\Multm U \Cast{p} T' \to_\Multn U') \, w
\reduces
(v \, (w \colon T' \Cast{\dual{p}} T)) \colon U \Cast{p} U'
\end{equation*}
The cast on the function decomposes into two casts, one on the domain
and one on the range.  The fact that subtyping (and consistent
subtyping) for function types is contravariant on the domain and
covariant on the range is reflected in the fact that the cast on the
domain is from $T'$ to $T$ and complements the blame label $\dual{p}$,
while the cast on the range is form $U$ to $U'$ and leaves the blame
label $p$ unchanged.  Casts for products follow a similar pattern,
though covariant on all components.

Reductions on session types follow the pattern of the reduction for a cast between send types:
\begin{equation*}
\send{v}{(w \colon \pl{T}S \Cast{p} \pl{T'}S')}
\reduces
(\send{(v \colon T' \Cast{\dual p} T)}{w}) \colon S \Cast{p} S'
\end{equation*}
The cast on the send decomposes into two casts, one on the value sent
and one on the residual session type.  The fact that subtyping (and
consistent subtyping) for send types is contravariant on the value
sent and covariant on the residual session type is reflected in the
fact that the cast on the value sent is from $T'$ to $T$ and
complements the blame label $\dual{p}$, while the cast on the residual
session type is from $S$ to $S'$ and leaves the blame label $p$
unchanged.  The casts for the remaining session types follow a similar
pattern, though covariant on all components.

Also typical of blame calculus, casts to the dynamic type
factor through a ground type,
\begin{gather*}
v \colon T \Cast{p} \DYN
\reduces
(v \colon T \Cast{p} \GroundT) \colon \GroundT \Cast{p} \DYN 
\end{gather*}
when $T \neq \DYN$, $T \neq \GroundT$, and $T \sim \GroundT$.
This factoring is unique because for every type $T$ such that $T \neq \DYN$
there is a unique ground type $\GroundT$ such that
$T \sim \GroundT$.  The additional condition $T \neq \GroundT$
ensures that the factoring is non-trivial and that reduction does not enter a loop.  Casts from
the dynamic type, and casts to and from the dynamic session type
are handled analogously.

Additional structural congruences and reductions for processes appear in
Figure~\ref{fig:ggv-reductions-processes}.
Like bindings for channel endpoints, bindings for references to linear values satisfy scope extrusion and reduction is a congruence with respect to them.

\begin{sloppypar}
The first five reduction rules for processes deal with 
references to linear values, ensuring that a value cast from a linear type to $\DYN$ is accessed exactly once.  As the only values of the dynamic type
are casts from a ground type, expressions of interest take the form
\begin{gather*}
v \colon \GroundT \Cast{p} \DYN
\end{gather*}
where $v$ is a value and $\GroundT$ is a linear ground type.  The
first rule introduces a reference, represented as a separate process
of the form $\Storea \mapsto v \colon \GroundT \Cast{p} \DYN$.
The context restriction $E \neq \contextF{[\,]\colon \DYN \Cast{q} \GroundU}$ ensures that a reference is only
introduced if the value is not immediately accessed; without
the restriction this rule would apply to a process of the form
$\proc{\contextE{(v \colon \GroundT \Cast{p} \DYN) \colon \DYN \Cast{q} \GroundU}}$,
to which the sixth or seventh rule should be applied.
Any attempt to
access the linear reference $\Storea$ must take the form
\begin{gather*}
\contextE{\Storea \colon \DYN \Cast{q} \GroundU}
\end{gather*}
where $E$ is an evaluation context and $\GroundU$ is a ground type
that may or may not be linear.  The second rule implements the first
access to a linear value by copying the value $v$ in place of the
reference $a$, and updating the reference process to
$\Storea \mapsto \lockedbl{p}$, indicating that the linear reference
has been accessed once.  The third rule implements any subsequent
attempt to access a linear value, which allocates blame to the two
casts involved, negative blame $\dual{p}$ from $\lockedbl{p}$, which
was a cast $v \colon \GroundT \Cast{p} \DYN$ before the first access,
and positive blame $q$ for the cast to access $a$, indicating that in
both cases blame is allocated to the side of the cast of type $\DYN$.
The blame term also contains $\flv{E}$, the set of free linear
variables that appear in the context $E$, which as mentioned earlier
is required to maintain the invariant on linear variables; all
occurrences of blame contain corresponding sets of linear variables,
which we will not mention further.  The final two rules indicate what
happens when all processes containing the reference finish execution.
If the linear reference is locked then it was accessed once, and the
reference may be deallocated as usual.  If the reference is not locked
then it was never accessed, and blame should be allocated to the
context of the original cast, which discarded the value rather than
using it linearly.  In practice, these rules would be implemented as
part of garbage collection.
\end{sloppypar}

The remaining six rules come in three pairs.
Typical of blame calculus is the first pair,
often called the \emph{collapse} and \emph{collide} rules:
\[
\begin{array}{rclr}
  \proc{\contextE{(v \colon \GroundT \Cast{p} \DYN) \colon \DYN \Cast{q} \GroundU}} &\reduces&
   \proc{\contextE{v}}&
   \text{if $\GroundT \Sub \GroundU$} \\
\proc{\contextE{(v \colon \GroundT \Cast{p} \DYN) \colon \DYN \Cast{q} \GroundU}} &\reduces&
   \blameboth{\dual{p}}{q}{(\flv{E} \cup \flv{v})}&
   \text{if $\GroundT \not\Sub \GroundU$}
\end{array}
\]
If the source type is a subtype of the target type, the casts collapse to the original value.
Types are preserved by subsumption: since $v$ has type $\GroundT$
and $\GroundT \Sub \GroundU$ then $v$ also has type $\GroundU$.
Conversely, if the source type is not a subtype of the target type, then the casts are
in collision and reduce to blame.  Blame is allocated to both of the casts involved, negative
blame $\dual{p}$ for the inner cast and positive blame $q$ for
the outer cast, indicating that in both cases blame is allocated
to the side of the cast of type $\DYN$.
Our choice to allocate blame to both casts
differs from the usual formulation of blame calculus, which only allocates blame to the
outer cast.
Allocating blame to only the outer cast is convenient if
one wishes to implement blame calculus by erasure to a dynamically
typed language, where injection of a value to the dynamic type is
represented by the value itself, that is, the erasure of $v \colon \GroundT \Cast{p} \DYN$ is
just taken to be the erasure of $v$ itself. However, this asymmetric
implementation is less appropriate in our situation.
For session types, a symmetric formulation is more appropriate, as we
will see shortly when we look at the interaction between casts and communication.

The next pair of rules transpose collapse and collide from types to session types.
The final pair of rules adapt collapse and collide to the case of communication
between two channel endpoints.  Here is the adapted collapse rule.
\begin{align*}
(\nu c,d)( \proc{ \contextE{\Chanc \colon \DC \Cast{p} \GroundS} }
           \PAR \proc{ \contextF{\Chand \colon \DC \Cast{q} {\GroundR}} }) &\reduces
  (\nu c,d)( \proc{ \contextE{\Chanc} } \PAR \proc{ \contextF{\Chand } })
  &\text{if $\dual{\GroundS} \Sub {\GroundR}$}
\end{align*}
The condition on this rule is symmetric, since $\dual{\GroundS} \Sub \GroundR$
if and only if $\dual{\GroundR} \Sub \GroundS$.
On the left-hand side of this rule $c,d$ both have session type $\DC$,
while on the right-hand side of the rule $c,d$ have session types $\GroundS,\dual{\GroundS}$
or $\dual{\GroundR},\GroundR$.
Again, types are preserved by subsumption, since if $c,d$ have session types
$\GroundS,\dual{\GroundS}$ and $\dual{\GroundS} \Sub \GroundR$ then
$c,d$ also have session types $\GroundS,\GroundR$, and similarly if
$c,d$ have session types $\dual{\GroundR},\GroundR$.
Analogously, the last rule adapts collide.

An alternative design might replace the final pair of rules
by a structural congruence that slides a cast from one endpoint
of a channel to the other:
\begin{align*}
(\nu c,d)( E[c \colon S \Cast{p} R] \mid F[d] ) &\equiv
(\nu c,d)( E[c] \mid F[d \colon \dual{R} \Cast{\dual{p}} \dual{S}] ).
\end{align*}
Setting $S$ to $\DC$ and $R$ to $\GroundS$, this congruence
can reduce the third collapse rule 
(on channel endpoints) to the second collapse rule
(on a nested pair of casts on session types).
However, even with this congruence the two collide rules are not
quite equivalent.  Our chosen formulation, though slightly longer,
is more symmetric and easier to implement.

Now we show a few examples of reduction, in which we abbreviate a nested cast
$(e \colon T_1 \Cast{p} T_2) \colon T_2 \Cast{q} T_3$ to
$e \colon T_1 \Cast{p} T_2 \Cast{q} T_3$ and use a sequential
composition $e_1; e_2$ with obvious typing and reduction rules.
First recall the term
$$ \textit{SOC} = \lambda_\un o. \lambda_\un c. \close{(\send{o}{(c
    \colon \DC \Cast{\ell_1} !\DYN. \DC)} \colon \DC \Cast{\ell_2}
  \Endpl)}$$ introduced above.  Given a channel endpoint
$d \colon !\intk.\Endpl$, the term
\[
  \textit{SOC}\; (42 \colon \intk \Cast{\ell_3} \DYN)\; (d \colon !\intk.\Endpl \Cast{\ell_4} \DC)
\]
reduces as follows:
\begin{align*}
  & \textit{SOC}\; (42 \colon \intk \Cast{\ell_3} \DYN)\; (d \colon !\intk.\Endpl \Cast{\ell_4} \DC) \\ \reduces 
  & (\lambda_\un c. \close{((\send{(42 \colon \intk \Cast{\ell_3} \DYN)}{(c \colon \DC \Cast{\ell_1} !\DYN. \DC)}) \colon \DC \Cast{\ell_2} \Endpl)})\; (d \colon !\intk.\Endpl \Cast{\ell_4} \DC) \\ \reduces 
  & (\lambda_\un c. \close{((\send{(42 \colon \intk \Cast{\ell_3} \DYN)}{(c \colon \DC \Cast{\ell_1} !\DYN. \DC)}) \colon \DC \Cast{\ell_2} \Endpl)})\\ & \hspace*{7.9cm} (d \colon !\intk.\Endpl \Cast{\ell_4} !\DYN.\DC \Cast{\ell_4} \DC) \\ \reduces
  & \close{((\send{(42 \colon \intk \Cast{\ell_3} \DYN)}{(d \colon !\intk.\Endpl \Cast{\ell_4} !\DYN.\DC \Cast{\ell_4} \DC \Cast{\ell_1} !\DYN. \DC)}) \colon \DC \Cast{\ell_2} \Endpl)}  \\ \reduces
  & \close{((\send{(42 \colon \intk \Cast{\ell_3} \DYN)}{(d \colon !\intk.\Endpl \Cast{\ell_4} !\DYN.\DC)}) \colon \DC \Cast{\ell_2} \Endpl)}  \\ \reduces
  & \close{((\send{(42 \colon \intk \Cast{\ell_3} \DYN \Cast{\dual \ell_4} \intk)}{d}) \colon \Endpl \Cast{\ell_4} \DC \Cast{\ell_2} \Endpl)} \\ \reduces
  & \close{((\send{42}{d}) \colon \Endpl \Cast{\ell_4} \DC \Cast{\ell_2} \Endpl)}.
\end{align*}
Thus, the process $$(\nu d, e)(\proc{\textit{SOC}\; (42 \colon \intk \Cast{\ell_3} \DYN)\; (d \colon !\intk.\Endpl \Cast{\ell_4} \DC)} \PAR \proc{\letin{x,y}{\recv e}{\wait{y}}})$$
reduces as follows:
\begin{align*}
  & (\nu d, e)(\proc{\textit{SOC}\; (42 \colon \intk \Cast{\ell_3} \DYN)\; (d \colon !\intk.\Endpl \Cast{\ell_4} \DC)} \PAR \proc{\letin{x,y}{\recv e}{\wait{y}}}) \\ \reduces^+
& (\nu d, e)(\proc{\close{((\send{42}{d}) \colon \Endpl \Cast{\ell_4} \DC \Cast{\ell_2} \Endpl)}} \PAR \proc{\letin{x,y}{\recv e}{\wait{y}}}) \\ \reduces
& (\nu d, e)(\proc{\close{(d \colon \Endpl \Cast{\ell_4} \DC \Cast{\ell_2} \Endpl)}} \PAR \proc{\letin{x,y}{(42, e)_\lin}{\wait{y}}}) \\ \reduces^+
& (\nu d, e)(\proc{\close{d}} \PAR \proc{\wait{e}}) \\ \reduces
& (\nu d, e)(\proc{()} \PAR \proc{()}).
\end{align*}
However, if $d$ is given type $!\intk.!\intk.\Endpl$, then
$\textit{SOC}\; (42 \colon \intk \Cast{\ell_3} \DYN)\; (d \colon !\intk.!\intk.\Endpl \Cast{\ell_4} \DC)$
is well typed but reduces to 
\[
  \close{((\send{42}{d}) \colon !\intk.\Endpl \Cast{\ell_4} \DC \Cast{\ell_2} \Endpl)}.
\]
Thus, the process
\[
  (\nu d, e)(\proc{\textit{SOC}\; (42 \colon \intk \Cast{\ell_3} \DYN)\; (d \colon !\intk.!\intk.\Endpl \Cast{\ell_4} \DC)} \PAR \proc{\letin{x,y}{\recv e}{\ldots}})
\]
reduces to
\[
 (\nu d, e)(\proc{\close{(d \colon !\intk.\Endpl \Cast{\ell_4} \DC \Cast{\ell_2} \Endpl)}} \PAR \proc{\letin{x,y}{(42, e)_\lin}{\ldots}})
\]
and then to
\[
 (\nu d, e)(\blameboth{\dual \ell_4}{\ell_2}{\{d\}} \PAR \proc{\letin{x,y}{(42, e)_\lin}{\ldots}}).
\]

We also show an example of dynamic linearity checking.  The function \textit{foo} below takes
an argument of type $\DYN$, cast it to $\Endpl$, and closes it:
\[
  \textit{foo} = \lambda_\un x.\close{(x \colon \DYN \Cast{\ell} \Endpl)}.
\]
Consider an application of \textit{foo} to a channel endpoint $c$ of type $\Endpl$.
It reduces as follows:
\begin{align*}
  & \proc{\textit{foo}\; (c \colon \Endpl \Cast{\ell'} \DYN)} \\ \reduces &
  \proc{\textit{foo}\; (c \colon \Endpl \Cast{\ell'} \DC \Cast{\ell'} \DYN)} \\ \reduces &
  (\nu a)(\proc{\textit{foo}\; a} \PAR a \mapsto c \colon \Endpl \Cast{\ell'} \DC \Cast{\ell'} \DYN)\\ \reduces &
  (\nu a)(\proc{\close{(a \colon \DYN \Cast{\ell} \Endpl)}} \PAR a \mapsto c \colon \Endpl \Cast{\ell'} \DC \Cast{\ell'} \DYN)\\ \reduces &
  (\nu a)(\proc{\close{(c \colon \Endpl \Cast{\ell'} \DC \Cast{\ell'} \DYN \Cast{\ell} \Endpl)}} \PAR a \mapsto \lockedbl{\ell'})\\ \reduces^+ &
  (\nu a)(\proc{\close{c}} \PAR a \mapsto \lockedbl{\ell'}) \\ \reduces &
  \proc{\close{c}}
\end{align*}
If the channel endpoint is passed to a function that uses the argument more than once,
blame will be raised.  Let \textit{bar} be
\(
  \lambda_\un x.\close{(x \colon \DYN \Cast{\ell} \Endpl)}; \close{(x \colon \DYN \Cast{\ell} \Endpl)}
\)
and observe that \(\textit{bar}\; (c \colon \Endpl \Cast{\ell'} \DYN)\)
reduces as follows:
\begin{align*}
  & \proc{\textit{bar}\; (c \colon \Endpl \Cast{\ell'} \DYN)} \\ \reduces &
  \proc{\textit{bar}\; (c \colon \Endpl \Cast{\ell'} \DC \Cast{\ell'} \DYN)} \\ \reduces &
  (\nu a)(\proc{\textit{bar}\; a} \PAR a \mapsto c \colon \Endpl \Cast{\ell'} \DC \Cast{\ell'} \DYN)\\ \reduces &
  (\nu a)(\proc{\close{(a \colon \DYN \Cast{\ell} \Endpl)}; \close{(a \colon \DYN \Cast{\ell} \Endpl)}} \PAR a \mapsto c \colon \Endpl \Cast{\ell'} \DC \Cast{\ell'} \DYN)\\ \reduces &
  (\nu a)(\proc{\close{(c \colon \Endpl \Cast{\ell'} \DC \Cast{\ell'} \DYN \Cast{\ell} \Endpl)}; \close{(a \colon \DYN \Cast{\ell} \Endpl)}} \PAR a \mapsto \lockedbl{\ell'})\\ \reduces^+ &
  (\nu a)(\proc{\close{c}; \close{(a \colon \DYN \Cast{\ell} \Endpl)}} \PAR a \mapsto \lockedbl{\ell'})
  % \\ \reduces^+ &
  % (\nu a)(\proc{\close{(a \colon \DYN \Cast{\ell} \Endpl)}} \PAR a \mapsto \lockedbl{\ell'}) \\ \reduces &
  % (\nu a)(\blameboth{\dual \ell'}{\ell}{\emptyset} \PAR a \mapsto \lockedbl{\ell'}).
\end{align*}
Then, parallel composition with a process wating at the other end $d$ of the endpoint $c$ will raise blame as follows:
\begin{align*}
  & (\nu c, d)(\proc{\textit{bar}\; (c \colon \Endpl \Cast{\ell'} \DYN)} \PAR \proc{\wait{d}}) \\ \reduces^+ &
  (\nu c,d)(\nu a)(\proc{\close{c}; \close{(a \colon \DYN \Cast{\ell} \Endpl)}} \PAR a \mapsto \lockedbl{\ell'} \PAR \proc{\wait{d}}) ) \\ \reduces &
  (\nu c,d)(\nu a)(\proc{\close{(a \colon \DYN \Cast{\ell} \Endpl)}} \PAR a \mapsto \lockedbl{\ell'} \PAR \proc{()} ) \\ \reduces &
  (\nu c,d)(\nu a)(\blameboth{\dual \ell'}{\ell}{\emptyset} \PAR a \mapsto \lockedbl{\ell'}).
\end{align*}

\subsubsection{External language \GGVe}
\label{sec:GGVe}

%% Copied from surface_typing

%%% environment splitting

\newcommand{\SplitEmpty}{\infax{
     \cdot   \ottsym{=}    \cdot   \circ   \cdot  
}}

\newcommand{\SplitUn}{\infrule{
    \Gamma  \ottsym{=}   \Gamma_{{\mathrm{1}}}  \circ  \Gamma_{{\mathrm{2}}}  \andalso
     \un   \ottsym{(}  \ottnt{T}  \ottsym{)}
  }{
    \Gamma  \ottsym{,}  \ottnt{z}  \ottsym{:}  \ottnt{T}  \ottsym{=}   \ottsym{(}  \Gamma_{{\mathrm{1}}}  \ottsym{,}  \ottnt{z}  \ottsym{:}  \ottnt{T}  \ottsym{)}  \circ  \ottsym{(}  \Gamma_{{\mathrm{2}}}  \ottsym{,}  \ottnt{z}  \ottsym{:}  \ottnt{T}  \ottsym{)} 
}}

\newcommand{\SplitLinL}{\infrule{
    \Gamma  \ottsym{=}   \Gamma_{{\mathrm{1}}}  \circ  \Gamma_{{\mathrm{2}}}  \andalso
     \lin   \ottsym{(}  \ottnt{T}  \ottsym{)}
  }{
    \Gamma  \ottsym{,}  \ottnt{z}  \ottsym{:}  \ottnt{T}  \ottsym{=}   \ottsym{(}  \Gamma_{{\mathrm{1}}}  \ottsym{,}  \ottnt{z}  \ottsym{:}  \ottnt{T}  \ottsym{)}  \circ  \Gamma_{{\mathrm{2}}} 
}}

\newcommand{\SplitLinR}{\infrule{
    \Gamma  \ottsym{=}   \Gamma_{{\mathrm{1}}}  \circ  \Gamma_{{\mathrm{2}}}  \andalso
     \lin   \ottsym{(}  \ottnt{T}  \ottsym{)}
  }{
    \Gamma  \ottsym{,}  \ottnt{z}  \ottsym{:}  \ottnt{T}  \ottsym{=}   \Gamma_{{\mathrm{1}}}  \circ  \ottsym{(}  \Gamma_{{\mathrm{2}}}  \ottsym{,}  \ottnt{z}  \ottsym{:}  \ottnt{T}  \ottsym{)} 
}}

%%%%%%%%

%%% typing expressions

\newcommand{\TyName}{\infrule[T-Name]{
     \un   \ottsym{(}  \Gamma  \ottsym{)}
  }{
    \Gamma  \ottsym{,}  \ottnt{z}  \ottsym{:}  \ottnt{T}  \vdash  \ottnt{z}  \ottsym{:}  \ottnt{T}
}}

\newcommand{\TyUnit}{\infrule[T-Unit]{
     \un   \ottsym{(}  \Gamma  \ottsym{)}
  }{
    \Gamma  \vdash  \ottsym{()}  \ottsym{:}   \unit 
}}

\newcommand{\TyAbs}{\infrule[T-Abs]{
    \Gamma  \ottsym{,}  \ottmv{x}  \ottsym{:}  \ottnt{T}  \vdash  \mathbb{e}  \ottsym{:}  \ottnt{U} \andalso
     \super{  \ottnt{m}  }{  \Gamma  } 
  }{
    \Gamma  \vdash   \lambda  _ \ottnt{m}   \ottmv{x} {:} \ottnt{T} .\,  \mathbb{e}   \ottsym{:}   \ottnt{T}   \rightarrow _{ \ottnt{m} }  \ottnt{U} 
}}

\newcommand{\TyApp}{\infrule[T-App]{
    \Gamma  \vdash  \mathbb{e}  \ottsym{:}  \ottnt{T_{{\mathrm{1}}}} \andalso
    \Delta  \vdash  \mathbb{f}  \ottsym{:}  \ottnt{T_{{\mathrm{2}}}} \andalso
     \matching{  \ottnt{T_{{\mathrm{1}}}}  }{   \ottnt{T_{{\mathrm{11}}}}   \rightarrow _{ \ottnt{m} }  \ottnt{T_{{\mathrm{12}}}}   }  \andalso
    \ottnt{T_{{\mathrm{2}}}}  \sim  \ottnt{T_{{\mathrm{11}}}}
  }{
     \Gamma  \circ  \Delta   \vdash  \mathbb{e} \, \mathbb{f}  \ottsym{:}  \ottnt{T_{{\mathrm{12}}}}
}}

\newcommand{\TyPairCons}{\infrule[T-PairCons]{
    \Gamma  \vdash  \mathbb{e}  \ottsym{:}  \ottnt{T} \andalso
    \Delta  \vdash  \mathbb{f}  \ottsym{:}  \ottnt{U} \andalso
     \super{  \ottnt{m}  }{  \ottnt{T}  }  \andalso
     \super{  \ottnt{m}  }{  \ottnt{U}  } 
  }{
     \Gamma  \circ  \Delta   \vdash   (  \mathbb{e} ,\,  \mathbb{f}  )_ \ottnt{m}   \ottsym{:}   \ottnt{T}   \times _{ \ottnt{m} }  \ottnt{U} 
}}

\newcommand{\TyPairDest}{\infrule[T-PairDest]{
    \Gamma  \vdash  \mathbb{e}  \ottsym{:}  \ottnt{T} \andalso
     \matching{  \ottnt{T}  }{   \ottnt{T_{{\mathrm{1}}}}   \times _{ \ottnt{m} }  \ottnt{T_{{\mathrm{2}}}}   }  \andalso
    \Delta  \ottsym{,}  \ottmv{x}  \ottsym{:}  \ottnt{T_{{\mathrm{1}}}}  \ottsym{,}  \ottmv{y}  \ottsym{:}  \ottnt{T_{{\mathrm{2}}}}  \vdash  \mathbb{f}  \ottsym{:}  \ottnt{U}
  }{
     \Gamma  \circ  \Delta   \vdash   \letin{  \ottmv{x} ,  \ottmv{y}  }{ \mathbb{e} }{ \mathbb{f} }   \ottsym{:}  \ottnt{U}
}}

\newcommand{\TyFork}{\infrule[T-Fork]{
    \Gamma  \vdash  \mathbb{e}  \ottsym{:}  \ottnt{T} \andalso
    \ottnt{T}  \sim   \unit 
  }{
    \Gamma  \vdash   \fork{ \mathbb{e} }   \ottsym{:}   \unit 
}}

\newcommand{\TyNew}{\infrule[T-New]{
     \un   \ottsym{(}  \Gamma  \ottsym{)}
  }{
    \Gamma  \vdash   \newk\, \ottnt{S}   \ottsym{:}   \ottnt{S}   \times _{  \lin  }   \dual{  \ottnt{S}  }  
}}

\newcommand{\TySend}{\infrule[T-Send]{
    \Gamma  \vdash  \mathbb{e}_{{\mathrm{1}}}  \ottsym{:}  \ottnt{T_{{\mathrm{1}}}} \andalso
    \Delta  \vdash  \mathbb{e}_{{\mathrm{2}}}  \ottsym{:}  \ottnt{T_{{\mathrm{2}}}} \andalso
     \matching{  \ottnt{T_{{\mathrm{2}}}}  }{   \pl{ \ottnt{T_{{\mathrm{3}}}} }  \ottnt{S}   }  \andalso
    \ottnt{T_{{\mathrm{1}}}}  \sim  \ottnt{T_{{\mathrm{3}}}}
  }{
     \Gamma  \circ  \Delta   \vdash   \send{ \mathbb{e}_{{\mathrm{1}}} }{ \mathbb{e}_{{\mathrm{2}}} }   \ottsym{:}  \ottnt{S}
}}

\newcommand{\TyReceive}{\infrule[T-Receive]{
    \Gamma  \vdash  \mathbb{e}  \ottsym{:}  \ottnt{T_{{\mathrm{1}}}} \andalso
     \matching{  \ottnt{T_{{\mathrm{1}}}}  }{   \qu{ \ottnt{T_{{\mathrm{2}}}} }  \ottnt{S}   } 
  }{
    \Gamma  \vdash   \recv{ \mathbb{e} }   \ottsym{:}   \ottnt{T_{{\mathrm{2}}}}   \times _{  \lin  }  \ottnt{S} 
}}

\newcommand{\TySelect}{\infrule[T-Select]{
    \Gamma  \vdash  \mathbb{e}  \ottsym{:}  \ottnt{T} \andalso
     \matching{  \ottnt{T}  }{   \oplus    \br{  \ottmv{l_{\ottmv{j}}}  :  \ottnt{S_{\ottmv{j}}}  }    } 
  }{
    \Gamma  \vdash   \select{ \ottmv{l_{\ottmv{j}}} }{ \mathbb{e} }   \ottsym{:}  \ottnt{S_{\ottmv{j}}}
}}

\newcommand{\TyCase}{\infrule[T-Case]{
    \Gamma  \vdash  \mathbb{e}  \ottsym{:}  \ottnt{T} \andalso
    \ottnt{T}  \sim   \&    \br{  \ottmv{l_{\ottmv{i}}}  :  \ottnt{S_{\ottmv{i}}}  }_{  \ottmv{i}  \in  \ottnt{I}  }   \andalso
     ( \Delta  \ottsym{,}  \ottmv{x_{\ottmv{i}}}  \ottsym{:}  \ottnt{S_{\ottmv{i}}}  \vdash  \mathbb{e}_{\ottmv{i}}  \ottsym{:}  \ottnt{U} )_{  \ottmv{i}  \in  \ottnt{I}  } 
  }{
     \Gamma  \circ  \Delta   \vdash   \case{ \mathbb{e} }{  \br{  \ottmv{l_{\ottmv{i}}} :  \ottmv{x_{\ottmv{i}}} .\,  \mathbb{e}_{\ottmv{i}}   }_{  \ottmv{i}  \in  \ottnt{I}  }  }   \ottsym{:}  \ottnt{U}
}}

\newcommand{\TyClose}{\infrule[T-Close]{
    \Gamma  \vdash  \mathbb{e}  \ottsym{:}  \ottnt{T} \andalso
    \ottnt{T}  \sim   \Endqu 
  }{
    \Gamma  \vdash   \close{ \mathbb{e} }   \ottsym{:}   \unit 
}}

\newcommand{\TyWait}{\infrule[T-Wait]{
    \Gamma  \vdash  \mathbb{e}  \ottsym{:}  \ottnt{T} \andalso
    \ottnt{T}  \sim   \Endpl 
  }{
    \Gamma  \vdash   \wait{ \mathbb{e} }   \ottsym{:}   \unit 
}}

\newcommand{\TySub}{\infrule[T-Sub]{
    \Gamma  \vdash  \mathbb{e}  \ottsym{:}  \ottnt{T} \andalso
    \ottnt{T}  \Sub  \ottnt{U}
  }{
    \Gamma  \vdash  \mathbb{e}  \ottsym{:}  \ottnt{U}
}}

%%% typing processes

%% \newcommand{\TyExp}{\infrule[T-Exp]{
%%     \Gamma  \vdash  \mathbb{e}  \ottsym{:}  \ottnt{T} \andalso
%%      \un   \ottsym{(}  \ottnt{T}  \ottsym{)}
%%   }{
%%     \Gamma  \vdash   \proc{  \mathbb{e}  } 
%% }}

%% \newcommand{\TyPar}{\infrule[T-Par]{
%%     \Gamma  \vdash  \ottnt{P_{{\mathrm{1}}}} \andalso
%%     \Delta  \vdash  \ottnt{P_{{\mathrm{2}}}}
%%   }{
%%      \Gamma  \circ  \Delta   \vdash  \ottnt{P_{{\mathrm{1}}}}  \PAR  \ottnt{P_{{\mathrm{2}}}}
%% }}

%% \newcommand{\TyNuBind}{\infrule[T-NuBind]{
%%     \Gamma  \ottsym{,}  \ottmv{c}  \ottsym{:}  \ottnt{S}  \ottsym{,}  \ottmv{d}  \ottsym{:}   \dual{  \ottnt{S}  }   \vdash  \ottnt{P}
%%   }{
%%     \Gamma  \vdash   ( \nu   \ottmv{c} {:} \ottnt{S} ,  \ottmv{d} )  \ottnt{P} 
%% }}

%%% Well-typed program

\newcommand{\TyProg}{\infrule{  % T-Exp
      \vdash  \mathbb{e}  \ottsym{:}  \ottnt{T} \andalso
     \un   \ottsym{(}  \ottnt{T}  \ottsym{)}
  }{
     \mathbb{e} \; \textbf{prog} 
}}

%%% Algorithmic typing rules
%% TODO: Use |-A ?

%% OR: \algTyApp, \TyApp2
\newcommand{\algTyApp}{\infrule[TA-App]{
    \Gamma  \vdash  \mathbb{e}_{{\mathrm{1}}}  \ottsym{:}  \ottnt{T_{{\mathrm{1}}}} \andalso
    \Delta  \vdash  \mathbb{e}_{{\mathrm{2}}}  \ottsym{:}  \ottnt{T_{{\mathrm{2}}}} \andalso
     \matching{  \ottnt{T_{{\mathrm{1}}}}  }{   \ottnt{T_{{\mathrm{11}}}}   \rightarrow _{ \ottnt{m} }  \ottnt{T_{{\mathrm{12}}}}   }  \andalso
    \ottnt{T_{{\mathrm{2}}}}  \lesssim  \ottnt{T_{{\mathrm{11}}}}
  }{
     \Gamma  \circ  \Delta   \vdash  \mathbb{e}_{{\mathrm{1}}} \, \mathbb{e}_{{\mathrm{2}}}  \ottsym{:}  \ottnt{T_{{\mathrm{12}}}}
}}

\newcommand{\algTySend}{\infrule[TA-Send]{
    \Gamma  \vdash  \mathbb{e}_{{\mathrm{1}}}  \ottsym{:}  \ottnt{T_{{\mathrm{1}}}} \andalso
    \Delta  \vdash  \mathbb{e}_{{\mathrm{2}}}  \ottsym{:}  \ottnt{T_{{\mathrm{2}}}} \andalso
     \matching{  \ottnt{T_{{\mathrm{2}}}}  }{   \pl{ \ottnt{T_{{\mathrm{3}}}} }  \ottnt{S}   }  \andalso
    \ottnt{T_{{\mathrm{1}}}}  \lesssim  \ottnt{T_{{\mathrm{3}}}}
  }{
     \Gamma  \circ  \Delta   \vdash   \send{ \mathbb{e}_{{\mathrm{1}}} }{ \mathbb{e}_{{\mathrm{2}}} }   \ottsym{:}  \ottnt{S}
}}

\newcommand{\algTyCase}{\infrule[TA-Case]{
    \Gamma  \vdash  \mathbb{e}  \ottsym{:}  \ottnt{T} \andalso
     \matching{  \ottnt{T}  }{   \&    \br{  \ottmv{l_{\ottmv{j}}}  :  \ottnt{R_{\ottmv{j}}}  }_{  \ottmv{j}  \in  \ottnt{J}  }    }  \andalso
     ( \Delta  \ottsym{,}  \ottmv{x_{\ottmv{j}}}  \ottsym{:}  \ottnt{R_{\ottmv{j}}}  \vdash  \mathbb{e}_{\ottmv{j}}  \ottsym{:}  \ottnt{U_{\ottmv{j}}} )_{  \ottmv{j}  \in  \ottnt{J}  }  \andalso
    U = \bigjoin \set{U_j}_{j\in J}
  }{
     \Gamma  \circ  \Delta   \vdash   \case{ \mathbb{e} }{  \br{  \ottmv{l_{\ottmv{j}}} :  \ottmv{x_{\ottmv{j}}} .\,  \mathbb{e}_{\ottmv{j}}   }_{  \ottmv{j}  \in  \ottnt{J}  }  }   \ottsym{:}  \ottnt{U}
}}

%% Copied from cast_insertion

%% cast insertion

\newcommand{\CIName}{\infrule[CI-Name]{
     \un   \ottsym{(}  \Gamma  \ottsym{)}
  }{
    \Gamma  \ottsym{,}  \ottnt{z}  \ottsym{:}  \ottnt{T}  \vdash  \ottnt{z}  \ciarrow  \ottnt{z}  \ottsym{:}  \ottnt{T}
}}

\newcommand{\CIUnit}{\infrule[CI-Unit]{
     \un   \ottsym{(}  \Gamma  \ottsym{)}
  }{
    \Gamma  \vdash  \ottsym{()}  \ciarrow  \ottsym{()}  \ottsym{:}   \unit 
}}

\newcommand{\CIAbs}{\infrule[CI-Abs]{
    \Gamma  \ottsym{,}  \ottmv{x}  \ottsym{:}  \ottnt{T}  \vdash  \mathbb{e}  \ciarrow  \ottnt{f}  \ottsym{:}  \ottnt{U} \andalso
     \super{  \ottnt{m}  }{  \Gamma  } 
  }{
    \Gamma  \vdash   \lambda  _ \ottnt{m}   \ottmv{x} {:} \ottnt{T} .\,  \mathbb{e}   \ciarrow   \lambda  _ \ottnt{m}   \ottmv{x} .\,  \ottnt{f}   \ottsym{:}   \ottnt{T}   \rightarrow _{ \ottnt{m} }  \ottnt{U} 
}}

\newcommand{\CIApp}{\infrule[CI-App]{
    \Gamma  \vdash  \mathbb{e}  \ciarrow  e  \ottsym{:}  \ottnt{T_{{\mathrm{1}}}} \andalso
    \Delta  \vdash  \mathbb{f}  \ciarrow  \ottnt{f}  \ottsym{:}  \ottnt{T_{{\mathrm{2}}}} \andalso
     \matching{  \ottnt{T_{{\mathrm{1}}}}  }{   \ottnt{T_{{\mathrm{11}}}}   \rightarrow _{ \ottnt{m} }  \ottnt{T_{{\mathrm{12}}}}   }  \andalso
    \ottnt{T_{{\mathrm{2}}}}  \lesssim  \ottnt{T_{{\mathrm{11}}}}
  }{
     \Gamma  \circ  \Delta   \vdash  \mathbb{e} \, \mathbb{f}  \ciarrow  \ottsym{(}   e  :  \ottnt{T_{{\mathrm{1}}}}  \Cast{ \ottnt{p} }_?   \ottnt{T_{{\mathrm{11}}}}   \rightarrow _{ \ottnt{m} }  \ottnt{T_{{\mathrm{12}}}}    \ottsym{)} \, \ottsym{(}   \ottnt{f}  :  \ottnt{T_{{\mathrm{2}}}}  \Cast{ \ottnt{p} }_?  \ottnt{T_{{\mathrm{11}}}}   \ottsym{)}  \ottsym{:}  \ottnt{T_{{\mathrm{12}}}}
}}

\newcommand{\CIPairCons}{\infrule[CI-PairCons]{
    \Gamma  \vdash  \mathbb{e}  \ciarrow  e  \ottsym{:}  \ottnt{T} \andalso
    \Delta  \vdash  \mathbb{f}  \ciarrow  \ottnt{f}  \ottsym{:}  \ottnt{U} \andalso
     \super{  \ottnt{m}  }{  \ottnt{T}  }  \andalso
     \super{  \ottnt{m}  }{  \ottnt{U}  } 
  }{
     \Gamma  \circ  \Delta   \vdash   (  \mathbb{e} ,\,  \mathbb{f}  )_ \ottnt{m}   \ciarrow   (  e ,\,  \ottnt{f}  )_ \ottnt{m}   \ottsym{:}   \ottnt{T}   \times _{ \ottnt{m} }  \ottnt{U} 
}}

\newcommand{\CIPairDest}{\infrule[CI-PairDest]{
    \Gamma  \vdash  \mathbb{e}  \ciarrow  e  \ottsym{:}  \ottnt{T} \andalso
     \matching{  \ottnt{T}  }{   \ottnt{T_{{\mathrm{1}}}}   \times _{ \ottnt{m} }  \ottnt{T_{{\mathrm{2}}}}   }  \andalso
    \Delta  \ottsym{,}  \ottmv{x}  \ottsym{:}  \ottnt{T_{{\mathrm{1}}}}  \ottsym{,}  \ottmv{y}  \ottsym{:}  \ottnt{T_{{\mathrm{2}}}}  \vdash  \mathbb{f}  \ciarrow  \ottnt{f}  \ottsym{:}  \ottnt{U}
  }{
     \Gamma  \circ  \Delta   \vdash   \letin{  \ottmv{x} ,  \ottmv{y}  }{ \mathbb{e} }{ \mathbb{f} }   \ciarrow   \letin{  \ottmv{x} ,  \ottmv{y}  }{  \ottsym{(}   e  :  \ottnt{T}  \Cast{ \ottnt{p} }_?   \ottnt{T_{{\mathrm{1}}}}   \times _{ \ottnt{m} }  \ottnt{T_{{\mathrm{2}}}}    \ottsym{)}  }{  \ottnt{f}  }   \ottsym{:}  \ottnt{U}
}}

\newcommand{\CIFork}{\infrule[CI-Fork]{
    \Gamma  \vdash  \mathbb{e}  \ciarrow  e  \ottsym{:}  \ottnt{T} \andalso
    \ottnt{T}  \sim   \unit 
  }{
    \Gamma  \vdash   \fork{ \mathbb{e} }   \ciarrow   \fork{ \ottsym{(}   e  :  \ottnt{T}  \Cast{ \ottnt{p} }_?   \unit    \ottsym{)} }   \ottsym{:}   \unit 
}}

\newcommand{\CINew}{\infrule[CI-New]{
     \un   \ottsym{(}  \Gamma  \ottsym{)}
  }{
    \Gamma  \vdash   \newk\, \ottnt{S}   \ciarrow   \newk   \ottsym{:}   \ottnt{S}   \times _{  \lin  }   \dual{  \ottnt{S}  }  
}}

\newcommand{\CISend}{\infrule[CI-Send]{
    \Gamma  \vdash  \mathbb{e}  \ciarrow  e  \ottsym{:}  \ottnt{T_{{\mathrm{1}}}} \andalso
    \Delta  \vdash  \mathbb{f}  \ciarrow  \ottnt{f}  \ottsym{:}  \ottnt{T_{{\mathrm{2}}}} \andalso
     \matching{  \ottnt{T_{{\mathrm{2}}}}  }{   \pl{ \ottnt{T_{{\mathrm{3}}}} }  \ottnt{S}   }  \andalso
    \ottnt{T_{{\mathrm{1}}}}  \lesssim  \ottnt{T_{{\mathrm{3}}}}
  }{
     \Gamma  \circ  \Delta   \vdash   \send{ \mathbb{e} }{ \mathbb{f} }   \ciarrow   \send{  \ottsym{(}   e  :  \ottnt{T_{{\mathrm{1}}}}  \Cast{ \ottnt{p} }_?  \ottnt{T_{{\mathrm{3}}}}   \ottsym{)}  }{  \ottsym{(}   \ottnt{f}  :  \ottnt{T_{{\mathrm{2}}}}  \Cast{ \ottnt{p} }_?   \pl{ \ottnt{T_{{\mathrm{3}}}} }  \ottnt{S}    \ottsym{)}  }   \ottsym{:}  \ottnt{S}
}}

\newcommand{\CIReceive}{\infrule[CI-Receive]{
    \Gamma  \vdash  \mathbb{e}  \ciarrow  e  \ottsym{:}  \ottnt{T_{{\mathrm{1}}}} \andalso
     \matching{  \ottnt{T_{{\mathrm{1}}}}  }{   \qu{ \ottnt{T_{{\mathrm{2}}}} }  \ottnt{S}   } 
  }{
    \Gamma  \vdash   \recv{ \mathbb{e} }   \ciarrow   \recv{ \ottsym{(}   e  :  \ottnt{T_{{\mathrm{1}}}}  \Cast{ \ottnt{p} }_?   \qu{ \ottnt{T_{{\mathrm{2}}}} }  \ottnt{S}    \ottsym{)} }   \ottsym{:}   \ottnt{T_{{\mathrm{2}}}}   \times _{  \lin  }  \ottnt{S} 
}}

\newcommand{\CISelect}{\infrule[CI-Select]{
    \Gamma  \vdash  \mathbb{e}  \ciarrow  e  \ottsym{:}  \ottnt{T} \andalso
     \matching{  \ottnt{T}  }{   \oplus    \br{  \ottmv{l_{\ottmv{j}}}  :  \ottnt{S_{\ottmv{j}}}  }    } 
  }{
    \Gamma  \vdash   \select{ \ottmv{l_{\ottmv{j}}} }{ \mathbb{e} }   \ciarrow   \select{  \ottmv{l_{\ottmv{j}}}  }{  \ottsym{(}   e  :  \ottnt{T}  \Cast{ \ottnt{p} }_?   \oplus    \br{  \ottmv{l_{\ottmv{j}}}  :  \ottnt{S_{\ottmv{j}}}  }     \ottsym{)}  }   \ottsym{:}  \ottnt{S_{\ottmv{j}}}
}}

\newcommand{\CICase}{\infrule[CI-Case]{
    \Gamma  \vdash  \mathbb{e}  \ciarrow  e  \ottsym{:}  \ottnt{T} \andalso
     \matching{  \ottnt{T}  }{   \&    \br{  \ottmv{l_{\ottmv{j}}}  :  \ottnt{R_{\ottmv{j}}}  }_{  \ottmv{j}  \in  \ottnt{J}  }    }  \andalso
     ( \Delta  \ottsym{,}  \ottmv{x_{\ottmv{j}}}  \ottsym{:}  \ottnt{R_{\ottmv{j}}}  \vdash  \mathbb{e}_{\ottmv{j}}  \ciarrow  \ottnt{f_{\ottmv{j}}}  \ottsym{:}  \ottnt{U_{\ottmv{j}}} )_{  \ottmv{j}  \in  \ottnt{J}  }  \andalso
    U = \bigjoin \set{U_j}_{j\in J}
  }{
     \Gamma  \circ  \Delta   \vdash   \case{ \mathbb{e} }{  \br{  \ottmv{l_{\ottmv{j}}} :  \ottmv{x_{\ottmv{j}}} .\,  \mathbb{e}_{\ottmv{j}}   }_{  \ottmv{j}  \in  \ottnt{J}  }  }   \ciarrow   \case{ \ottsym{(}   e  :  \ottnt{T}  \Cast{ \ottnt{p} }_?   \&    \br{  \ottmv{l_{\ottmv{j}}}  :  \ottnt{R_{\ottmv{j}}}  }_{  \ottmv{j}  \in  \ottnt{J}  }     \ottsym{)} }{  \br{  \ottmv{l_{\ottmv{j}}} :  \ottmv{x_{\ottmv{j}}} .\,   \ottnt{f_{\ottmv{j}}}  :  \ottnt{U_{\ottmv{j}}}  \Cast{ \ottnt{p} }_?  \ottnt{U}    }_{  \ottmv{j}  \in  \ottnt{J}  }  }   \ottsym{:}  \ottnt{U}
}}

\newcommand{\CIClose}{\infrule[CI-Close]{
    \Gamma  \vdash  \mathbb{e}  \ciarrow  e  \ottsym{:}  \ottnt{T} \andalso
    \ottnt{T}  \sim   \Endpl 
  }{
    \Gamma  \vdash   \close{ \mathbb{e} }   \ciarrow   \close{ \ottsym{(}   e  :  \ottnt{T}  \Cast{ \ottnt{p} }_?   \Endpl    \ottsym{)} }   \ottsym{:}   \unit 
}}

\newcommand{\CIWait}{\infrule[CI-Wait]{
    \Gamma  \vdash  \mathbb{e}  \ciarrow  e  \ottsym{:}  \ottnt{T} \andalso
    \ottnt{T}  \sim   \Endqu 
  }{
    \Gamma  \vdash   \wait{ \mathbb{e} }   \ciarrow   \wait{ \ottsym{(}   e  :  \ottnt{T}  \Cast{ \ottnt{p} }_?   \Endqu    \ottsym{)} }   \ottsym{:}   \unit 
}}

%%% CI for proc

%% \newcommand{\CIExp}{\infrule[CI-Exp]{
%%     \Gamma  \vdash  \mathbb{e}  \ciarrow  \ottnt{f}  \ottsym{:}  \ottnt{T} \andalso
%%      \un   \ottsym{(}  \ottnt{T}  \ottsym{)}
%%   }{
%%     \Gamma  \vdash   \proc{  \mathbb{e}  }   \ciarrow   \proc{ \ottnt{f} } 
%% }}

%% \newcommand{\CIPar}{\infrule[CI-Par]{
%%     \Gamma  \vdash  \ottnt{P_{{\mathrm{1}}}}  \ciarrow  \ottnt{Q_{{\mathrm{1}}}} \andalso
%%     \Delta  \vdash  \ottnt{P_{{\mathrm{2}}}}  \ciarrow  \ottnt{Q_{{\mathrm{2}}}}
%%   }{
%%      \Gamma  \circ  \Delta   \vdash  \ottnt{P_{{\mathrm{1}}}}  \PAR  \ottnt{P_{{\mathrm{2}}}}  \ciarrow  \ottnt{Q_{{\mathrm{1}}}}  \PAR  \ottnt{Q_{{\mathrm{2}}}}
%% }}

%% \newcommand{\CINuBind}{\infrule[CI-NuBind]{
%%     \Gamma  \ottsym{,}  \ottmv{c}  \ottsym{:}  \ottnt{S}  \ottsym{,}  \ottmv{d}  \ottsym{:}   \dual{  \ottnt{S}  }   \vdash  \ottnt{P}  \ciarrow  \ottnt{Q}
%%   }{
%%     \Gamma  \vdash   ( \nu   \ottmv{c} {:} \ottnt{S} ,  \ottmv{d} )  \ottnt{P}   \ciarrow   ( \nu   \ottmv{c} ,  \ottmv{d} )  \ottnt{Q} 
%% }}

%%% Syntax

\begin{figure}[t]
  \begin{align*}
%%% Types
%     \textrm{Multiplicities} && m, n & ::=
%     \lin \mid \un \\
% %
%     \textrm{Types} && T,U & ::=
%     \unit \mid
%     S \mid
%      \ottnt{T}   \rightarrow _{ \ottnt{m} }  \ottnt{U}  \mid
%      \ottnt{T}   \times _{ \ottnt{m} }  \ottnt{U}  \mid
%     \DYN \\
% %
%     \textrm{Session Types} && S,R & ::=
%      \pl{ \ottnt{T} }  \ottnt{S}  \mid
%      \qu{ \ottnt{T} }  \ottnt{S}  \mid
%      \oplus    \br{  \ottmv{l_{\ottmv{i}}}  :  \ottnt{S_{\ottmv{i}}}  }_{  \ottmv{i}  \in  \ottnt{I}  }   \mid
%      \&    \br{  \ottmv{l_{\ottmv{i}}}  :  \ottnt{S_{\ottmv{i}}}  }_{  \ottmv{i}  \in  \ottnt{I}  }   \mid
%     \Endpl \mid \Endqu \mid
%     \DC\\
%%% PROGRAM
%%     \textrm{Variables} && x \\
%%     \textrm{Channel endpoints} && c,d \\
    % \textrm{Names} && z & ::=
    % x \mid c \\
%
    &\text{Expressions} &
    \mathbb{e}, \mathbb{f} \grmeq & \ottnt{z}
    \grmor \ottsym{(}    \ottsym{)}
    \grmor  \lambda  _ \ottnt{m}   \ottmv{x} {:} \ottnt{T} .\,  \mathbb{e} 
    \grmor \mathbb{e} \, \mathbb{f}
    \grmor  (  \mathbb{e} ,\,  \mathbb{f}  )_ \ottnt{m} 
    \grmor  \letin{  \ottmv{x} ,  \ottmv{y}  }{ \mathbb{e} }{ \mathbb{f} } 
    \grmor  \fork{ \mathbb{e} } 
    \\ &&
    \grmor &  \newk\, \ottnt{S} 
    \grmor  \send{ \mathbb{e} }{ \mathbb{f} } 
    \grmor  \recv{ \mathbb{e} } 
    \grmor  \select{ \ottmv{l} }{ \mathbb{e} } 
    \grmor  \case{ \mathbb{e} }{  \br{  \ottmv{l_{\ottmv{i}}} :  \ottmv{x_{\ottmv{i}}} .\,  \mathbb{e}_{\ottmv{i}}   }_{  \ottmv{i}  \in  \ottnt{I}  }  } 
    \\ &&
    \grmor &  \close{ \mathbb{e} } 
    \grmor  \wait{ \mathbb{e} } 
    %
    %% \textrm{Processes} && P & ::=
    %% \proc{e} \mid
    %% ( \ottnt{P_{{\mathrm{1}}}}  \PAR  \ottnt{P_{{\mathrm{2}}}} ) \mid
    %%  ( \nu   \ottmv{c} {:} \ottnt{S} ,  \ottmv{d} )  \ottnt{P}  \\
%
    % \textrm{Type Environements} && \Gamma, \Delta & ::=
    %  \cdot  \mid
    % \Gamma  \ottsym{,}  \ottnt{z}  \ottsym{:}  \ottnt{T}
  \end{align*}
  \hcaption{ Expressions in \GGVe{}.}
  \label{fig:surface-syntax}
\end{figure}

Having defined the internal language, we introduce the external
language \GGVe, in which source programs are written.  The syntax of
expressions of \GGVe{} is presented in Figure~\ref{fig:surface-syntax}.
For ease of typechecking, variable declarations in functions and channel
endpoint creations are explicitly typed.
There are no processes in \GGVe: a program is a well-typed closed
expression and it is translated to a \GGVi{} expression before it runs.

\begin{figure}[t]
  Matching \hfill \fbox{$T \matchingsym U$}
  \[
\begin{array}{l@{\;}l@{\qquad}l@{\;}l@{\qquad}l@{\;}l}
  %% \unit &\matchingsym \unit &
  %% \DYN &\matchingsym \unit
  %% \\
   \ottnt{T}   \rightarrow _{ \ottnt{m} }  \ottnt{U}  &\matchingsym  \ottnt{T}   \rightarrow _{ \ottnt{m} }  \ottnt{U}  &
  \DYN &\matchingsym   \DYN    \rightarrow _{  \lin  }   \DYN  
  \\
   \ottnt{T}   \times _{ \ottnt{m} }  \ottnt{U}  &\matchingsym  \ottnt{T}   \times _{ \ottnt{m} }  \ottnt{U}  &
  \DYN &\matchingsym   \DYN    \times _{  \lin  }   \DYN  
  \\
   \pl{ \ottnt{T} }  \ottnt{S}  &\matchingsym  \pl{ \ottnt{T} }  \ottnt{S}  &
  \DYN &\matchingsym  \pl{  \DYN  }   \DC   &
  \DC &\matchingsym  \pl{  \DYN  }   \DC  
  \\
   \qu{ \ottnt{T} }  \ottnt{S}  &\matchingsym  \qu{ \ottnt{T} }  \ottnt{S}  &
  \DYN &\matchingsym  \qu{  \DYN  }   \DC   &
  \DC &\matchingsym  \qu{  \DYN  }   \DC   \\
  \\
   \oplus    \br{  \ottmv{l_{\ottmv{i}}}  :  \ottnt{S_{\ottmv{i}}}  }_{  \ottmv{i}  \in  \ottnt{I}  }   &\matchingsym  \oplus    \br{  \ottmv{l_{\ottmv{j}}}  :  \ottnt{S_{\ottmv{j}}}  }   &
  \DYN &\matchingsym  \oplus    \br{  \ottmv{l_{\ottmv{j}}}  :   \DC   }   &
  \DC &\matchingsym  \oplus    \br{  \ottmv{l_{\ottmv{j}}}  :   \DC   }   \\
  & \qquad (j \in I)
  \\
   \&    \br{  \ottmv{l_{\ottmv{i}}}  :  \ottnt{S_{\ottmv{i}}}  }_{  \ottmv{i}  \in  \ottnt{I}  }   &\matchingsym  \&    \br{  \ottmv{l_{\ottmv{i}}}  :  \ottnt{S_{\ottmv{i}}}  }_{  \ottmv{i}  \in  \ottnt{I}  }   \cup  \br{  \ottmv{l_{\ottmv{j}}}  :  \ottnt{S_{\ottmv{j}}}  }_{  \ottmv{j}  \in  \ottnt{J}  }  &
  \DYN &\matchingsym  \&    \br{  \ottmv{l_{\ottmv{j}}}  :   \DC   }_{  \ottmv{j}  \in  \ottnt{J}  }   &
  \DC &\matchingsym  \&    \br{  \ottmv{l_{\ottmv{j}}}  :   \DC   }_{  \ottmv{j}  \in  \ottnt{J}  }   \\
  & \qquad (J \cap I = \emptyset)
  %% \\
  %%  \Endpl  &\matchingsym  \Endpl  &
  %% \DC &\matchingsym  \Endpl  &
  %% \DYN &\matchingsym  \Endpl 
  %% \\
  %%  \Endqu  &\matchingsym  \Endqu  &
  %% \DC &\matchingsym  \Endqu  &
  %% \DYN &\matchingsym  \Endqu 
\end{array}
\]
%  \vv{Why not all $\DYN$ types on the 2nd column and $\DC$ on the third? (too much space before $\matchingsym$?)}
%  \ai{Good point, done!}
  \hcaption{Matching.}
  \label{fig:matching}
\end{figure}

\begin{figure}[t]
  Multiplicity join and meet \hfill\fbox{$\ottnt{m}  \vee  \ottnt{n}$ \quad $\ottnt{m}  \wedge  \ottnt{n}$}
  \begin{gather*}
   \un   \vee   \un  = \un \gap
   \un   \vee   \lin  = \lin \gap
   \lin   \vee   \un  = \lin \gap
   \lin   \vee   \lin  = \lin \\
   \un   \wedge   \un  = \un \gap
   \un   \wedge   \lin  = \un \gap
   \lin   \wedge   \un  = \un \gap
   \lin   \wedge   \lin  = \lin
  \end{gather*}
%  \spacer
  Type join \hfill\fbox{$\ottnt{T}  \vee  \ottnt{U}$}
  \begin{align*}
     \unit   \vee   \unit  &= \unit \\
    \ottsym{(}   \ottnt{T}   \rightarrow _{ \ottnt{m} }  \ottnt{U}   \ottsym{)}  \vee  \ottsym{(}   \ottnt{T'}   \rightarrow _{ \ottnt{n} }  \ottnt{U'}   \ottsym{)} &=  \ottsym{(}  \ottnt{T}  \wedge  \ottnt{T'}  \ottsym{)}   \rightarrow _{  \ottnt{m}  \vee  \ottnt{n}  }  \ottsym{(}  \ottnt{U}  \vee  \ottnt{U'}  \ottsym{)}  \\
    \ottsym{(}   \ottnt{T}   \times _{ \ottnt{m} }  \ottnt{U}   \ottsym{)}  \vee  \ottsym{(}   \ottnt{T'}   \times _{ \ottnt{n} }  \ottnt{U'}   \ottsym{)} &=  \ottsym{(}  \ottnt{T}  \vee  \ottnt{T'}  \ottsym{)}   \times _{  \ottnt{m}  \vee  \ottnt{n}  }  \ottsym{(}  \ottnt{U}  \vee  \ottnt{U'}  \ottsym{)}  \\
     \pl{ \ottnt{T} }  \ottnt{S}   \vee   \pl{ \ottnt{T'} }  \ottnt{S'}  &=  \pl{ \ottsym{(}  \ottnt{T}  \wedge  \ottnt{T'}  \ottsym{)} }  \ottsym{(}  \ottnt{S}  \vee  \ottnt{S'}  \ottsym{)}  \\
     \qu{ \ottnt{T} }  \ottnt{S}   \vee   \qu{ \ottnt{T'} }  \ottnt{S'}  &=  \qu{ \ottsym{(}  \ottnt{T}  \vee  \ottnt{T'}  \ottsym{)} }  \ottsym{(}  \ottnt{S}  \vee  \ottnt{S'}  \ottsym{)}  \\
    %% select
     \oplus    \br{  \ottmv{l_{\ottmv{i}}}  :  \ottnt{S_{\ottmv{i}}}  }_{  \ottmv{i}  \in  \ottnt{I}  }    \vee   \oplus    \br{  \ottmv{l_{\ottmv{j}}}  :  \ottnt{R_{\ottmv{j}}}  }_{  \ottmv{j}  \in  \ottnt{J}  }  
    % &=  \oplus    \br{  \ottmv{l_{\ottmv{i}}}  :  \ottnt{S_{\ottmv{i}}}  \vee  \ottnt{R_{\ottmv{i}}}  }_{  \ottmv{i}  \in   \ottnt{I}  \cap  \ottnt{J}   }   \\
    &= \oplus \br{\ottmv{l_{\ottmv{i}}}: \ottnt{S'_{\ottmv{i}}} \mid \ottnt{S'_{\ottmv{i}}} = \ottnt{S_{\ottmv{i}}}  \vee  \ottnt{R_{\ottmv{i}}} \text{ is defined and }  \ottmv{i}  \in   \ottnt{I}  \cap  \ottnt{J}  } \\
    %% case; break I \/ J into three parts
     \&    \br{  \ottmv{l_{\ottmv{i}}}  :  \ottnt{S_{\ottmv{i}}}  }_{  \ottmv{i}  \in  \ottnt{I}  }    \vee   \&    \br{  \ottmv{l_{\ottmv{j}}}  :  \ottnt{R_{\ottmv{j}}}  }_{  \ottmv{j}  \in  \ottnt{J}  }  
    &= \&  \br{  \ottmv{l_{\ottmv{i}}}  :  \ottnt{S_{\ottmv{i}}}  }_{  \ottmv{i}  \in   \ottnt{I}  \setminus  \ottnt{J}   } 
       \cup  \br{  \ottmv{l_{\ottmv{k}}}  :  \ottnt{S_{\ottmv{k}}}  \vee  \ottnt{R_{\ottmv{k}}}  }_{  \ottmv{k}  \in   \ottnt{I}  \cap  \ottnt{J}   } 
       \cup  \br{  \ottmv{l_{\ottmv{j}}}  :  \ottnt{R_{\ottmv{j}}}  }_{  \ottmv{j}  \in   \ottnt{J}  \setminus  \ottnt{I}   }  \\
     \Endpl   \vee   \Endpl  &= \Endpl \\
     \Endqu   \vee   \Endqu  &= \Endqu \\
     \DYN   \vee  \ottnt{T} &= T \\
    \ottnt{T}  \vee   \DYN  &= T \\
     \DC   \vee  \ottnt{S} &= S \\
    \ottnt{S}  \vee   \DC  &= S
%     \DYN   \vee   \DYN  &= \DYN \\
%     \DC   \vee   \DC  &= \DC %\\
%    \ottnt{T}  \vee  \ottnt{U} &= \textrm{undefined} \quad \textrm{(otherwise)}
  \end{align*}
%  \spacer
  Type meet \hfill\fbox{$\ottnt{T}  \wedge  \ottnt{U}$}
  \begin{align*}
     \unit   \wedge   \unit  &= \unit \\
    \ottsym{(}   \ottnt{T}   \rightarrow _{ \ottnt{m} }  \ottnt{U}   \ottsym{)}  \wedge  \ottsym{(}   \ottnt{T'}   \rightarrow _{ \ottnt{n} }  \ottnt{U'}   \ottsym{)} &=  \ottsym{(}  \ottnt{T}  \vee  \ottnt{T'}  \ottsym{)}   \rightarrow _{  \ottnt{m}  \wedge  \ottnt{n}  }  \ottsym{(}  \ottnt{U}  \wedge  \ottnt{U'}  \ottsym{)}  \\
    \ottsym{(}   \ottnt{T}   \times _{ \ottnt{m} }  \ottnt{U}   \ottsym{)}  \wedge  \ottsym{(}   \ottnt{T'}   \times _{ \ottnt{n} }  \ottnt{U'}   \ottsym{)} &=  \ottsym{(}  \ottnt{T}  \wedge  \ottnt{T'}  \ottsym{)}   \times _{  \ottnt{m}  \wedge  \ottnt{n}  }  \ottsym{(}  \ottnt{U}  \wedge  \ottnt{U'}  \ottsym{)}  \\
     \pl{ \ottnt{T} }  \ottnt{S}   \wedge   \pl{ \ottnt{T'} }  \ottnt{S'}  &=  \pl{ \ottsym{(}  \ottnt{T}  \vee  \ottnt{T'}  \ottsym{)} }  \ottsym{(}  \ottnt{S}  \wedge  \ottnt{S'}  \ottsym{)}  \\
     \qu{ \ottnt{T} }  \ottnt{S}   \wedge   \qu{ \ottnt{T'} }  \ottnt{S'}  &=  \qu{ \ottsym{(}  \ottnt{T}  \wedge  \ottnt{T'}  \ottsym{)} }  \ottsym{(}  \ottnt{S}  \wedge  \ottnt{S'}  \ottsym{)}  \\
    %% select; break I \/ J into three parts
     \oplus    \br{  \ottmv{l_{\ottmv{i}}}  :  \ottnt{S_{\ottmv{i}}}  }_{  \ottmv{i}  \in  \ottnt{I}  }    \wedge   \oplus    \br{  \ottmv{l_{\ottmv{j}}}  :  \ottnt{R_{\ottmv{j}}}  }_{  \ottmv{j}  \in  \ottnt{J}  }  
    &= \oplus  \br{  \ottmv{l_{\ottmv{i}}}  :  \ottnt{S_{\ottmv{i}}}  }_{  \ottmv{i}  \in   \ottnt{I}  \setminus  \ottnt{J}   } 
       \cup  \br{  \ottmv{l_{\ottmv{k}}}  :  \ottnt{S_{\ottmv{k}}}  \wedge  \ottnt{R_{\ottmv{k}}}  }_{  \ottmv{k}  \in   \ottnt{I}  \cap  \ottnt{J}   } 
       \cup  \br{  \ottmv{l_{\ottmv{j}}}  :  \ottnt{R_{\ottmv{j}}}  }_{  \ottmv{j}  \in   \ottnt{J}  \setminus  \ottnt{I}   }  \\
    %% case
     \&    \br{  \ottmv{l_{\ottmv{i}}}  :  \ottnt{S_{\ottmv{i}}}  }_{  \ottmv{i}  \in  \ottnt{I}  }    \wedge   \&    \br{  \ottmv{l_{\ottmv{j}}}  :  \ottnt{R_{\ottmv{j}}}  }_{  \ottmv{j}  \in  \ottnt{J}  }  
    %% &=  \&    \br{  \ottmv{l_{\ottmv{i}}}  :  \ottnt{S_{\ottmv{i}}}  \wedge  \ottnt{R_{\ottmv{i}}}  }_{  \ottmv{i}  \in   \ottnt{I}  \cap  \ottnt{J}   }   \\
    &=  \&  \br{\ottmv{l_{\ottmv{i}}}: \ottnt{S'_{\ottmv{i}}} \mid \ottnt{S'_{\ottmv{i}}} = \ottnt{S_{\ottmv{i}}}  \wedge  \ottnt{R_{\ottmv{i}}} \text{ is defined and }  \ottmv{i}  \in   \ottnt{I}  \cap  \ottnt{J}  } \\     \Endpl   \wedge   \Endpl  &= \Endpl \\
     \Endqu   \wedge   \Endqu  &= \Endqu \\
     \DYN   \wedge  \ottnt{T} &= T \\
    \ottnt{T}  \wedge   \DYN  &= T \\
     \DC   \wedge  \ottnt{S} &= S \\
    \ottnt{S}  \wedge   \DC  &= S
%     \DYN   \wedge   \DYN  &= \DYN \\
%     \DC   \wedge   \DC  &= \DC %\\
%    \ottnt{T}  \wedge  \ottnt{U} &= \textrm{undefined} \quad \textrm{(otherwise)}
  \end{align*}
%\vv{why not $\oplus\{l_k\colon S_k \vee R_k\}_{k\in I\cap J}$ in the rhs, as in $\&$? (and similarly for $\wedge$).}
%\ai{Added comments in the body of the paper.}
  \hcaption{Join and meet of types.}
  \label{fig:join-meet}
\end{figure}

%% algorithmic rules (without subsumption rule)
\begin{figure}[t]
  Typing expressions\hfill\fbox{$\Gamma  \vdash  \mathbb{e}  \ottsym{:}  \ottnt{T}$}
  \begin{center}
  \TyName \gap
  \TyUnit \gap
  \TyAbs \\
  \algTyApp \\
  \TyPairCons \\
  \TyPairDest \\
  \TyFork \gap
  \TyNew \\
  \algTySend \\
  \TyReceive \gap
  \TySelect \\
  \algTyCase \\
  \TyClose \gap
  \TyWait
  %%\TySub
  \end{center}
  Typing programs\hfill\fbox{$ \mathbb{e} \; \textbf{prog} $}
  \begin{center}
    \TyProg
  \end{center}
  \hcaption{Expression typing in \GGVe.}
  \label{fig:surface-typing}
\end{figure}

%% \begin{figure}[tbp]
%%   \begin{center}
%%   \TyExp \\
%%   \TyPar \gap
%%   \TyNuBind
%%   \end{center}
%%   \caption{Typing rules for processes}
%%   \label{fig:surface-typing-proc}
%% \end{figure}

%% \begin{figure}[t]
%%   \begin{center}
%%   \algTyApp \gap
%%   \algTySend \gap
%%   \algTyCase \gap
%%   \end{center}
%%   \hcaption{More algorithmic typing rules}
%%   \label{fig:surface-typing-proc}
%% \end{figure}

The type system of \GGVe{} adheres to standard practice for gradually
typed languages
\cite{Siek-et-al-2015-criteria,DBLP:conf/popl/CiminiS16}, but requires
a few adaptations to cater for features not covered in previous work.
We first introduce a few auxiliary definitions used in typing rules.
Figure~\ref{fig:matching} defines the matching relation
$T \triangleright U$~\cite{DBLP:conf/popl/CiminiS16}.  Roughly
speaking, $T \triangleright U$ means that $T$ can be used, after
necessary run-time checking, as $U$.  The second and third columns declare
that, if $T$ is $\DYN$ or $\DC$, then it can be used as any type or
session type, respectively.  Otherwise, the matching relation extracts
substructure, i.e., the domain type, the codomain type, the
first-element type, and so on, from $T$.  So, we have neither
$\DYN \triangleright \unit$ nor $\DC \triangleright \Endpl$ or
$\DC \triangleright \Endqu$.

Matching for the internal and external choice types is slightly
involved as it has to cater for subtyping. Matching for internal choice
is invoked in the type rule for an expression $\select{l}{\mathbb{e}}$. Thanks to
subtyping, the type of $\mathbb{e}$ can be any internal choice with a branch
for label $l$. Hence, matching only asks for the presence of this
single label and extracts its residual.

Dually, matching for external choice is invoked in the rule for a
{$\case{\mathbb{e}} \dots$} expression. Again due to subtyping, the $\casek$
expression can check more labels than provided by the type of
$\mathbb{e}$. Hence, matching allows extra branches to be checked with
arbitrary residual types ($l_j : S_j$ in the definition) while
extracting the residual types for all branches provided by $\mathbb{e}$.

Obtaining the result type of a $\casek$ expression from the types of
its branches requires a join operation $T \vee U$ that ensures that
its result is (in a certain sense) a supertype of both $T$ and $U$.
Figure~\ref{fig:join-meet} contains the definitions of join and its
companion meet, which is needed in contravariant positions of the
type.  Both operations are partial: join or meet is undefined for
cases other than those listed in Figure~\ref{fig:join-meet}.

Join of two $\oplus$-types can be obtained by taking the
joins of the types associated with common labels.  Note that labels
where the joins $S_i \vee R_i$ do not exist will be dropped.  On the
other hand, the label set of the join of two $\&$-types is the union
of the two label sets from the input.  For the common labels in
$I \cap J$, the joins $S_k \vee R_k$ must exist.
Join or meet is undefined if the resulting type
is $\oplus\br{}$ or $\&\br{}$ (with the empty set of labels)
as they are ill-formed types.

Without the last four clauses, which deal with $\DYN$ and $\DC$, the
definitions of the join and meet coincide with those for ordinary
subtyping.  This is motivated by the static embedding property of the
Criteria for Gradual Typing~\cite{Siek-et-al-2015-criteria}, which
requires the typability of a \GGVe{} term without $\DYN$ (or $\DC$ in
our case) is the same as the typability under the \GV{} typing rules.
There are a few choices for the join (and meet) of $\DYN$ and other
types and we choose $\DYN \vee T$ to be $T$ for any $T$ because, as we
prove later, our join then corresponds to the least upper bound with
respect to negative subtyping~\cite{Wadler-Findler-2009}, which is
formally defined later, and we can construct a typechecking algorithm
that produces a minimal type with respect to the negative subtyping.
(The least upper bound with respect to positive subtyping is not a
good choice because $\intk \vee \boolk = \DYN$ holds, invalidating the
static embedding property.)

Typing rules are presented in Figure~\ref{fig:surface-typing}.  The
matching relation is used in elimination rules.  To obtain a
syntax-directed inference system, the subsumption is merged into
function application, sending, $\selectk$, and $\casek$. Moreover, subtyping is
replaced with consistent subtyping.
The type of the whole $\casek$ expression is
obtained by joining the types of the branches.  Finally, the judgment
$\mathbb{e}\ \mathbf{prog}$ means that $\mathbb{e}$ is a Gradual GV
program, which is a closed, well-typed \GGVe{} expression of
unrestricted type. Cast insertion discussed below translates a program to a \GGVi{}
expression $e$, which runs as a process $\proc e$.
For example, we can derive
\begin{align*}
&\vdash \lambda_\un o : \DYN. \lambda_\un c : \DC. \close{(\send{o}{c})} : \DYN \to_\un \DC \to_\un \unitk.
\end{align*}

We also develop a typechecking algorithm for \GGVe{} by following
the standard approach \cite{DBLP:journals/toplas/KobayashiPT99,walker:substructural-type-systems}.
We
define an algorithm $ \textsc{\tyexp}( \Gamma , \mathbb{e} ) $, which
takes a type environment $\Gamma$ and an expression $\mathbb{e}$ and
returns a type $T$ of $\mathbb{e}$ and the set $X$ of linear variables
in $\mathbb{e}$.  We avoid nondeterminism involved in environment
splitting by introducing $X$, which is used to check whether
subexpressions do not use the same (linear) variable more than once.
We present the algorithm in full and prove its correctness in
Appendix~\ref{sec:typech-algor-extern}.  In particular, the algorithm
is shown to compute, for given $\Gamma$ and $\mathbb e$, a minimal
type with respect to negative subtyping (if a typing
exists).

\subsubsection{Cast-inserting translation}
\label{sec:cast-insertion}

A well-typed \GGVe{} expression is translated to a \GGVi{} expression
by dropping type annotations and inserting casts.
Figure~\ref{fig:CI-expr} presents cast insertion.  The judgment
$\Gamma \vdash \mathbb{e} \ciarrow f : T$ means that ``under type
environment $\Gamma$, a \GGVe{} expression $\mathbb{e}$ is translated
to a \GGVi{} expression $f$ at type $T$.''  Most rules are
straightforward: casts are inserted where the matching or consistent
subtyping is used.  In each rule, blame label $p$ is supposed to be
fresh and positive. The notation $f : T \Cast{p}_? U$ is used to avoid
inserting unnecessary casts.
\[
  f : T \Cast{p}_? U  = \begin{cases}
  f                 & \text{if $T \Sub U $} \\
  f : T \Cast{p} U & \text{otherwise}
\end{cases}
\]
Thanks to this optimisation, we can show that a program that does not use \(\DYN\)
or \(\DC\) is translated to a cast-free \GGVi{} expression,
whose behaviour obviously coincides with GV.

\begin{figure}[tbp]
  Cast insertion \hfill \fbox{$\Gamma  \vdash  \mathbb{e}  \ciarrow  \ottnt{f}  \ottsym{:}  \ottnt{T}$}

  \CIName \gap
  \CIUnit \gap
  \CIAbs \nextline
  \CIApp \nextline
  \CIPairCons \nextline
  \CIPairDest \nextline
  \CIFork \gap
  \CINew \nextline
  \CISend \gap
  \CIReceive \nextline
  \CISelect \nextline
  \CICase \nextline
  \CIClose \gap
  \CIWait \nextline

  % \vv{Rule 3: should the two blame labels be different?}
  % \ai{The two casts originate from the same program location, so I think using the same label twice is fine (and it's standard).}
  \hcaption{Cast insertion.}
  \label{fig:CI-expr}
\end{figure}

%% CI for proc
%% \begin{figure}[tbp]
%%   \centering

%%   \CIExp \gap
%%   \CIPar \gap
%%   \CINuBind

%%   \caption{Cast insertion for processes}
%%   \label{fig:CI-proc}
%% \end{figure}

For example, we can derive
\begin{align*}
&\vdash \lambda_\un o : \DYN. \lambda_\un c : \DC. \close{(\send{o}{c})}
 \\ & \qquad \ciarrow
 \lambda_\un o. \lambda_\un c. \close{((\send{o}{(c \colon \DC \Cast{p} !\DYN.\DC)}) \colon \DC \Cast{q} \Endpl)}
: \DYN \to_\un \DC \to_\un \unitk
\end{align*}
for some $p$ and $q$.

\subsubsection{Embedding}
\label{sec:embedding}

\begin{figure}[t]
  \begin{align*}
    \Embed{x} &= x
    \\
    \Embed{()} &= () \colon \unitk \Cast{} \DYN
    \\
    \Embed{\lambda x.e} &=
      (\lambda_\un x. \Embed{e})
        \colon \DYN \to_\un \DYN \Cast{} \DYN
    \\
    \Embed{(e,f)} &=
      (\Embed{e}, \Embed{f})_\un
        \colon \DYN \tcpair_\un \DYN \Cast{} \DYN
    \\
    \Embed{e\,f} &=
      (\Embed{e} \colon \DYN \Cast{} \DYN \to_\lin \DYN) \, \Embed{f}
    \\
    \Embed{\letin{x,y}{e}{f}} &=
      \letin{x,y}{(\Embed{e} \colon \DYN \Cast{} \DYN \times_\lin \DYN)}{\Embed{f}}
    \\
    \Embed{\fork e} &=
      (\fork {(\Embed e \colon \DYN  \Cast{} \unitk)}) \colon \unitk \Cast{} \DYN
    \\
    \Embed{\newk} &=
      \newk \colon \DC \times_\lin \DC \Cast{} \DYN
    \\
    \Embed{\send ef} &=
      (\send{\Embed{e}}{(\Embed{f} \colon \DYN \Cast{} {!{\DYN}.\DC})})
        \colon \DC \Cast{} \DYN
    \\
    \Embed{\recv e} &=
      (\recv{(\Embed{e} \colon \DYN \Cast{} {?{\DYN}.\DC})})
        \colon \DYN \times_\lin \DC \Cast{} \DYN
    \\
    \Embed{\select l e} &=
      (\select l {(\Embed{e} \colon \DYN \Cast{} \oplus\{l \colon \DC\})})
        \colon \DC \Cast{} \DYN
    \\
    \Embed{\case{e}{\{ \casebranch{l_i}{x_i}{e_i} \}_{i\in I}}} &=
      \case{(\Embed{e} \colon \DYN \Cast{} \with\{l_i\colon \DC\}_{i\in I} )}
           {\{\casebranch{l_i}{y_i}{\letin{x_i}{(y_i : \DC \Cast{} \DYN)}{\Embed{e_i}}} \}_{i\in I}}
    \\
    \Embed{\close e} &=
      (\close{(\Embed{e} \colon \DYN \Cast{} \End_!)}) \colon \unitk \Cast{} \DYN
    \\
    \Embed{\wait e} &=
      (\wait{(\Embed{e} \colon \DYN \Cast{} \End_?)}) \colon \unitk \Cast{} \DYN
    %% \\
    %% \Embed{\proc{e}} &= \proc{\Embed{e}}
    %% \\
    %% \Embed{P \PAR Q} &= (\Embed{P} \PAR \Embed{Q})
  \end{align*}
  \caption{Embedding of the unityped calculus.}
  \label{fig:embedding-unityped}
\end{figure}

One desideratum for a gradual typing system---if it is equipped with dynamic typing---is that it is possible to
embed an untyped (or rather, unityped) language within it
\cite{Siek-et-al-2015-criteria}.  An embedding of an untyped variant
of \GV\ into \GGVi{} is given in Figure~\ref{fig:embedding-unityped}.
Blame labels are omitted; each cast should receive a unique blame
label.  The untyped variant has the same syntax as the expressions of
\GV, but every expression has type $\DYN$ and multiplicities are implicitly assumed to be \(\un\).  The embedding extends that
of \citep{Wadler-Findler-2009} for the untyped lambda calculus into
the blame calculus.

\section{Results}
\label{sec:results}

We study some of the basic properties~\cite{Siek-et-al-2015-criteria}
of \GGV{} in this section.  They include (1) type safety of \GGVi{} and
(2) blame safety of \GGVi, (3) conservative typing of \GGVe{} over
\GV{}, and (4) the gradual guarantee for \GGVe{}.  Since \GGVi{} do
not guarantee deadlock freedom, type safety is stated as the
combination of preservation and absence of run-time errors, rather
than progress.  We show that (1)--(3) hold with their proof sketches.
For (4), we show that \GGVe{} does \emph{not} satisfy the gradual
guarantee.

\subsection{Preservation and Absence of Run-Time Errors for \GGVi}

We show preservation and absence of run-time errors for \GGVi.  The
basic structure of the proof follows Gay and Vasconcelos
\shortcite{Gay-Vasconcelos-2010}.  In proofs, we often use inversion
properties for the typing relation, such as ``if
$\Gamma \vdash x \colon T$, then $\Gamma = \Gamma', x \colon S$ for
some $S$ and $\Gamma'$ such that $S \Sub T$ and $\un(\Gamma')$,''
 without even stating.  They are easy (but tedious) to state and prove because the only rule that makes typing rules not syntax-directed is \rn{T-Sub} (see, for example, \cite{Pierce2002-tpl} for details).  Similarly, we omit inversion for subtyping, which is syntax-directed.

\begin{lemma}[Weakening]
  \label{lem:weakening}
  If $\Gamma \vdash e\colon T$ and $\un(U)$, then
  $\Gamma,x\colon U \vdash e\colon T$.
\end{lemma}

\begin{proof}
  By induction on $\Gamma \vdash e\colon T$.
\end{proof}

\begin{lemma}[Strengthening]
  \label{lem:strengthening}
  If   $\Gamma,x\colon U \vdash e\colon T$ and
  $x$ does not occur free in $e$, then
  $\Gamma \vdash e\colon T$.
\end{lemma}

\begin{proof}
  By induction on $\Gamma, x\colon U \vdash e\colon T$.
\end{proof}

\begin{lemma}[Preservation for $\equiv$]
  \label{lem:preservation-equiv}
  If $P \equiv Q$, then $\Gamma\vdash P$ if and only if $\Gamma\vdash Q$.
\end{lemma}

\begin{proof}
  By induction on $P \equiv Q$.  Use Lemmas~\ref{lem:weakening},
  \ref{lem:strengthening}, and %\ref{lem:scope-extrusion-circ}
  basic properties of context
  splitting~\cite{vasconcelos:fundamental-sessions,walker:substructural-type-systems}
  for the scope extrusion rules.
\end{proof}

\begin{lemma}
  \label{lem:circ-un}
  If $\Gamma = \Gamma_1 \circ \Gamma_2$ and  $\un(\Gamma_1)$, then
  $\Gamma = \Gamma_{2}$.
\end{lemma}

\begin{proof}
By induction on   $\Gamma = \Gamma_1 \circ \Gamma_2$.
\end{proof}

\begin{lemma}\label{lem:unT-unGamma}
  If $\Gamma \vdash v : T$ and $\un(T)$, then $\un(\Gamma)$.
\end{lemma}

\begin{proof}
  By case analysis on the last rule used to derive $\Gamma \vdash v : T$.
\end{proof}

\begin{lemma}[Substitution]
\label{lem:substitution}
  If $\Gamma_1\vdash v: U$ and $\Gamma_2, x: U \vdash e : T$ and $\Gamma =
  \Gamma_1 \circ \Gamma_2$, then $\Gamma \vdash e \subs vx : T$.
\end{lemma}

\begin{proof}
  By induction on $\Gamma_2, x: U \vdash e : T$ with case analysis on
  the last derivation rule used.  We show main cases below.

  \begin{description}
  \item[Case] (variables):
    If $e=x$ and $T=U$ and $\un(\Gamma_2)$, then we have, by Lemma~\ref{lem:circ-un},
    $\Gamma = \Gamma_1$, finishing the case.
    If $e=y \neq x$, then Lemma~\ref{lem:strengthening} finishes the case.

  % \item[Case] where the last rule is for abstractions:
  %   We have $e= \lambda_m y.e_0$ and $\Gamma_2, x:U, y:T_1 \to_\Multm T_2$ and $\super{\Multm}{\Gamma, x:U}$.

  \item[Case] (applications):
    We have $e = e_1\,e_2$ and $\Gamma_{11} \vdash e_1 : T_2 \to_\Multm T$ and
    $\Gamma_{12} \vdash e_2 : T_2$ and $\Gamma, x:U = \Gamma_{11} \circ \Gamma_{12}$.
    We have two subcases depending on whether $\un(U)$ or not.
    \begin{description}
    \item[Subcase] $\un(U)$: We have $\Gamma_{11} = \Gamma_{11}', x:U$ and $\Gamma_{12} = \Gamma_{12}', x:U$ and $\Gamma = \Gamma_{11}' \circ \Gamma_{12}'$.
    The induction hypothesis give us
    $\Gamma_{11}' \circ \Gamma_2 \vdash e_1 \subs vx : T_2 \to_\Multm
    T$ and $\Gamma_{12}' \circ \Gamma_2 \vdash e_2 \subs vx : T_2$.
    By Lemma~\ref{lem:unT-unGamma}, we have $\un(\Gamma_2)$.
    The typing rule for applications shows
    $(\Gamma_{11}' \circ \Gamma_2) \circ (\Gamma_{12}' \circ \Gamma_2)
    \vdash (e_1\,e_2)\subs vx : T$.
    Lemma~\ref{lem:circ-un} finishes the subcase.
  \item[Subcase] $\lin(U)$: either (1) $\Gamma_{11} = \Gamma_{11}'$
    and $\Gamma_{12} = \Gamma_{12}', x:U$ and
    $\Gamma = \Gamma_{11}' \circ \Gamma_{12}'$,
    in which case we have
    $\Gamma_{12}' \vdash e_2 \subs vx : T_1$ by the induction
    hypothesis and also $e_1 \subs vx = e_1$ and the typing rule for
    applications finishes;
    or (2)
    $\Gamma_{11} = \Gamma_{11}', x:U$ and $\Gamma_{12} = \Gamma_{12}'$
    and $\Gamma = \Gamma_{11}' \circ \Gamma_{12}'$, in which case
    the conclusion is similarly proved.
    \end{description}
  \end{description}
\end{proof}

The following two lemmas are adapted from earlier work~\cite{Gay-Vasconcelos-2010}.

\begin{lemma}[Sub-derivation introduction]
  \label{lem:derivation-intro}
  If $\mathcal D$ is a derivation of $\Gamma \vdash \context{e} : T$,
  then there exist $\Gamma_1$, $\Gamma_2$ and $U$ such that
  $\Gamma = \Gamma_1\circ\Gamma_2$ and $\mathcal D$ has a sub-derivation
  $\mathcal D'$ concluding $\Gamma_2 \vdash e : U$ and the
  position of $\mathcal D'$ in $\mathcal D$ corresponds to the
  position of the hole in $E$.
\end{lemma}

\begin{proof}
  By induction on $E$.
\end{proof}

\begin{lemma}[Sub-derivation elimination]
  \label{lem:derivation-elim}
  $\Gamma \vdash E[f] : T$ holds, if
  \begin{itemize}
    \item $\mathcal D$ is a derivation of
      $\Gamma_1\circ\Gamma_2 \vdash \context{e} : T$,
    \item $\mathcal D'$ is a sub-derivation of $\mathcal D$
      concluding $\Gamma_2 \vdash e : U$,
    \item the position of $\mathcal D'$ in $\mathcal D$ corresponds
      to the position of the hole in $E$,
    \item $\Gamma_3 \vdash f : U$, and
    \item $\Gamma =\Gamma_1\circ\Gamma_3$.
  \end{itemize}
\end{lemma}

\begin{proof}
  By induction on $E$.
\end{proof}

\begin{lemma}\label{lem:flv}
  If $\Gamma \vdash e : T$, then $\flv{\Gamma} = \flv{e}$.
\end{lemma}

\begin{proof}
  Easy induction on $\Gamma \vdash e : T$.
\end{proof}

\begin{theorem}[Preservation for expressions]
\label{thm:sr-exp}
  If $e \reduces f$ and $\Gamma \vdash e:T$, then
    $\Gamma \vdash f:T$.
\end{theorem}
\begin{proof}
  By rule induction on the first hypothesis.
  For $\beta$-reduction and $\letk$ we use the substitution lemma
  (Lemma~\ref{lem:substitution}) and inversion of the typing relation.
\end{proof}

\begin{theorem}[Preservation for processes]
\label{thm:sr-process}
   If $P \reduces Q$ and $\Gamma \vdash P$, then
    $\Gamma \vdash Q$.
\end{theorem}

\begin{proof}
  By rule induction on the first hypothesis, using basic
  properties of context
  splitting~\cite{vasconcelos:fundamental-sessions,walker:substructural-type-systems}
  and weakening (Lemma~\ref{lem:weakening}). Rules that make use of
  context use subderivation introduction
  (Lemma~\ref{lem:derivation-intro}) to build the derivation for the
  hypothesis, and subderivation elimination
  (Lemma~\ref{lem:derivation-elim}) to build the derivation for the
  conclusion.
  Rules for reduction to blame use Lemma~\ref{lem:flv}.
  Reduction underneath parallel composition and scope restriction
  follow by induction.
  The rule for $\equiv$ uses Lemma~\ref{lem:preservation-equiv}.
  Closure under evaluation contexts uses Theorem~\ref{thm:sr-exp}.
\end{proof}

\begin{lemma}[Ground types, subtyping, and consistent subtyping]
  % ~\\[-\baselineskip]
  \label{lem:gtypes}
  \begin{enumerate}
  \item If $T \ne \DYN$,  there is a unique ground type $\GroundT$
    such that $\GroundT \sim T$.
  \item If $S \ne \DC$, there is a unique ground session type
    $\GroundS$ such that $\GroundS \sim S$.
  \item\label{item:consistentsubtype} $\GroundT \lesssim \GroundU$ iff
    $\GroundT \Sub \GroundU$.
  \item\label{item:consistentsubsession} $\GroundS \lesssim \GroundR$
    iff $\GroundS \Sub \GroundR$.
  \end{enumerate}
\end{lemma}
\begin{proof}
  \begin{enumerate}
  \item By case analysis on $T$.
  \item By case analysis on $S$.
  \item By case analysis on $\GroundT$ and $\GroundU$.
  \item By case analysis on $\GroundS$ and $\GroundR$.
  \end{enumerate}
\end{proof}

\begin{lemma}[Canonical forms]
  \label{lem:canonical-forms}
  Suppose that $\Gamma \vdash v : T$ where $\Gamma$ contains
  session types and $\DYN$, only.
 % the only types in $\Gamma$
 %  are session types and $\DYN$.
  \begin{enumerate}
  \item If\/ $T = \DYN$, then either $v = w \colon \GroundT \Cast{p}
    \DYN$ with $\un (\GroundT)$
    % , for some\/ $w$, $p$, and $\GroundT$;
    or $v=\Storea$.
   \item If\/ $T = S$, then either\/ $v = \Chanc$ or\/
     $v = w \colon \GroundS \Cast p \DC$ and\/ $S = \DC$ or\/ 
     $v = w \colon R_1 \Cast{p} R_2$ with $R_2 \Sub S$.
  \item If\/ $T =  \unitk$, then\/ $v = ()$.
  \item If\/ $T =  U_1 \to_\Multm U_2$, then either\/ $v = \lambda_\Multn x.e$
    % for some\/ $\Multn$
    with\/ $\Multn \Sub \Multm$ or\/
    $v = w\colon T_1 \to_{\Multn_1} T_2 \Cast{p} U'_1 \to_{\Multn_2} U'_2$
    with\/ $\Multn_2 \Multsub \Multm$ and $U_1 \Sub U'_1$ and $U'_2 \Sub U_2$.
  \item If\/ $T = T_1 \tcpair_\Multm T_2$, then\/ $v = (w_1,w_2)_\Multn$
    with $\Multn \Multsub \Multm$.
  \end{enumerate}
\end{lemma}
\begin{proof}
  By induction on the derivation on $\Gamma \vdash v : T$.
\end{proof}

\begin{theorem}[Progress for expressions]
  Suppose that $\Gamma \vdash e:T$ and that $\Gamma$ only
  contains channel endpoints and references. Then exactly one of the following cases holds.
\begin{enumerate}
  \item $e$ is a value,
  \item $e \reduces f$ (as an expression),
  \item $e = \context{f}$ and $f$ is one of the \GV{} operations: $\fork f'$,
    $\new$, $\send v \Chanc$, $\recv \Chanc$, $\select l \Chanc$,
    $\casek\, \Chanc\, \ofk \{\casebranch{l_i}{x_i}{ e_i}\}$,  % sacrificing the _{i\in I} for the sake of aesthetics, vv
    $\close \Chanc$, or $\wait \Chanc$,
  \item $e = \context{f}$ and $f$ is a \GGV{} operation:
    \begin{itemize}
    \item $w \colon \GroundT \Cast{p} \DYN$, with $\lin(\GroundT)$,
      % and $\context{}$ does not end in a cast,
    \item 
      $a \colon \DYN \Cast{p} \GroundU$,
    \item 
      $(v \colon \GroundT \Cast{p} \DYN) \colon \DYN \Cast{q}
      \GroundU$, with $\un (\GroundT)$,
    \item 
      $(v \colon \GroundS \Cast{p} \DC) \colon \DC \Cast{q} \GroundR$, or
    \item 
      $c \colon \DC \Cast{p} \GroundS$.
  \end{itemize}
  % \item $e = \context{\Storea \colon \DC \Cast{p}S}$,
  % \item $e = \context{v \colon S \Cast p \DC}$,
  % \item $e = \context{\blame p}$.
\end{enumerate}
\end{theorem}
\begin{proof}
  By induction on expressions, using Canonical forms
  (Lemma~\ref{lem:canonical-forms}).
\end{proof}

The notion of \emph{run-time errors} helps us state our type safety
result.  The \emph{subject} of an expression~$e$, denoted by
$\subj(e)$, is~$c$ when $e$ falls into one of the following cases and
undefined in all other cases.
\begin{equation*}
  \send f\Chanc \qquad
  \recv \Chanc \qquad
  \select l\Chanc \qquad
  \case \Chanc \{\casebranch{l_i}{x_i}{f_i}\}_{i\in I}\qquad
  \close \Chanc \qquad
  \wait \Chanc
\end{equation*}

Two expressions~$e$ and $f$ \emph{agree} on a channel with ends in set
$\{c,d\}$ where $c\ne d$, denoted $\agree^{\{c,d\}}\{e,f\}$, a
relation on two two-element sets, in the following cases.
\begin{enumerate}
\item
  $\agree^{\{c,d\}}\{\send vc, \recv d\}$;
  % $\agree^{\{c,d\}}(\recv c,\send vd)$;
\item
  $\agree^{\{c,d\}}\{\select{l_j}{c}, \case d \{\casebranch{l_i}{x_i}{f_i}\}_{i\in I}\}$ and $j\in I$;
  % $\agree^{\{c,d\}}(\case c \{\casebranch{l_i}{x_i}{f_i}\}_{i\in I}, \select{l_j}{d})$ and $j\in I$;
\item
  $\agree^{\{c,d\}}\{\close c, \wait d\}$.
  % $\agree^{\{c,d\}}(\wait c, \close d)$.
\end{enumerate}

A process is an \emph{error} if it is structurally congruent to some
process that contains a subprocess of one of the following forms.
\begin{enumerate}
\item $\proc{\context{ve}}$ and $v$ is not an abstraction;
\item $\proc{\context{\letin{a,b}{v}{e}}}$ and $v$ is not a pair;
\item $\proc{\context e} \PAR \proc{\context[F] f}$ and $\subj(e) =
  \subj(f)$;
\item $(\nu c,d)(\proc{\context e} \PAR \proc{\context[F] f})$
  % (where $(\overrightarrow \nu)$ denotes a sequence of 
  % other $\nu$-binders vv: no need for we are talking of a subprocess})
  and $\subj(e)=c$ and $\subj(f)=d$ and not $\agree^{\{c,d\}}\{e,f\}$.
\end{enumerate}
The first two cases are typical of functional languages.
The third case ensures no two threads hold references to the
same channel endpoint.
The fourth case ensures channel endpoints agree at all times: if one
process is ready to send then the other is ready to receive,
and similarly for select and case, close and wait.

For processes, rather than a progress result, we present a type safety
result as our type system does not rule out deadlocks, which are
formed by a series of processes each waiting for the next in a
circular arrangement; these are exactly the deadlocked processes of
\GV.  Our result holds both for \GV\ and \GGV\ alike.  The condition
on $\Gamma$ in the statement is to exclude processes getting stuck due
to a free variable in an application ($xe$) or a pair destruction
($\letin{a,b}{x}{e}$).

\begin{theorem}[Absence of run-time errors]
  Let $\Gamma \vdash P$ where $\Gamma$ does not contain function or
  pair types, and let $P \reduces^* Q$. Then $Q$ is not an error.
\end{theorem}
\begin{proof} % vv: omitted for lack of space
  By induction on the length of reduction steps $P \reduces^* Q$.
  For the base case, where $P = Q$, we show $P$ is not an error
  by showing all error processes cannot be well typed.
  
  All cases use Lemma~\ref{lem:derivation-intro} and inversion of the
  typing relation.
  The cases for application and $\letk$ follow from the fact that
  $\Gamma$ does not contain function or pair types.
  The third case follows from the fact that~$c$, being the subject of
  expressions, is of a linear type, hence cannot occur in two distinct
  processes.
  The fourth case follows from the fact that typability implies that
  $c$ and $d$ are of dual types, which in turn implies
  $\agree^{\{c,d\}}(e,f)$.
\end{proof}

\subsection{Blame Safety}
\label{sec:blame-safety}

\begin{figure}[t]
Positive and negative subtyping \hfill \fbox{$T \Sub^+ U \quad T \Sub^- U$}
  \begin{gather*}
    T \Sub^+ \DYN
    \gap
    S \Sub^+ \DC
    \gap
    \DYN \Sub^- T 
    \gap
    \DC \Sub^- S
    \gap
    % \frac
    %  {T \Sub^- \GroundT}
    %  {T \Sub^- \DYN}
    \nextline
    \unit \Sub^\pm \unit
    \gap
    \frac{
      T' \Sub^\mp T \quad
      U \Sub^\pm U' \quad
      m \Multsub n
    }{
      T \to_m U \Sub^\pm T' \to_n U'
    }
    \gap
    \frac{
      T \Sub^\pm T' \quad
      U \Sub^\pm U' \quad
      m \Multsub n
    }{
      T \times_m U \Sub^\pm T' \times_n U'
    }
    \nextline
    \frac{
      T' \Sub^\mp T \quad
      S \Sub^\pm S'
    }{
      \pl{T}S \Sub^\pm \pl{T'}S' 
    }
    \gap 
    \frac{
      T \Sub^\pm T' \quad
      S \Sub^\pm S'
    }{
      \qu{T}S \Sub^\pm \qu{T'}S' 
    }
    % \gap
    % \frac{
    %    (S_k \Sub^- R_k)_{k \in I\cap J}
    % }{
    %   \oplus\br{l_i\colon S_i}_{i\in I} \Sub^-
    %   \oplus\br{l_j\colon R_j}_{j\in J}
    % }
    \nextline
    % \frac{
    %   (S_k \Sub^- R_k)_{k \in I\cap J}
    % }{
    %   \with\{l_i\colon S_i\}_{i\in I} \Sub^-
    %   \with\{l_j\colon R_j\}_{i\in J}
    % }
    % \gap
    \frac{
       J \subseteq I \quad
       (S_j \Sub^\pm R_j)_{j \in J}
    }{
      \oplus\br{l_i\colon S_i}_{i\in I} \Sub^\pm
      \oplus\br{l_j\colon R_j}_{j\in J}
    }
    \gap
    \frac{
      I \subseteq J \quad
      (S_i \Sub^\pm R_i)_{i \in I}
    }{
      \with\{l_i\colon S_i\}_{i\in I} \Sub^\pm
      \with\{l_j\colon R_j\}_{i\in J}
    }
    \nextline
    \Endpl \Sub^\pm \Endpl
    \gap
    \Endqu \Sub^\pm \Endqu
  \end{gather*}
  \spacer

Naive subtyping \hfill \fbox{$T \Subn U$}
  \begin{gather*}
    T \Subn \DYN
    \gap
    S \Subn \DC
    \nextline
    \unit \Subn \unit
    \gap
    \frac{
      T \Subn T' \quad
      U \Subn U'
    }{
      T \to_m U \Subn T' \to_m U'
    }
    \gap
    \frac{
      T \Subn T' \quad
      U \Subn U'
    }{
      T \times_m U \Sub T' \times_m U'
    }
    \nextline
    \frac{
      T \Subn T' \quad
      S \Subn S'
    }{
      \pl{T}S \Subn \pl{T'}S' 
    }
    \gap 
    \frac{
      T \Subn T' \quad
      S \Subn S'
    }{
      \qu{T}S \Subn \qu{T'}S' 
    }
    \nextline
    \frac{
      (S_i \Subn R_i)_{i \in I}
    }{
      \oplus\br{l_i\colon S_i}_{i\in I} \Subn
      \oplus\br{l_i\colon R_i}_{i\in I}
    }
    \gap
    \frac{
      (S_i \Subn R_i)_{i \in I}
    }{
      \with\br{l_i\colon S_i}_{i\in I} \Subn
      \with\br{l_i\colon R_i}_{i\in I}
    }
    \nextline
    \Endpl \Subn \Endpl
    \gap
    \Endqu \Subn \Endqu
  \end{gather*}
  \spacer

Blame safety \hfill \fbox{$e \SafeFor p$}
  \begin{gather*}
    \frac
    { e \SafeFor p \quad T \Sub^+ U }
    { e \colon T \Cast{p} U \SafeFor p }
  \qquad
    \frac
    { e \SafeFor \dual{p} \quad T \Sub^- U }
    { e \colon T \Cast{p} U \SafeFor \dual{p} }
  \qquad
    \frac
    { e \SafeFor p \quad  q \neq p \quad q \neq \dual{p}}
    { e \colon T \Cast{q} U \SafeFor p}
  \\
    \frac
    { q \neq p \quad q' \neq p}
    { \blameboth{q}{q'}{X} \SafeFor p}
  \quad
    \frac
    { q \neq p }
    { \blamegc{q}{X} \SafeFor p}
  \end{gather*}

  \hcaption{Subtyping and blame safety.}
  \label{fig:subtyping}
\end{figure}

Following Wadler and Findler \shortcite{Wadler-Findler-2009} we introduce three new subtyping
relations: $\Sub^+$, $\Sub^-$, and $\Subn$, called positive, negative,
and naive subtyping---also known as precision---respectively, in Figure~\ref{fig:subtyping}, in addition to the ordinary
subtyping $\Sub$ defined in Figure~\ref{fig:ggv-types}.

A cast from $T$ to $U$ with label $p$ may either return a value
or may raise blame labeled $p$ (called \emph{positive} blame) or
$\dual{p}$ (called \emph{negative} blame).
The original subtyping relation $T \Sub U$ of {\GGVi} characterises when a cast from $T$ to $U$
\emph{never} yields blame; relations $T \Sub^+ U$ and $T \Sub^- U$
characterise when a cast from $T$ to $U$ cannot yield \emph{positive}
or \emph{negative} blame, respectively; and relation $T \Subn U$
characterises when type $T$ is more \emph{precise} (in the sense of being less dynamic) than type $U$.  All
four relations are reflexive and transitive, and subtyping, positive
subtyping, and naive subtyping are antisymmetric.

Wadler and Findler \shortcite{Wadler-Findler-2009} have an additional rule that makes any subtype of a ground type a
subtype of $\DYN$, i.e., ${T \Sub \DYN}$ if ${T \Sub
  \GroundT}$. This rule is not sound in \GGV\ because our collide rule blames both casts:  
\[
\proc{\contextE{(v \colon \GroundT \Cast{p} \DYN) \colon \DYN \Cast{q} \GroundU}} \reduces
       \blameboth{\dual{p}}{q}{(\flv{E} \cup \flv{v})} 
       \quad \text{if $\GroundT \not\Sub \GroundU$}
   % \blame{q}&
   % \rtext{if $\GroundT \neq \GroundU$}
\]

The four subtyping relations are closely related.  In previous work
\citep{Wadler-Findler-2009,Siek-et-al-2015-coercion} one has that
proper subtyping decomposes into positive and negative subtyping,
which---after reversing the order on negative subtyping---recompose
into naive subtyping.  Here we have three-quarters of the previous result.
\begin{theorem}[3/4 Tangram] \label{thm:tangram} %~\\[-\baselineskip] 
  \begin{enumerate}
  \item $T \Sub U$\/ implies $T \Sub^+ U$\/ and $T \Sub^- U$.
  \item $S \Sub R$\/ implies $S \Sub^+ R$\/ and $S \Sub^- R$
  \item $T \Subn U$\/ if and only if $T \Sub^+ U$\/ and $U \Sub^- T$.
  \item $S \Subn R$\/ if and only if $S \Sub^+ R$\/ and $R \Sub^- S$
  \end{enumerate}
\end{theorem}
\begin{proof}
By induction on types.
\end{proof}
Here the first and second items are an implication, rather than an equivalence
as in the third and fourth items and previous work.  In order to get an
equivalence, we would need to alter subtyping such that $T \Sub \DYN$ for all $T$
and $S \Sub \DC$ for all $S$, which would interfere with our Canonical Forms lemma (Lemma~\ref{lem:canonical-forms}).
However, implication in all four items is sufficient to ensure the
most important result, Corollary~\ref{cor:wtpcbb} below.

The definitions of negative subtyping and naive subtyping have been
changed since the conference version of the paper.  Now, negative
subtyping supports width subtyping and naive does not.  This change is
motivated by the type system for the external language, in particular
the join operation.  (See the discussion on the join in
Section~\ref{sec:GGVe}.)

The following technical result is used in the proof of
Theorem~\ref{thm:preservation-safe-terms}.
\begin{lemma}\label{lem:ground-positive}
  \begin{enumerate}
  \item If $T \neq \DYN$ and $T \sim \GroundT$, then
    $T \Sub^+ \GroundT$.
  \item If $S \neq \DC$ and $S \sim \GroundS$, then
    $S \Sub^+ \GroundS$.
  \end{enumerate}
\end{lemma}
\begin{proof}
  (1) A case analysis on $T$. Lemma~\ref{lem:gtypes} tells us that
  $\GroundT$ is unique. We show the case for functions. Let $T$ be the
  type $U \rightarrow_\Multm V$; we know that $\GroundT$ is
  $\DYN \rightarrow_\Multm \DYN$, that $\DYN \Sub^- U$, and
  $V \Sub^+ \DYN$. Conclude with the positive subtyping rule for
  functions.
  (2) Similar.
\end{proof}

We say that a process $P$ is \emph{safe}\/ for blame label $p$, if all
occurrences of casts involving $p$ or $\dual p$ correspond to
subsumptions in the (positive or negative) blame subtyping
relation. Figure~\ref{fig:subtyping} defines judgments
$e \SafeFor p$ and $P \SafeFor p$, extended homomorphically to all
other forms of expressions and processes. The \SafeForP{} predicate on
well-typed programs is preserved by reduction.

\begin{theorem}[Preservation of safe terms]
\label{thm:preservation-safe-terms}
  If $\Gamma \vdash P$ with $P\SafeFor p$ and $P \reduces Q$,
  then $Q\SafeFor p$.
\end{theorem}
\begin{proof}
  It is sufficient to examine all reductions whose contractum involves
  coercions. %We list some interesting cases.
  We start with the reductions in
  Figure~\ref{fig:ggv-reductions-expressions}. The four rules starting
  from the one with reductum $v \colon T \Cast{p} \DYN$ follow from
  Lemma~\ref{lem:ground-positive}.
  Then, the standard function cast is analogous to previous work
  \cite{Wadler-Findler-2009}, and the case for pairs is similar.
  The casts for session types (send, receive, select, case, close, and
  wait) are new; we concentrate on send.
  \begin{equation*}
      \send{v}{(w \colon {!T.S} \Cast p {!T'.S'})} \reduces
      (\send{(v \colon T' \Cast{\dual p} T)}{w}) \colon S \Cast p S'
    \end{equation*}
  By assumption $(w \colon \,!T.S \Cast p \,!T'.S') \SafeFor p$.
  Inversion of the $\keyword{safe}\ \keyword{for}$ relation yields
  $T'\Sub^\mp T$ and $S \Sub^\pm S'$.  Hence
  $(v \colon T' \Cast{\dual p} T) \SafeFor p$ and
  $(\dots) \colon S \Cast p S' \SafeFor p$.
  Finally, all rules in Figure~\ref{fig:ggv-reductions-processes}
  preserve casts.
\end{proof}

  % The process reductions for
  % $w \colon S \Cast{p} \DC$, $w \colon \End \Cast p \DC$, and 
  % $v \colon G \Cast{p} \DYN$ all remove the coercion and thus preserve safety.
  % The function domain cast is slightly different.
  % \begin{gather*}
  %   (v \Uparrow R \SCast{p}\Storea S) w \reduces v (w \colon S
  %   \SCast{\dual p}{\Storea} R)
  % \end{gather*}
  % By assumption $(v \Uparrow R \SCast{p}\Storea S) \SafeFor p$.
  % Inversion yields $S \Sub R$, so that $(w \colon S
  % \SCast{\dual p}{\Storea} R) \SafeFor p$.

A process \emph{$P$ blames label $p$} if $P \equiv \Pi(Q \PAR R)$
where $Q$ is  $\blameboth{p}{q}{X}$, $\blameboth{q}{p}{X}$,
or $\blamegc{p}{X}$, for some $q$ and $X$, and
prefix $\Pi$ of bindings
for channel endpoints and references.

\begin{theorem}[Progress of safe terms]
  % For any well-typed process $P$ and blame label $p$.
  If $\Gamma \vdash P$ and $P \SafeFor p$, then
  $P \nreduce Q$ where $Q$ blames~$p$.
\end{theorem}
\begin{proof}
  We analyse all reduction rules whose contractum includes blame.
  From Figure~\ref{fig:ggv-reductions-expressions} take the rule with
  reductum
  $(v \colon \GroundT \Cast{p} \DYN) \colon \DYN \Cast{q} \GroundU$.
  It may blame $\dual{p}$ and $q$, 
  if $\GroundT \notSub \GroundU$.
  However, if it is safe for $\dual{p}$ then $\GroundT \Sub^{-} \DYN$,
  which cannot hold (because only $\DYN \Sub^{-} \DYN$ and $\GroundT$ cannot
  be $\DYN$), and similar reasoning applies for $q$ and $\GroundU$.
  The remaining rules are similar.
\end{proof}

We are finally in a position to state the main result of this
section.

\begin{corollary}[Well-typed programs can’t be blamed]
\label{cor:wtpcbb}
Let $P$ be a well-typed process with a subterm of the form $e\colon T
\Cast p U$ containing the only occurrence of~$p$ and~$\dual p$ in $P$. Then:
\begin{itemize}
\item If $T \Sub^+U$ then $P\nreduce^* Q$ where $Q$ blames $p$.
\item If $T \Sub^-U$ then $P\nreduce^* Q$ where $Q$ blames $\dual{p}$.
\item If $T \Sub U$ then $P\nreduce^* Q$ where $Q$ blames $p$ or $\dual{p}$.
\end{itemize}
\end{corollary}

For example, the redex
$
  (v \colon \GroundT \Cast{p} \DYN) \colon \DYN \Cast{q} \GroundU
$
may fail and blame $\dual{p}$ and $q$ if $\GroundT \notSub \GroundU$.
And indeed we have that $\GroundT \notSub^{-} \DYN$ and $\DYN \notSub^{+} \GroundU$,
so it is not safe for $\dual{p}$ or $q$.  However, $\GroundT \Sub^{+} \DYN$
and $\DYN \Sub^{-} \GroundU$, and the redex will not blame $p$ or $\dual{q}$.

Wadler and Findler \shortcite{Wadler-Findler-2009} explain how casting
between terms related by naive subtyping always places the blame (if
any) on the less-precisely-typed term or context, as appropriate.

\subsection{Properties of \GGVe{}}

Now we turn our attention to \GGVe{} and prove that cast insertion succeeds for well typed
\GGVe{} expressions and preserves typing and
that the \GGVe{} typing conservatively extends the GV typing.
As we need to relate the judgments of different systems, let $\vdashG$ denote the \GGVe{} typing,
$\vdashC$ denote the \GGVi{} typing, and
$\vdashS$ denote the \GV{} typing.

Proposition~\ref{prop:con-sub} goes back to an observation by Siek and Taha \shortcite{DBLP:conf/ecoop/SiekT07}.
\begin{proposition}[Consistent Subtyping] \label{prop:con-sub}
  \begin{enumerate}
  \item $\ottnt{T_{{\mathrm{1}}}}  \lesssim  \ottnt{T_{{\mathrm{2}}}}$ if and only if
    $\ottnt{T_{{\mathrm{1}}}}  \sim  \ottnt{T'_{{\mathrm{1}}}}$ and $\ottnt{T'_{{\mathrm{1}}}}  \Sub  \ottnt{T_{{\mathrm{2}}}}$ for some $\ottnt{T'_{{\mathrm{1}}}}$.
  \item $\ottnt{T_{{\mathrm{1}}}}  \lesssim  \ottnt{T_{{\mathrm{2}}}}$ if and only if
    $\ottnt{T_{{\mathrm{1}}}}  \Sub  \ottnt{T'_{{\mathrm{2}}}}$ and $\ottnt{T'_{{\mathrm{2}}}}  \sim  \ottnt{T_{{\mathrm{2}}}}$ for some $\ottnt{T'_{{\mathrm{2}}}}$.
  \end{enumerate}
\end{proposition}
\begin{proof}
  The left-to-right direction is proved by induction on $\ottnt{T_{{\mathrm{1}}}}  \lesssim  \ottnt{T_{{\mathrm{2}}}}$ and
  the right-to-left is by induction on subtyping with case analysis on $\ottnt{T_{{\mathrm{1}}}}$, $\ottnt{T'_{{\mathrm{2}}}}$, and $\ottnt{T_{{\mathrm{2}}}}$.
\end{proof}

The next lemma clarifies the relation between subtyping, positive and negative subtyping, and consistent subtyping.
\begin{lemma}[Subtyping Hierarchy]\label{lemma:subtyping-hierarchy}
  \begin{enumerate}
  \item   ${} \Sub {} \subseteq {} \Sub^{+}  {} \subseteq {}  \lesssim {} $.
  \item   ${} \Sub {} \subseteq {}    \Sub^{-}     {} \subseteq {}  \lesssim {} $.
  \end{enumerate}
\end{lemma}
\begin{proof}
  ${} \Sub {} \subseteq {} \Sub^{+} $ and
  ${} \Sub {} \subseteq {}    \Sub^{-}    $ follow from Theorem~\ref{thm:tangram}.
  $ \Sub^{+} {} \subseteq {}  \lesssim {} $ and
  $    \Sub^{-}    {} \subseteq {}  \lesssim {} $ are by induction on
  $ \ottnt{T_{{\mathrm{1}}}}   \Sub^{+}   \ottnt{T_{{\mathrm{2}}}} $ and $ \ottnt{T_{{\mathrm{1}}}}   \Sub^{-}   \ottnt{T_{{\mathrm{2}}}} $, respectively.
\end{proof}

\begin{lemma}[Upper bound and lower bound]\label{lemma:bounds}
  \begin{enumerate}
  \item If $\ottnt{T_{{\mathrm{1}}}}  \vee  \ottnt{T_{{\mathrm{2}}}}  \ottsym{=}  \ottnt{U}$, then    % Join
    $ \ottnt{T_{{\mathrm{1}}}}   \Sub^{-}   \ottnt{U} $ and $ \ottnt{T_{{\mathrm{2}}}}   \Sub^{-}   \ottnt{U} $.
  \item If $\ottnt{T_{{\mathrm{1}}}}  \wedge  \ottnt{T_{{\mathrm{2}}}}  \ottsym{=}  \ottnt{U}$, then    % Meet
    $ \ottnt{U}   \Sub^{+}   \ottnt{T_{{\mathrm{1}}}} $ and $ \ottnt{U}   \Sub^{+}   \ottnt{T_{{\mathrm{2}}}} $.
  \end{enumerate}
\end{lemma}
\begin{proof}
  By simultaneous induction on $\ottnt{T_{{\mathrm{1}}}}  \vee  \ottnt{T_{{\mathrm{2}}}}  \ottsym{=}  \ottnt{U}$ (for the first item) and
  $\ottnt{T_{{\mathrm{1}}}}  \wedge  \ottnt{T_{{\mathrm{2}}}}  \ottsym{=}  \ottnt{U}$ (for the second item).
\end{proof}

\begin{lemma}[Least upper bound and greatest lower bound]\label{lemma:polarised-subtyping-lub-glb}
  \begin{enumerate}
  \item If $ \ottnt{T_{{\mathrm{1}}}}   \Sub^{-}   \ottnt{U} $ and $ \ottnt{T_{{\mathrm{2}}}}   \Sub^{-}   \ottnt{U} $, then
    there exists some $U'$ such that $\ottnt{T_{{\mathrm{1}}}}  \vee  \ottnt{T_{{\mathrm{2}}}}  \ottsym{=}  \ottnt{U'}$ and $ \ottnt{U'}   \Sub^{-}   \ottnt{U} $.
  \item If $ \ottnt{U}   \Sub^{+}   \ottnt{T_{{\mathrm{1}}}} $ and $ \ottnt{U}   \Sub^{+}   \ottnt{T_{{\mathrm{2}}}} $, then
    there exists some $U'$ such that $\ottnt{T_{{\mathrm{1}}}}  \wedge  \ottnt{T_{{\mathrm{2}}}}  \ottsym{=}  \ottnt{U'}$ and $ \ottnt{U}   \Sub^{+}   \ottnt{U'} $.
  \end{enumerate}
\end{lemma}

\begin{proof}
  The two items are simultaneously proved by induction on $ \ottnt{T_{{\mathrm{1}}}}   \Sub^{-}   \ottnt{U} $ and
  $ \ottnt{U}   \Sub^{+}   \ottnt{T_{{\mathrm{1}}}} $.
\end{proof}

Theorem~\ref{thm:CI-and-typing} states that
cast insertion succeeds for well typed external language
and preserves typing. A few lemmas are required in preparation. 

\begin{lemma} \label{lem:join-upper-bound}
  If $\ottnt{T_{{\mathrm{1}}}}  \vee  \ottnt{T_{{\mathrm{2}}}}  \ottsym{=}  \ottnt{U}$, then $\ottnt{T_{{\mathrm{1}}}}  \lesssim  \ottnt{U}$ and $\ottnt{T_{{\mathrm{2}}}}  \lesssim  \ottnt{U}$.
  %% No need for meet:
  %%If $\ottnt{T_{{\mathrm{1}}}}  \wedge  \ottnt{T_{{\mathrm{2}}}}  \ottsym{=}  \ottnt{U}$, then $\ottnt{U}  \lesssim  \ottnt{T_{{\mathrm{1}}}}$ and $\ottnt{U}  \lesssim  \ottnt{T_{{\mathrm{2}}}}$.
\end{lemma}
\begin{proof}
  Immediate from Lemmas~\ref{lemma:subtyping-hierarchy}~(1) and~\ref{lemma:bounds}~(1).
\end{proof}

\begin{lemma} \label{lem:matching-con}
  If $ \matching{  \ottnt{T}  }{  \ottnt{U}  } $, then $\ottnt{T}  \lesssim  \ottnt{U}$. % con-sub is enough
  %% Suppose $ \matching{  \ottnt{T}  }{  \ottnt{U}  } $.
  %% \begin{enumerate}
  %% \item If $T$ or $U$ is $\&$-type or $\oplus$-type, then $\ottnt{T}  \lesssim  \ottnt{U}$.
  %% \item If $T$ and $U$ are other type, then $\ottnt{T}  \sim  \ottnt{U}$.
  %% \end{enumerate}
\end{lemma}
\begin{proof}
  By case analysis on $ \matching{  \ottnt{T}  }{  \ottnt{U}  } $.
  %% We use $\ottnt{T}  \lesssim  \ottnt{T}$ and $ \DC   \lesssim  \ottnt{S}$ and $ \DYN   \lesssim  \ottnt{T}$.
  %% We also use if $\ottnt{T}  \Sub  \ottnt{U}$ then $\ottnt{T}  \Sub  \ottnt{U}$ by Lemma~\ref{lemma:subtyping-hierarchy}.
\end{proof}

\begin{theorem}[Cast insertion succeeds and preserves typing] \label{thm:CI-and-typing}
  If $\Gamma  \vdashG  \mathbb{e}  \ottsym{:}  \ottnt{T}$, then
  there exists some $\ottnt{f}$ such that
  $\Gamma  \vdash  \mathbb{e}  \ciarrow  \ottnt{f}  \ottsym{:}  \ottnt{T}$ and $\Gamma  \vdashC  \ottnt{f}  \ottsym{:}  \ottnt{T}$.
  %$\Gamma  \vdash  \mathbb{e}  \ciarrow  \ottnt{f}  \ottsym{:}  \ottnt{T}$ and $\Gamma  \vdashC  \ottnt{f}  \ottsym{:}  \ottnt{T}$ for some $f$.
\end{theorem}
\begin{proof}
  By rules induction on the derivation of $\Gamma  \vdashG  \mathbb{e}  \ottsym{:}  \ottnt{T}$.
  % We proceed by case analysis on the rule applied last.
  We show main cases below.
  \begin{description}
  %% T-App (GGVe)
  \item[Case] application rule: We are given
    \begin{align*}
      &\Gamma  \ottsym{=}   \Gamma_{{\mathrm{1}}}  \circ  \Gamma_{{\mathrm{2}}}  \gap
      \mathbb{e}  \ottsym{=}  \mathbb{e}_{{\mathrm{1}}} \, \mathbb{e}_{{\mathrm{2}}} \gap
      \ottnt{T}  \ottsym{=}  \ottnt{T_{{\mathrm{12}}}} \\
      &\Gamma_{{\mathrm{1}}}  \vdash  \mathbb{e}_{{\mathrm{1}}}  \ottsym{:}  \ottnt{T_{{\mathrm{1}}}} \gap
      \Gamma_{{\mathrm{2}}}  \vdash  \mathbb{e}_{{\mathrm{2}}}  \ottsym{:}  \ottnt{T_{{\mathrm{2}}}} \gap
       \matching{  \ottnt{T_{{\mathrm{1}}}}  }{   \ottnt{T_{{\mathrm{11}}}}   \rightarrow _{ \ottnt{m} }  \ottnt{T_{{\mathrm{12}}}}   }  \gap
      \ottnt{T_{{\mathrm{2}}}}  \lesssim  \ottnt{T_{{\mathrm{11}}}}.
    \end{align*}
    %for some G1,G2,e1,e2,T1,T11,T12,T2
    %
    By $\Gamma_{{\mathrm{1}}}  \vdash  \mathbb{e}_{{\mathrm{1}}}  \ottsym{:}  \ottnt{T_{{\mathrm{1}}}}$ and the IH,
    $\Gamma_{{\mathrm{1}}}  \vdash  \mathbb{e}_{{\mathrm{1}}}  \ciarrow  \ottnt{f_{{\mathrm{1}}}}  \ottsym{:}  \ottnt{T_{{\mathrm{1}}}}$ and $\Gamma_{{\mathrm{1}}}  \vdash  \ottnt{f_{{\mathrm{1}}}}  \ottsym{:}  \ottnt{T_{{\mathrm{1}}}}$
    for some $\ottnt{f_{{\mathrm{1}}}}$.
    By $\Gamma_{{\mathrm{2}}}  \vdash  \mathbb{e}_{{\mathrm{2}}}  \ottsym{:}  \ottnt{T_{{\mathrm{2}}}}$ and the IH,
    $\Gamma_{{\mathrm{2}}}  \vdash  \mathbb{e}_{{\mathrm{2}}}  \ciarrow  \ottnt{f_{{\mathrm{2}}}}  \ottsym{:}  \ottnt{T_{{\mathrm{2}}}}$ and $\Gamma_{{\mathrm{2}}}  \vdash  \ottnt{f_{{\mathrm{2}}}}  \ottsym{:}  \ottnt{T_{{\mathrm{2}}}}$
    for some $\ottnt{f_{{\mathrm{2}}}}$.
    Let \[ \ottnt{f}  \ottsym{=}  \ottsym{(}   \ottnt{f_{{\mathrm{1}}}}  :  \ottnt{T_{{\mathrm{1}}}}  \Cast{ \ottnt{p} }_?   \ottnt{T_{{\mathrm{11}}}}   \rightarrow _{ \ottnt{m} }  \ottnt{T_{{\mathrm{12}}}}    \ottsym{)} \, \ottsym{(}   \ottnt{f_{{\mathrm{2}}}}  :  \ottnt{T_{{\mathrm{2}}}}  \Cast{ \ottnt{p} }_?  \ottnt{T_{{\mathrm{11}}}}   \ottsym{)}. \]
    By the application rule,      % CI-App
    $\Gamma  \vdash  \mathbb{e}  \ciarrow  \ottnt{f}  \ottsym{:}  \ottnt{T}$.

    %% Second half:
    Let us assume $ \ottnt{f_{{\mathrm{1}}}}  :  \ottnt{T_{{\mathrm{1}}}}  \Cast{ \ottnt{p} }_?   \ottnt{T_{{\mathrm{11}}}}   \rightarrow _{ \ottnt{m} }  \ottnt{T_{{\mathrm{12}}}}  $
    equals $ \ottnt{f_{{\mathrm{1}}}}  :  \ottnt{T_{{\mathrm{1}}}}  \Cast{ \ottnt{p} }   \ottnt{T_{{\mathrm{11}}}}   \rightarrow _{ \ottnt{m} }  \ottnt{T_{{\mathrm{12}}}}  $.
    (If $ \ottnt{f_{{\mathrm{1}}}}  :  \ottnt{T_{{\mathrm{1}}}}  \Cast{ \ottnt{p} }_?   \ottnt{T_{{\mathrm{11}}}}   \rightarrow _{ \ottnt{m} }  \ottnt{T_{{\mathrm{12}}}}  $ equals $\ottnt{f_{{\mathrm{1}}}}$,
    we could replace the cast rule with the subsumption rule
    in what follows.)
    We also take similar assumptions in other cases.

    By $ \matching{  \ottnt{T_{{\mathrm{1}}}}  }{   \ottnt{T_{{\mathrm{11}}}}   \rightarrow _{ \ottnt{m} }  \ottnt{T_{{\mathrm{12}}}}   } $ and Lemma~\ref{lem:matching-con},
    $\ottnt{T_{{\mathrm{1}}}}  \lesssim   \ottnt{T_{{\mathrm{11}}}}   \rightarrow _{ \ottnt{m} }  \ottnt{T_{{\mathrm{12}}}} $.
    By $\Gamma_{{\mathrm{1}}}  \vdash  \ottnt{f_{{\mathrm{1}}}}  \ottsym{:}  \ottnt{T_{{\mathrm{1}}}}$ and the cast rule, % T-Cast (GGVi)
    \[ \Gamma_{{\mathrm{1}}}  \vdash  \ottsym{(}   \ottnt{f_{{\mathrm{1}}}}  :  \ottnt{T_{{\mathrm{1}}}}  \Cast{ \ottnt{p} }   \ottnt{T_{{\mathrm{11}}}}   \rightarrow _{ \ottnt{m} }  \ottnt{T_{{\mathrm{12}}}}    \ottsym{)}  \ottsym{:}   \ottnt{T_{{\mathrm{11}}}}   \rightarrow _{ \ottnt{m} }  \ottnt{T_{{\mathrm{12}}}} . \]
    By $\Gamma_{{\mathrm{2}}}  \vdash  \ottnt{f_{{\mathrm{2}}}}  \ottsym{:}  \ottnt{T_{{\mathrm{2}}}}$ and $\ottnt{T_{{\mathrm{2}}}}  \lesssim  \ottnt{T_{{\mathrm{11}}}}$ and the cast rule, % T-Cast
    \[ \Gamma_{{\mathrm{2}}}  \vdash  \ottsym{(}   \ottnt{f_{{\mathrm{2}}}}  :  \ottnt{T_{{\mathrm{2}}}}  \Cast{ \ottnt{p} }  \ottnt{T_{{\mathrm{11}}}}   \ottsym{)}  \ottsym{:}  \ottnt{T_{{\mathrm{11}}}}. \]
    Thus, by the application rule, % T-App (GGVi)
    $\Gamma  \vdash  \ottnt{f}  \ottsym{:}  \ottnt{T}$.

  %% T-Case (GGVe)
  \item[Case] $\casek$ rule: We are given
    \begin{align*}
      & \Gamma  \ottsym{=}   \Gamma'  \circ  \Delta  \gap
      \mathbb{e}  \ottsym{=}   \case{ \mathbb{e}' }{  \br{  \ottmv{l_{\ottmv{j}}} :  \ottmv{x_{\ottmv{j}}} .\,  \mathbb{e}_{\ottmv{j}}   }_{  \ottmv{j}  \in  \ottnt{J}  }  }  \gap
      \ottnt{T}  \ottsym{=}  \ottnt{U} \\
      & \Gamma'  \vdash  \mathbb{e}'  \ottsym{:}  \ottnt{T'} \gap
       \matching{  \ottnt{T'}  }{   \&    \br{  \ottmv{l_{\ottmv{j}}}  :  \ottnt{R_{\ottmv{j}}}  }_{  \ottmv{j}  \in  \ottnt{J}  }    } 
       ( \Delta  \ottsym{,}  \ottmv{x_{\ottmv{j}}}  \ottsym{:}  \ottnt{R_{\ottmv{j}}}  \vdash  \mathbb{e}_{\ottmv{j}}  \ottsym{:}  \ottnt{U_{\ottmv{j}}} )_{  \ottmv{j}  \in  \ottnt{J}  }  \gap
      U = \bigjoin \set{U_j}_{j\in J}.
    \end{align*}
    By $\Gamma'  \vdash  \mathbb{e}'  \ottsym{:}  \ottnt{T'}$ and the IH,
    $\Gamma'  \vdash  \mathbb{e}'  \ciarrow  \ottnt{f}  \ottsym{:}  \ottnt{T'}$ and
    $\Gamma'  \vdash  \ottnt{f'}  \ottsym{:}  \ottnt{T'}$ for some $\ottnt{f'}$.
    We take some $j \in J$.
    By $\Delta  \ottsym{,}  \ottmv{x_{\ottmv{j}}}  \ottsym{:}  \ottnt{R_{\ottmv{j}}}  \vdash  \mathbb{e}_{\ottmv{j}}  \ottsym{:}  \ottnt{U_{\ottmv{j}}}$ and the IH, we have
    $\Delta  \ottsym{,}  \ottmv{x_{\ottmv{j}}}  \ottsym{:}  \ottnt{R_{\ottmv{j}}}  \vdash  \mathbb{e}_{\ottmv{j}}  \ciarrow  \ottnt{f_{\ottmv{j}}}  \ottsym{:}  \ottnt{U_{\ottmv{j}}}$ and
    $\Delta  \ottsym{,}  \ottmv{x_{\ottmv{j}}}  \ottsym{:}  \ottnt{R_{\ottmv{j}}}  \vdash  \ottnt{f_{\ottmv{j}}}  \ottsym{:}  \ottnt{U_{\ottmv{j}}}$ for some $\ottnt{f_{\ottmv{j}}}$.
    Let \[ \ottnt{f}  \ottsym{=}   \case{ \ottsym{(}   \ottnt{f'}  :  \ottnt{T}  \Cast{ \ottnt{p} }_?   \&    \br{  \ottmv{l_{\ottmv{j}}}  :  \ottnt{R_{\ottmv{j}}}  }_{  \ottmv{j}  \in  \ottnt{J}  }     \ottsym{)} }{  \br{  \ottmv{l_{\ottmv{j}}} :  \ottmv{x_{\ottmv{j}}} .\,   \ottnt{f_{\ottmv{j}}}  :  \ottnt{U_{\ottmv{j}}}  \Cast{ \ottnt{p} }_?  \ottnt{U}    }_{  \ottmv{j}  \in  \ottnt{J}  }  } . \]
    % CI-Case
    By the $\casek$ rule, $\Gamma  \vdash  \mathbb{e}  \ciarrow  \ottnt{f}  \ottsym{:}  \ottnt{T}$.

    %% Second half:
    %% consider only when =>? is =>
    Next, by $ \matching{  \ottnt{T'}  }{   \&    \br{  \ottmv{l_{\ottmv{j}}}  :  \ottnt{R_{\ottmv{j}}}  }_{  \ottmv{j}  \in  \ottnt{J}  }    } $ and Lemma~\ref{lem:matching-con},
    $\ottnt{T'}  \lesssim   \&    \br{  \ottmv{l_{\ottmv{j}}}  :  \ottnt{R_{\ottmv{j}}}  }_{  \ottmv{j}  \in  \ottnt{J}  }  $.
    By $\Gamma'  \vdash  \ottnt{f'}  \ottsym{:}  \ottnt{T'}$ and the cast rule, % T-Cast
    \[ \Gamma'  \vdash  \ottsym{(}   \ottnt{f'}  :  \ottnt{T'}  \Cast{ \ottnt{p} }   \&    \br{  \ottmv{l_{\ottmv{j}}}  :  \ottnt{R_{\ottmv{j}}}  }_{  \ottmv{j}  \in  \ottnt{J}  }     \ottsym{)}  \ottsym{:}   \&    \br{  \ottmv{l_{\ottmv{j}}}  :  \ottnt{R_{\ottmv{j}}}  }_{  \ottmv{j}  \in  \ottnt{J}  }  . \]
    We take some $j \in J$.
    By $U = \bigjoin \set{U_j}_{j\in J}$ and Lemma~\ref{lem:join-upper-bound},
    $\ottnt{U_{\ottmv{j}}}  \lesssim  \ottnt{U}$.
    By $\Delta  \ottsym{,}  \ottmv{x_{\ottmv{j}}}  \ottsym{:}  \ottnt{R_{\ottmv{j}}}  \vdash  \ottnt{f_{\ottmv{j}}}  \ottsym{:}  \ottnt{U_{\ottmv{j}}}$ and the cast rule, % T-Cast (GGVi)
    \[ \Delta  \ottsym{,}  \ottmv{x_{\ottmv{j}}}  \ottsym{:}  \ottnt{R_{\ottmv{j}}}  \vdash  \ottsym{(}   \ottnt{f_{\ottmv{j}}}  :  \ottnt{U_{\ottmv{j}}}  \Cast{ \ottnt{p} }  \ottnt{U}   \ottsym{)}  \ottsym{:}  \ottnt{U}. \]
    % T-Case
    Thus, by the $\casek$ rule, $\Gamma  \vdash  \ottnt{f}  \ottsym{:}  \ottnt{T}$.
  \end{description}
\end{proof}

We say that a type, a type environment, or an expression
is \emph{static} in the following sense.
\begin{itemize}
\item A type $T$ is \emph{static} if $T$ does not contain
  any dynamic types: i.e., $\DYN$ or $\DC$.
\item A type environment $\Gamma$ is \emph{static} if
$\Gamma$ contains only static types.
\item An expression $\mathbb{e}$ of \GGVe{} is \emph{static} if all
types declared in $\mathbb{e}$ are static.
\end{itemize}

\begin{lemma} \label{lem:static-con-matching}
  \begin{enumerate}
  \item Suppose $T, U$ are static.
    % no converse here, instead use lemma:subtyping-hierarchy
    If $\ottnt{T}  \lesssim  \ottnt{U}$, then $\ottnt{T}  \Sub  \ottnt{U}$.
    %We have $\ottnt{T}  \lesssim  \ottnt{U}$ if and only if $\ottnt{T}  \Sub  \ottnt{U}$.
  %% \item Suppose $T, U$ are static.
  %%   We have $\ottnt{T}  \sim  \ottnt{U}$ if and only if $T = U$.

  %% NOTE: if T is static, obviously $T \ne \DYN, \DC$.
  %% no converse; because matching definition is enough
  \item Suppose $ \matching{  \ottnt{T}  }{  \ottnt{U}  } $ and $T \ne \DYN, \DC$.
    \begin{enumerate}
    \item If $U$ is neither $\&$-type nor $\oplus$-type, then $T = U$.
    \item If $U$ is either $\&$-type or $\oplus$-type, then $\ottnt{T}  \Sub  \ottnt{U}$.

    %% _static_ part is divided
    %% it is not the case with MatchingCase
    \item If $T$ is static and $U$ is not $\&$-type, then $U$ is static.
    \end{enumerate}

  \item Suppose $T_1, T_2$ are static.
    If $\ottnt{T_{{\mathrm{1}}}}  \vee  \ottnt{T_{{\mathrm{2}}}}  \ottsym{=}  \ottnt{U}$, then
    \begin{enumerate}
    \item $\ottnt{U}$ is static,
    \item $\ottnt{T_{{\mathrm{1}}}}  \Sub  \ottnt{U}$ and $\ottnt{T_{{\mathrm{2}}}}  \Sub  \ottnt{U}$,
    \item $\ottnt{U}  \Sub  \ottnt{U'}$ for any static $\ottnt{U'}$ such that $\ottnt{T_{{\mathrm{1}}}}  \Sub  \ottnt{U'}$ and $\ottnt{T_{{\mathrm{2}}}}  \Sub  \ottnt{U'}$.
    \end{enumerate}
  \item Suppose $T_1, T_2, U'$ are static.
    If $\ottnt{T_{{\mathrm{1}}}}  \Sub  \ottnt{U'}$ and $\ottnt{T_{{\mathrm{2}}}}  \Sub  \ottnt{U'}$, then
    there exists some static $\ottnt{U}$ such that
    $\ottnt{U}  \ottsym{=}  \ottnt{T_{{\mathrm{1}}}}  \vee  \ottnt{T_{{\mathrm{2}}}}$.
  \end{enumerate}
\end{lemma}
\begin{proof}
  The first item is by induction on $\ottnt{T}  \lesssim  \ottnt{U}$.
  The second item is by case analysis on $ \matching{  \ottnt{T}  }{  \ottnt{U}  } $.
  Here, we can prove
  \begin{center}
    if $T, U$ are static and $ \ottnt{T}   \Sub^{-}   \ottnt{U} $, then $\ottnt{T}  \Sub  \ottnt{U}$
  \end{center}
  by induction on $ \ottnt{T}   \Sub^{-}   \ottnt{U} $.
  By Lemma~\ref{lemma:subtyping-hierarchy}, ${} \Sub {}\subseteq{}     \Sub^{-}    $.
  Thus, we have
  \begin{center}
    if $T, U$ are static, then
    $\ottnt{T}  \Sub  \ottnt{U}$ if and only if $ \ottnt{T}   \Sub^{-}   \ottnt{U} $.
  \end{center}
  With this fact,
  the third and fourth item can be proved
  by Lemmas~\ref{lemma:bounds} and~\ref{lemma:polarised-subtyping-lub-glb} respectively.
\end{proof}

We define the type erasure $ \Erase{ \mathbb{e} } $, which is obtained by removing type annotations
from an expression $\mathbb{e}$ of \GGVe. The main cases of its definition
are as follows.
\begin{align*}
   \Erase{  \newk\, \ottnt{S}  }  &= \new \\
   \Erase{  \lambda  _ \ottnt{m}   \ottmv{x} {:} \ottnt{T} .\,  \mathbb{e}  }  &=  \lambda  _ \ottnt{m}   \ottmv{x} .\,   \Erase{ \mathbb{e} }  \ . % \\
  %  \Erase{  \case{ \mathbb{e} }{  \br{  \ottmv{l_{\ottmv{i}}} :  \ottmv{x_{\ottmv{i}}} .\,  \mathbb{e}_{\ottmv{i}}   }_{  \ottmv{i}  \in  \ottnt{I}  }  }  }  &=
  %  \case{  \Erase{ \mathbb{e} }  }{  \br{  \ottmv{l_{\ottmv{i}}} :  \ottmv{x_{\ottmv{i}}} .\,   \Erase{ \mathbb{e}_{\ottmv{i}} }    }_{  \ottmv{i}  \in  \ottnt{I}  }  } 
\end{align*}
(It is extended homomorphically for all other forms of expressions.)

Theorem~\ref{thm:conservativity-typing} states that
the \GGVe{} typing is a conservative extension of the \GV{} typing.
We have to take care of the difference between
the declarative type system of \GV{} and
the algorithmic type system of \GGVe.

\begin{theorem}[Typing Conservation over \GV] \label{thm:conservativity-typing}
  Suppose that $\Gamma$ is static.
  \begin{enumerate}
  \item If $\mathbb{e}$ is static and type environments that appear in the
    derivation of $\Gamma  \vdashG  \mathbb{e}  \ottsym{:}  \ottnt{T}$ are all static, then
    $T$ is static and $\Gamma\vdashS  \Erase{ \mathbb{e} }  : T$.
  \item If $f$ is an expression of \GV{} and $\Gamma\vdashS f : T$, then
    $T$ is static and there exist static $\mathbb{e}$ and static $\ottnt{T'}$ such that
    $ \Erase{ \mathbb{e} }  = f$ and
    $\Gamma  \vdashG  \mathbb{e}  \ottsym{:}  \ottnt{T'}$ and $\ottnt{T'}  \Sub  \ottnt{T}$.
  \end{enumerate}
\end{theorem}
\begin{proof}
  %% soundness of algorithmic typing
  %% using Lemma~\ref{lem:static-con-matching}.
  The first item is by induction on $\Gamma  \vdashG  \mathbb{e}  \ottsym{:}  \ottnt{T}$
  with case analysis on the rule applied last.
  We show the main cases below.
  \begin{description}
  %% T-App (GGVe; renamed from TA-App)
  \item[Case] application rule: We are given
    \begin{align*}
      &\Gamma  \ottsym{=}   \Gamma_{{\mathrm{1}}}  \circ  \Gamma_{{\mathrm{2}}}  \gap
      \mathbb{e}  \ottsym{=}  \mathbb{e}_{{\mathrm{1}}} \, \mathbb{e}_{{\mathrm{2}}} \gap
      \ottnt{T}  \ottsym{=}  \ottnt{T_{{\mathrm{12}}}} \\
      &\Gamma_{{\mathrm{1}}}  \vdash  \mathbb{e}_{{\mathrm{1}}}  \ottsym{:}  \ottnt{T_{{\mathrm{1}}}} \gap
      \Gamma_{{\mathrm{2}}}  \vdash  \mathbb{e}_{{\mathrm{2}}}  \ottsym{:}  \ottnt{T_{{\mathrm{2}}}} \gap
       \matching{  \ottnt{T_{{\mathrm{1}}}}  }{   \ottnt{T_{{\mathrm{11}}}}   \rightarrow _{ \ottnt{m} }  \ottnt{T_{{\mathrm{12}}}}   }  \gap
      \ottnt{T_{{\mathrm{2}}}}  \lesssim  \ottnt{T_{{\mathrm{11}}}}.
    \end{align*}
    Since $\Gamma$, $\mathbb{e}$ are static,
    $\Gamma_{{\mathrm{1}}}$, $\Gamma_{{\mathrm{2}}}$, $\mathbb{e}_{{\mathrm{1}}}$, $\mathbb{e}_{{\mathrm{2}}}$
    are also static. % by trivial lemma
    By $\Gamma_{{\mathrm{1}}}  \vdash  \mathbb{e}_{{\mathrm{1}}}  \ottsym{:}  \ottnt{T_{{\mathrm{1}}}}$ and the IH,
    $\ottnt{T_{{\mathrm{1}}}}$ is static and $\Gamma_{{\mathrm{1}}}  \vdash   \Erase{ \mathbb{e}_{{\mathrm{1}}} }   \ottsym{:}  \ottnt{T_{{\mathrm{1}}}}$.
    By $ \matching{  \ottnt{T_{{\mathrm{1}}}}  }{   \ottnt{T_{{\mathrm{11}}}}   \rightarrow _{ \ottnt{m} }  \ottnt{T_{{\mathrm{12}}}}   } $ and Lemma~\ref{lem:static-con-matching}~(2),
    $\ottnt{T_{{\mathrm{11}}}}$, $\ottnt{T_{{\mathrm{12}}}}$ are static and
    $\ottnt{T_{{\mathrm{1}}}}  \ottsym{=}   \ottnt{T_{{\mathrm{11}}}}   \rightarrow _{ \ottnt{m} }  \ottnt{T_{{\mathrm{12}}}} $. % by trivial lemma
    By $\Gamma_{{\mathrm{2}}}  \vdash  \mathbb{e}_{{\mathrm{2}}}  \ottsym{:}  \ottnt{T_{{\mathrm{2}}}}$ and the IH,
    $\ottnt{T_{{\mathrm{2}}}}$ is static and $\Gamma_{{\mathrm{2}}}  \vdash   \Erase{ \mathbb{e}_{{\mathrm{2}}} }   \ottsym{:}  \ottnt{T_{{\mathrm{2}}}}$.
    By $\ottnt{T_{{\mathrm{2}}}}  \lesssim  \ottnt{T_{{\mathrm{11}}}}$ and Lemma~\ref{lem:static-con-matching}~(1),
    $\ottnt{T_{{\mathrm{2}}}}  \Sub  \ottnt{T_{{\mathrm{11}}}}$.
    %% T-Sub (GGVi)
    By the subsumption rule, $\Gamma_{{\mathrm{2}}}  \vdash   \Erase{ \mathbb{e}_{{\mathrm{2}}} }   \ottsym{:}  \ottnt{T_{{\mathrm{11}}}}$.
    Thus, by
    \[
      \Gamma_{{\mathrm{1}}}  \vdash   \Erase{ \mathbb{e}_{{\mathrm{1}}} }   \ottsym{:}   \ottnt{T_{{\mathrm{11}}}}   \rightarrow _{ \ottnt{m} }  \ottnt{T_{{\mathrm{12}}}}  \gap
      \Gamma_{{\mathrm{2}}}  \vdash   \Erase{ \mathbb{e}_{{\mathrm{2}}} }   \ottsym{:}  \ottnt{T_{{\mathrm{11}}}} \gap
       \Erase{ \mathbb{e}_{{\mathrm{1}}} \, \mathbb{e}_{{\mathrm{2}}} }   \ottsym{=}   \Erase{ \mathbb{e}_{{\mathrm{1}}} }  \,  \Erase{ \mathbb{e}_{{\mathrm{2}}} } 
    \]
    and the application rule, % T-App
    we have $ \Gamma_{{\mathrm{1}}}  \circ  \Gamma_{{\mathrm{2}}}   \vdash   \Erase{ \mathbb{e}_{{\mathrm{1}}} \, \mathbb{e}_{{\mathrm{2}}} }   \ottsym{:}  \ottnt{T_{{\mathrm{12}}}}$.
    %We already have $\ottnt{T}  \ottsym{=}  \ottnt{T_{{\mathrm{12}}}}$ is static.

  %% T-Case (GGVe; renamed from TA-Case)
  \item[Case] $\casek$ rule: We are given
    \begin{align*}
      &\Gamma  \ottsym{=}   \Gamma'  \circ  \Delta  \gap
      \mathbb{e}  \ottsym{=}   \case{ \mathbb{e}' }{  \br{  \ottmv{l_{\ottmv{j}}} :  \ottmv{x_{\ottmv{j}}} .\,  \mathbb{e}_{\ottmv{j}}   }_{  \ottmv{j}  \in  \ottnt{J}  }  }  \gap
      \ottnt{T}  \ottsym{=}  \ottnt{U}\\
      &\Gamma'  \vdash  \mathbb{e}'  \ottsym{:}  \ottnt{T'} \gap
       \matching{  \ottnt{T'}  }{   \&    \br{  \ottmv{l_{\ottmv{j}}}  :  \ottnt{R_{\ottmv{j}}}  }_{  \ottmv{j}  \in  \ottnt{J}  }    }  \gap
       ( \Delta  \ottsym{,}  \ottmv{x_{\ottmv{j}}}  \ottsym{:}  \ottnt{R_{\ottmv{j}}}  \vdash  \mathbb{e}_{\ottmv{j}}  \ottsym{:}  \ottnt{U_{\ottmv{j}}} )_{  \ottmv{j}  \in  \ottnt{J}  }  \gap
      U = \bigjoin \set{U_j}_{j\in J}.
    \end{align*}
    Since $\Gamma$, $\mathbb{e}$ are static,
    $\Gamma'$, $\Delta$, $\mathbb{e}'$, $\mathbb{e}_{\ottmv{j}}$
    are also static. % by trivial lemma
    %% By assumption; because derivation is static
    Since any type environment $\Delta  \ottsym{,}  \ottmv{x_{\ottmv{j}}}  \ottsym{:}  \ottnt{R_{\ottmv{j}}}$ is static,
    any $R_j$ is static.
    By $\Gamma'  \vdash  \mathbb{e}'  \ottsym{:}  \ottnt{T'}$ and the IH,
    $\ottnt{T'}$ is static and $\Gamma'  \vdash   \Erase{ \mathbb{e}' }   \ottsym{:}  \ottnt{T'}$.
    %% Rj are static by assumption, not by lemma
    By $ \matching{  \ottnt{T'}  }{   \&    \br{  \ottmv{l_{\ottmv{j}}}  :  \ottnt{R_{\ottmv{j}}}  }_{  \ottmv{j}  \in  \ottnt{J}  }    } $ and Lemma~\ref{lem:static-con-matching}~(2),
    $\ottnt{T'}  \Sub   \&    \br{  \ottmv{l_{\ottmv{j}}}  :  \ottnt{R_{\ottmv{j}}}  }_{  \ottmv{j}  \in  \ottnt{J}  }  $.
    %% T-Sub
    By the subsumption rule,
    \[ \Gamma'  \vdash   \Erase{ \mathbb{e}' }   \ottsym{:}   \&    \br{  \ottmv{l_{\ottmv{j}}}  :  \ottnt{R_{\ottmv{j}}}  }_{  \ottmv{j}  \in  \ottnt{J}  }  . \]
    We take some $j \in J$.
    Since $\Delta  \ottsym{,}  \ottmv{x_{\ottmv{j}}}  \ottsym{:}  \ottnt{R_{\ottmv{j}}}$ is static, % by assumption
    by $\Delta  \ottsym{,}  \ottmv{x_{\ottmv{j}}}  \ottsym{:}  \ottnt{R_{\ottmv{j}}}  \vdash  \mathbb{e}_{\ottmv{j}}  \ottsym{:}  \ottnt{U_{\ottmv{j}}}$ and the IH,
    $\ottnt{U_{\ottmv{j}}}$ is static and $\Delta  \ottsym{,}  \ottmv{x_{\ottmv{j}}}  \ottsym{:}  \ottnt{R_{\ottmv{j}}}  \vdash   \Erase{ \mathbb{e}_{\ottmv{j}} }   \ottsym{:}  \ottnt{U_{\ottmv{j}}}$.
    By $U = \bigjoin \set{U_j}_{j\in J}$ and Lemma~\ref{lem:static-con-matching}~(3),
    $\ottnt{U_{\ottmv{j}}}  \Sub  \ottnt{U}$ and $\ottnt{U}$ is static.
    %% T-Sub
    By the subsumption rule,
    \[ \Delta  \ottsym{,}  \ottmv{x_{\ottmv{j}}}  \ottsym{:}  \ottnt{R_{\ottmv{j}}}  \vdash   \Erase{ \mathbb{e}_{\ottmv{j}} }   \ottsym{:}  \ottnt{U}. \]
    %
    %%Then, we use \rnp{T-Case}.
    Thus, by
    \begin{align*}
      &\Gamma'  \vdash   \Erase{ \mathbb{e}' }   \ottsym{:}   \&    \br{  \ottmv{l_{\ottmv{j}}}  :  \ottnt{R_{\ottmv{j}}}  }_{  \ottmv{j}  \in  \ottnt{J}  }   \gap
       ( \Delta  \ottsym{,}  \ottmv{x_{\ottmv{j}}}  \ottsym{:}  \ottnt{R_{\ottmv{j}}}  \vdash   \Erase{ \mathbb{e}_{\ottmv{j}} }   \ottsym{:}  \ottnt{U} )_{  \ottmv{j}  \in  \ottnt{J}  }  \\
      & \Erase{  \case{ \mathbb{e}' }{  \br{  \ottmv{l_{\ottmv{j}}} :  \ottmv{x_{\ottmv{j}}} .\,  \mathbb{e}_{\ottmv{j}}   }_{  \ottmv{j}  \in  \ottnt{J}  }  }  }   \ottsym{=}   \case{  \Erase{ \mathbb{e}' }  }{  \br{  \ottmv{l_{\ottmv{j}}} :  \ottmv{x_{\ottmv{j}}} .\,   \Erase{ \mathbb{e}_{\ottmv{j}} }    }_{  \ottmv{j}  \in  \ottnt{J}  }  } 
    \end{align*}
    the $\casek$ rule,            % T-Case
    we have $ \Gamma'  \circ  \Delta   \vdash   \Erase{  \case{ \mathbb{e}' }{  \br{  \ottmv{l_{\ottmv{j}}} :  \ottmv{x_{\ottmv{j}}} .\,  \mathbb{e}_{\ottmv{j}}   }_{  \ottmv{j}  \in  \ottnt{J}  }  }  }   \ottsym{:}  \ottnt{U}$.
    %We already have $\ottnt{T}  \ottsym{=}  \ottnt{U}$ is static.
  \end{description}

  %% completeness of algorithmic typing (TAPL solution 16.2.5)
  %% using Lemma~\ref{lem:static-con-matching}.
  The second item is by induction on $\Gamma \vdashS f : T$
  with case analysis on the rule applied last.
  We show the main cases below.
  \begin{description}
  %% T-App (GV)
  \item[Case] application rule: We are given
    \begin{align*}
      &\Gamma  \ottsym{=}   \Gamma_{{\mathrm{1}}}  \circ  \Gamma_{{\mathrm{2}}}  \gap
      \ottnt{f}  \ottsym{=}  \ottnt{f_{{\mathrm{1}}}} \, \ottnt{f_{{\mathrm{2}}}} \gap
      \ottnt{T}  \ottsym{=}  \ottnt{T_{{\mathrm{12}}}} \gap
      \Gamma_{{\mathrm{1}}}  \vdash  \ottnt{f_{{\mathrm{1}}}}  \ottsym{:}   \ottnt{T_{{\mathrm{11}}}}   \rightarrow _{ \ottnt{m} }  \ottnt{T_{{\mathrm{12}}}}  \gap
      \Gamma_{{\mathrm{2}}}  \vdash  \ottnt{f_{{\mathrm{2}}}}  \ottsym{:}  \ottnt{T_{{\mathrm{11}}}}.
    \end{align*}
    %% f1, f2 is GV expression; this is obvious
    Since $\Gamma$ is static,
    $\Gamma_{{\mathrm{1}}}$, $\Gamma_{{\mathrm{2}}}$ are also static.
    By $\Gamma_{{\mathrm{1}}}  \vdash  \ottnt{f_{{\mathrm{1}}}}  \ottsym{:}   \ottnt{T_{{\mathrm{11}}}}   \rightarrow _{ \ottnt{m} }  \ottnt{T_{{\mathrm{12}}}} $ and the IH,
    $\ottnt{T_{{\mathrm{11}}}}$, $\ottnt{T_{{\mathrm{12}}}}$ are static and
    there exist static $\mathbb{e}_{{\mathrm{1}}}$, $\ottnt{U_{{\mathrm{1}}}}$ such that
    \[
      \Gamma_{{\mathrm{1}}}  \vdash  \mathbb{e}_{{\mathrm{1}}}  \ottsym{:}  \ottnt{U_{{\mathrm{1}}}} \gap
      \ottnt{U_{{\mathrm{1}}}}  \Sub   \ottnt{T_{{\mathrm{11}}}}   \rightarrow _{ \ottnt{m} }  \ottnt{T_{{\mathrm{12}}}}  \gap
       \Erase{ \mathbb{e}_{{\mathrm{1}}} }   \ottsym{=}  \ottnt{f_{{\mathrm{1}}}}.
    \]
    %
    %% $\ottnt{U_{{\mathrm{1}}}}  \Sub   \ottnt{T_{{\mathrm{11}}}}   \rightarrow _{ \ottnt{m} }  \ottnt{T_{{\mathrm{12}}}} $
    By inversion of $ \Sub $, % lemma needed??
    we have $\ottnt{U_{{\mathrm{1}}}}  \ottsym{=}   \ottnt{U_{{\mathrm{11}}}}   \rightarrow _{ \ottnt{n} }  \ottnt{U_{{\mathrm{12}}}} $ and $\ottnt{n}  \Sub  \ottnt{m}$ and
    $\ottnt{T_{{\mathrm{11}}}}  \Sub  \ottnt{U_{{\mathrm{11}}}}$ and $\ottnt{U_{{\mathrm{12}}}}  \Sub  \ottnt{T_{{\mathrm{12}}}}$
    for some $\ottnt{U_{{\mathrm{11}}}}$, $\ottnt{U_{{\mathrm{12}}}}$, $\ottnt{n}$.
    Since $\ottnt{U_{{\mathrm{1}}}}$ is static, $\ottnt{U_{{\mathrm{11}}}}$, $\ottnt{U_{{\mathrm{12}}}}$ are also static.
    By $\Gamma_{{\mathrm{2}}}  \vdash  \ottnt{f_{{\mathrm{2}}}}  \ottsym{:}  \ottnt{T_{{\mathrm{11}}}}$ and the IH,
    $\ottnt{T_{{\mathrm{11}}}}$ is static and     % already known?
    there exist static $\mathbb{e}_{{\mathrm{2}}}$, $\ottnt{U_{{\mathrm{2}}}}$ such that
    \[
      \Gamma_{{\mathrm{2}}}  \vdash  \mathbb{e}_{{\mathrm{2}}}  \ottsym{:}  \ottnt{U_{{\mathrm{2}}}} \gap
      \ottnt{U_{{\mathrm{2}}}}  \Sub  \ottnt{T_{{\mathrm{11}}}} \gap
       \Erase{ \mathbb{e}_{{\mathrm{2}}} }   \ottsym{=}  \ottnt{f_{{\mathrm{2}}}}.
    \]
    By $\ottnt{U_{{\mathrm{2}}}}  \Sub  \ottnt{T_{{\mathrm{11}}}}$ and $\ottnt{T_{{\mathrm{11}}}}  \Sub  \ottnt{U_{{\mathrm{11}}}}$ and transitivity,
    $\ottnt{U_{{\mathrm{2}}}}  \Sub  \ottnt{U_{{\mathrm{11}}}}$.
    %% no need: U2, U11 are static
    By Lemma~\ref{lemma:subtyping-hierarchy}, $\ottnt{U_{{\mathrm{2}}}}  \lesssim  \ottnt{U_{{\mathrm{11}}}}$.
    %% obvious from def of matching
    %% no need: Since $\ottnt{U_{{\mathrm{11}}}}$, $\ottnt{U_{{\mathrm{12}}}}$ are static,
    From Figure~\ref{fig:matching},
    we have $ \matching{   \ottnt{U_{{\mathrm{11}}}}   \rightarrow _{ \ottnt{n} }  \ottnt{U_{{\mathrm{12}}}}   }{   \ottnt{U_{{\mathrm{11}}}}   \rightarrow _{ \ottnt{n} }  \ottnt{U_{{\mathrm{12}}}}   } $.
    Thus, by
    \[
      \Gamma_{{\mathrm{1}}}  \vdash  \mathbb{e}_{{\mathrm{1}}}  \ottsym{:}   \ottnt{U_{{\mathrm{11}}}}   \rightarrow _{ \ottnt{n} }  \ottnt{U_{{\mathrm{12}}}}  \gap
      \Gamma_{{\mathrm{2}}}  \vdash  \mathbb{e}_{{\mathrm{2}}}  \ottsym{:}  \ottnt{U_{{\mathrm{2}}}} \gap
       \matching{   \ottnt{U_{{\mathrm{11}}}}   \rightarrow _{ \ottnt{n} }  \ottnt{U_{{\mathrm{12}}}}   }{   \ottnt{U_{{\mathrm{11}}}}   \rightarrow _{ \ottnt{n} }  \ottnt{U_{{\mathrm{12}}}}   }  \gap
      \ottnt{U_{{\mathrm{2}}}}  \lesssim  \ottnt{U_{{\mathrm{11}}}}
    \]
    and the application rule, % T-App (GGVe)
    $ \Gamma_{{\mathrm{1}}}  \circ  \Gamma_{{\mathrm{2}}}   \vdash  \mathbb{e}_{{\mathrm{1}}} \, \mathbb{e}_{{\mathrm{2}}}  \ottsym{:}  \ottnt{U_{{\mathrm{12}}}}$.
    Additionally,
    \[
       \Erase{ \mathbb{e}_{{\mathrm{1}}} \, \mathbb{e}_{{\mathrm{2}}} }   \ottsym{=}   \Erase{ \mathbb{e}_{{\mathrm{1}}} }  \,  \Erase{ \mathbb{e}_{{\mathrm{2}}} }  = \ottnt{f_{{\mathrm{1}}}} \, \ottnt{f_{{\mathrm{2}}}}  \ottsym{=}  \ottnt{f} \gap
      \ottnt{U_{{\mathrm{12}}}}  \Sub  \ottnt{T_{{\mathrm{12}}}} = T.
    \]
    %% We already have U12 static

  %% T-Case (GV)
  \item[Case] $\casek$ rule: We are given
    \begin{align*}
      &\Gamma  \ottsym{=}   \Gamma'  \circ  \Delta  \gap
      \ottnt{f}  \ottsym{=}   \case{ \ottnt{f'} }{  \br{  \ottmv{l_{\ottmv{j}}} :  \ottmv{x_{\ottmv{j}}} .\,  \ottnt{f_{\ottmv{j}}}   }_{  \ottmv{j}  \in  \ottnt{J}  }  }  \\
      &\Gamma'  \vdash  \ottnt{f'}  \ottsym{:}   \&    \br{  \ottmv{l_{\ottmv{j}}}  :  \ottnt{R_{\ottmv{j}}}  }_{  \ottmv{j}  \in  \ottnt{J}  }   \gap
       ( \Delta  \ottsym{,}  \ottmv{x_{\ottmv{j}}}  \ottsym{:}  \ottnt{R_{\ottmv{j}}}  \vdash  \ottnt{f_{\ottmv{j}}}  \ottsym{:}  \ottnt{T} )_{  \ottmv{j}  \in  \ottnt{J}  } .
    \end{align*}
    Since $\Gamma$ is static,
    $\Gamma'$, $\Delta$ are also static.
    By $\Gamma'  \vdash  \ottnt{f'}  \ottsym{:}   \&    \br{  \ottmv{l_{\ottmv{j}}}  :  \ottnt{R_{\ottmv{j}}}  }_{  \ottmv{j}  \in  \ottnt{J}  }  $ and the IH,
    all $\ottnt{R_{\ottmv{j}}}$ are static and
    there exist static $\mathbb{e}'$, $\ottnt{T'}$ such that
    \[
       \Erase{ \mathbb{e}' }   \ottsym{=}  \ottnt{f'} \gap
      \Gamma'  \vdash  \mathbb{e}'  \ottsym{:}  \ottnt{T'} \gap
      \ottnt{T'}  \Sub   \&    \br{  \ottmv{l_{\ottmv{j}}}  :  \ottnt{R_{\ottmv{j}}}  }_{  \ottmv{j}  \in  \ottnt{J}  }  .
    \]
    %% no need of inversion lemma, here
    %
    %% no need: T' and all Rj are static
    By $\ottnt{T'}  \Sub   \&    \br{  \ottmv{l_{\ottmv{j}}}  :  \ottnt{R_{\ottmv{j}}}  }_{  \ottmv{j}  \in  \ottnt{J}  }  $ and Lemma~\ref{lemma:subtyping-hierarchy},
    $\ottnt{T'}  \lesssim   \&    \br{  \ottmv{l_{\ottmv{j}}}  :  \ottnt{R_{\ottmv{j}}}  }_{  \ottmv{j}  \in  \ottnt{J}  }  $.
    We take some $j \in J$.
    By $\Delta  \ottsym{,}  \ottmv{x_{\ottmv{j}}}  \ottsym{:}  \ottnt{R_{\ottmv{j}}}  \vdash  \ottnt{f_{\ottmv{j}}}  \ottsym{:}  \ottnt{T}$ and the IH,
    $\ottnt{T}$ is static and
    there exist static $\mathbb{e}_{\ottmv{j}}$, $\ottnt{T_{\ottmv{j}}}$ such that
    \[
       \Erase{ \mathbb{e}_{\ottmv{j}} }   \ottsym{=}  \ottnt{f_{\ottmv{j}}} \gap
      \Delta  \ottsym{,}  \ottmv{x_{\ottmv{j}}}  \ottsym{:}  \ottnt{R_{\ottmv{j}}}  \vdash  \mathbb{e}_{\ottmv{j}}  \ottsym{:}  \ottnt{T_{\ottmv{j}}} \gap
      \ottnt{T_{\ottmv{j}}}  \Sub  \ottnt{T}.
    \]
    So, $ ( \ottnt{T_{\ottmv{j}}}  \Sub  \ottnt{T} )_{  \ottmv{j}  \in  \ottnt{J}  } $. % T is an upper bound
    %% join does always exist, by the last clause
    %% We need! : Tj is static
    By Lemma~\ref{lem:static-con-matching}~(4),
    there exist some static $U$ such that $U = \bigjoin\set{T_j}_{j\in J}$.
    % (3c) join U is the _least_ upper bound
    By Lemma~\ref{lem:static-con-matching}~(3), $\ottnt{U}  \Sub  \ottnt{T}$.
    Thus, by
    \[
      \Gamma'  \vdash  \mathbb{e}'  \ottsym{:}  \ottnt{T'} \gap
       ( \Delta  \ottsym{,}  \ottmv{x_{\ottmv{j}}}  \ottsym{:}  \ottnt{R_{\ottmv{j}}}  \vdash  \mathbb{e}_{\ottmv{j}}  \ottsym{:}  \ottnt{T_{\ottmv{j}}} )_{  \ottmv{j}  \in  \ottnt{J}  }  \gap
      \ottnt{T'}  \lesssim   \&    \br{  \ottmv{l_{\ottmv{j}}}  :  \ottnt{R_{\ottmv{j}}}  }_{  \ottmv{j}  \in  \ottnt{J}  }   \gap
      U = \bigjoin\set{T_j}_{j\in J}
    \]
    and the $\casek$ rule,   % T-Case (GGVe)
    $ \Gamma'  \circ  \Delta   \vdash   \case{ \mathbb{e}' }{  \br{  \ottmv{l_{\ottmv{j}}} :  \ottmv{x_{\ottmv{j}}} .\,  \mathbb{e}_{\ottmv{j}}   }_{  \ottmv{j}  \in  \ottnt{J}  }  }   \ottsym{:}  \ottnt{U}$.
    Additionally,
    \[
       \Erase{  \case{ \mathbb{e}' }{  \br{  \ottmv{l_{\ottmv{j}}} :  \ottmv{x_{\ottmv{j}}} .\,  \mathbb{e}_{\ottmv{j}}   }_{  \ottmv{j}  \in  \ottnt{J}  }  }  }  =
       \case{  \Erase{ \mathbb{e}' }  }{  \br{  \ottmv{l_{\ottmv{j}}} :  \ottmv{x_{\ottmv{j}}} .\,   \Erase{ \mathbb{e}_{\ottmv{j}} }    }_{  \ottmv{j}  \in  \ottnt{J}  }  }  =
       \case{ \ottnt{f'} }{  \br{  \ottmv{l_{\ottmv{j}}} :  \ottmv{x_{\ottmv{j}}} .\,  \ottnt{f_{\ottmv{j}}}   }_{  \ottmv{j}  \in  \ottnt{J}  }  }  = f.
    \]
    We already have $\ottnt{U}  \Sub  \ottnt{T}$.
    %% we already have U is static

  %% T-Sub (GV)
  \item[Case] subsumption rule: We are given
    $\Gamma  \vdash  \ottnt{f}  \ottsym{:}  \ottnt{U}$ and $\ottnt{U}  \Sub  \ottnt{T}$.
    %% $\Gamma  \vdash  \ottnt{f}  \ottsym{:}  \ottnt{U}$ and f is GV exp
    By the IH, $U$ is static and
    there exist static $\mathbb{e}$ and $\ottnt{U'}$ such that
    \[
      \Gamma  \vdash  \mathbb{e}  \ottsym{:}  \ottnt{U'} \gap
      \ottnt{U'}  \Sub  \ottnt{U} \gap
       \Erase{ \mathbb{e} }   \ottsym{=}  \ottnt{f}.
    \]
    By $\ottnt{U'}  \Sub  \ottnt{U}$ and $\ottnt{U}  \Sub  \ottnt{T}$ and transitivity, $\ottnt{U'}  \Sub  \ottnt{T}$.
    %% We already have $\Gamma  \vdash  \mathbb{e}  \ottsym{:}  \ottnt{U'}$ and U' is static
  \end{description}
\end{proof}

Proposition~\ref{prop:CI-static} states that
the cast-insertion translation does not insert casts
for static expressions, which can be seen as expressions of \GV{}
if type annotations are removed.
The proof is similar to
that of Theorem~\ref{thm:conservativity-typing}~(1).
\begin{proposition} \label{prop:CI-static}
  Suppose that $\Gamma$ and $\mathbb{e}$ are both static.
  If $\Gamma  \vdash  \mathbb{e}  \ciarrow  \ottnt{f}  \ottsym{:}  \ottnt{T}$,
  then $T$ is static and $ \Erase{ \mathbb{e} }  = \ottnt{f}$.
\end{proposition}
\begin{proof}
  By induction on $\Gamma  \vdash  \mathbb{e}  \ciarrow  \ottnt{f}  \ottsym{:}  \ottnt{T}$.
  with case analysis on the rule applied last.
  We show one of the main cases below.
  \begin{description}
  %% CI-App
  \item[Case] application rule: We are given
    \begin{align*}
      &\Gamma  \ottsym{=}   \Gamma_{{\mathrm{1}}}  \circ  \Gamma_{{\mathrm{2}}}  \gap
      \mathbb{e}  \ottsym{=}  \mathbb{e}_{{\mathrm{1}}} \, \mathbb{e}_{{\mathrm{2}}} \gap
      \ottnt{f}  \ottsym{=}  \ottsym{(}   \ottnt{f_{{\mathrm{1}}}}  :  \ottnt{T_{{\mathrm{1}}}}  \Cast{ \ottnt{p} }_?   \ottnt{T_{{\mathrm{11}}}}   \rightarrow _{ \ottnt{m} }  \ottnt{T_{{\mathrm{12}}}}    \ottsym{)} \, \ottsym{(}   \ottnt{f_{{\mathrm{2}}}}  :  \ottnt{T_{{\mathrm{2}}}}  \Cast{ \ottnt{p} }_?  \ottnt{T_{{\mathrm{11}}}}   \ottsym{)} \gap
      \ottnt{T}  \ottsym{=}  \ottnt{T_{{\mathrm{12}}}} \\
      &\Gamma_{{\mathrm{1}}}  \vdash  \mathbb{e}_{{\mathrm{1}}}  \ciarrow  \ottnt{f_{{\mathrm{1}}}}  \ottsym{:}  \ottnt{T_{{\mathrm{1}}}} \gap
      \Gamma_{{\mathrm{2}}}  \vdash  \mathbb{e}_{{\mathrm{2}}}  \ciarrow  \ottnt{f_{{\mathrm{2}}}}  \ottsym{:}  \ottnt{T_{{\mathrm{2}}}} \gap
       \matching{  \ottnt{T_{{\mathrm{1}}}}  }{   \ottnt{T_{{\mathrm{11}}}}   \rightarrow _{ \ottnt{m} }  \ottnt{T_{{\mathrm{12}}}}   }  \gap
      \ottnt{T_{{\mathrm{2}}}}  \lesssim  \ottnt{T_{{\mathrm{11}}}}.
    \end{align*}
    Since $\Gamma$, $\mathbb{e}$ are static,
    $\Gamma_{{\mathrm{1}}}$, $\Gamma_{{\mathrm{2}}}$, $\mathbb{e}_{{\mathrm{1}}}$, $\mathbb{e}_{{\mathrm{2}}}$
    are also static. % by trivial lemma
    By $\Gamma_{{\mathrm{1}}}  \vdash  \mathbb{e}_{{\mathrm{1}}}  \ciarrow  \ottnt{f_{{\mathrm{1}}}}  \ottsym{:}  \ottnt{T_{{\mathrm{1}}}}$ and the IH,
    $\ottnt{T_{{\mathrm{1}}}}$ is static and $ \Erase{ \mathbb{e}_{{\mathrm{1}}} }   \ottsym{=}  \ottnt{f_{{\mathrm{1}}}}$.
    By $ \matching{  \ottnt{T_{{\mathrm{1}}}}  }{   \ottnt{T_{{\mathrm{11}}}}   \rightarrow _{ \ottnt{m} }  \ottnt{T_{{\mathrm{12}}}}   } $ and Lemma~\ref{lem:static-con-matching}~(2),
    $\ottnt{T_{{\mathrm{11}}}}$, $\ottnt{T_{{\mathrm{12}}}}$ are static and
    $\ottnt{T_{{\mathrm{1}}}}  \ottsym{=}   \ottnt{T_{{\mathrm{11}}}}   \rightarrow _{ \ottnt{m} }  \ottnt{T_{{\mathrm{12}}}} $. % by trivial lemma
    So, \[  \ottnt{f_{{\mathrm{1}}}}  :  \ottnt{T_{{\mathrm{1}}}}  \Cast{ \ottnt{p} }_?   \ottnt{T_{{\mathrm{11}}}}   \rightarrow _{ \ottnt{m} }  \ottnt{T_{{\mathrm{12}}}}   = \ottnt{f_{{\mathrm{1}}}}  \ottsym{=}   \Erase{ \mathbb{e}_{{\mathrm{1}}} } . \]
    By $\Gamma_{{\mathrm{2}}}  \vdash  \mathbb{e}_{{\mathrm{2}}}  \ciarrow  \ottnt{f_{{\mathrm{2}}}}  \ottsym{:}  \ottnt{T_{{\mathrm{2}}}}$ and the IH,
    $\ottnt{T_{{\mathrm{2}}}}$ is static and $ \Erase{ \mathbb{e}_{{\mathrm{2}}} }   \ottsym{=}  \ottnt{f_{{\mathrm{2}}}}$.
    By $\ottnt{T_{{\mathrm{2}}}}  \lesssim  \ottnt{T_{{\mathrm{11}}}}$ and Lemma~\ref{lem:static-con-matching}~(1),
    $\ottnt{T_{{\mathrm{2}}}}  \Sub  \ottnt{T_{{\mathrm{11}}}}$.
    So,
    \[  \ottnt{f_{{\mathrm{2}}}}  :  \ottnt{T_{{\mathrm{2}}}}  \Cast{ \ottnt{p} }_?  \ottnt{T_{{\mathrm{11}}}}  = \ottnt{f_{{\mathrm{2}}}}  \ottsym{=}   \Erase{ \mathbb{e}_{{\mathrm{2}}} } . \]
    Thus,
    $f =  \Erase{ \mathbb{e}_{{\mathrm{1}}} }  \,  \Erase{ \mathbb{e}_{{\mathrm{2}}} }   \ottsym{=}   \Erase{ \mathbb{e}_{{\mathrm{1}}} \, \mathbb{e}_{{\mathrm{2}}} } $.
    %% and already have $\ottnt{T}  \ottsym{=}  \ottnt{T_{{\mathrm{12}}}}$ is static.
  \end{description}
\end{proof}

\subsection{(Failure of) The Gradual Guarantee}

In a gradually typed language, changing type annotations in a program
should not change the static or dynamic behavior---except for run-time
errors caused by casts.  Such an expectation is formalised by Siek et
al.~\cite{Siek-et-al-2015-criteria} as the gradual guarantee property.
It usually consists of two statements concerning the static and
dynamic aspects of programs.  The static counterpart of the gradual
guarantee (simply called the static gradual guarantee) states that
less precise type annotations make the type of an expression less
precise, whereas the dynamic gradual guarantee states that making type
annotations less precise does not change the final outcome of a
program.

We will show that, unfortunately, \GGVe{} satisfies neither the static
nor dynamic gradual guarantee by constructing counterexamples.  We
analyse the problem and argue that it is not easy to recover without
losing other good properties.

First, to capture the notion of programs with more precise type
annotations formally, the precision over types is extended to type
environments and expressions.  The relation $ \Gamma_{{\mathrm{1}}}   \sqsubseteq   \Gamma_{{\mathrm{2}}} $ is the
least relation that satisfies $  \cdot    \sqsubseteq    \cdot  $ and
$ \Gamma_{{\mathrm{1}}}  \ottsym{,}  \ottmv{x}  \ottsym{:}  \ottnt{T_{{\mathrm{1}}}}   \sqsubseteq   \Gamma_{{\mathrm{2}}}  \ottsym{,}  \ottmv{x}  \ottsym{:}  \ottnt{T_{{\mathrm{2}}}} $ if $ \Gamma_{{\mathrm{1}}}   \sqsubseteq   \Gamma_{{\mathrm{2}}} $ and
$\ottnt{T_{{\mathrm{1}}}} \, \sqsubseteq \, \ottnt{T_{{\mathrm{2}}}}$ and the relation $ \mathbb{e}_{{\mathrm{1}}}   \sqsubseteq   \mathbb{e}_{{\mathrm{2}}} $ is the least
precongruence that is closed under the following rules:
\begin{center}
\infrule{
  \ottnt{T_{{\mathrm{1}}}} \, \sqsubseteq \, \ottnt{T_{{\mathrm{2}}}} \andalso
   \mathbb{e}_{{\mathrm{1}}}   \sqsubseteq   \mathbb{e}_{{\mathrm{2}}} 
}{
    \lambda  _ \ottnt{m}   \ottmv{x} {:} \ottnt{T_{{\mathrm{1}}}} .\,  \mathbb{e}_{{\mathrm{1}}}    \sqsubseteq    \lambda  _ \ottnt{m}   \ottmv{x} {:} \ottnt{T_{{\mathrm{2}}}} .\,  \mathbb{e}_{{\mathrm{2}}}  
}
\qquad
\infrule{
  \ottnt{S_{{\mathrm{1}}}} \, \sqsubseteq \, \ottnt{S_{{\mathrm{2}}}}
}{
    \newk\, \ottnt{S_{{\mathrm{1}}}}    \sqsubseteq    \newk\, \ottnt{S_{{\mathrm{2}}}}  
}
\end{center}

Using the precision, the static gradual guarantee can be
stated as follows.

\begin{quote}
  If $ \Gamma_{{\mathrm{1}}}   \sqsubseteq   \Gamma_{{\mathrm{2}}} $, $ \mathbb{e}_{{\mathrm{1}}}   \sqsubseteq   \mathbb{e}_{{\mathrm{2}}} $, and $\Gamma_{{\mathrm{1}}}  \vdash  \mathbb{e}_{{\mathrm{1}}}  \ottsym{:}  \ottnt{T_{{\mathrm{1}}}}$, then
  $\Gamma_{{\mathrm{2}}}  \vdash  \mathbb{e}_{{\mathrm{2}}}  \ottsym{:}  \ottnt{T_{{\mathrm{2}}}}$ and $\ottnt{T_{{\mathrm{1}}}} \, \sqsubseteq \, \ottnt{T_{{\mathrm{2}}}}$ for some $\ottnt{T_{{\mathrm{2}}}}$.
\end{quote}

However, it does not hold:

\begin{theorem}[Failure of the Static Gradual Guarantee]
  There exist $\Gamma_{{\mathrm{1}}}$, $\Gamma_{{\mathrm{2}}}$, $\mathbb{e}_{{\mathrm{1}}}$, $\mathbb{e}_{{\mathrm{2}}}$, and $\ottnt{T_{{\mathrm{1}}}}$
  such that $ \Gamma_{{\mathrm{1}}}   \sqsubseteq   \Gamma_{{\mathrm{2}}} $, $ \mathbb{e}_{{\mathrm{1}}}   \sqsubseteq   \mathbb{e}_{{\mathrm{2}}} $, $\Gamma_{{\mathrm{1}}}  \vdashG  \mathbb{e}_{{\mathrm{1}}}  \ottsym{:}  \ottnt{T_{{\mathrm{1}}}}$
  and, for any $\ottnt{T_{{\mathrm{2}}}}$ such that $\Gamma_{{\mathrm{2}}}  \vdashG  \mathbb{e}_{{\mathrm{2}}}  \ottsym{:}  \ottnt{T_{{\mathrm{2}}}}$,
   $\ottnt{T_{{\mathrm{1}}}} \not \Subn \ottnt{T_{{\mathrm{2}}}}$.
\end{theorem}
\begin{proof}
Let
    \begin{align*}
      \Gamma_{{\mathrm{1}}} &=  x:\ottnt{T_{{\mathrm{1}}}}, y:\ottnt{T_{{\mathrm{2}}}}, z:  \&    \br{  \ottmv{l_{{\mathrm{1}}}} :  \Endpl  ,  \ottmv{l_{{\mathrm{2}}}} :  \Endpl   }   \\
      \Gamma_{{\mathrm{2}}} &=  x: \DYN , y:\ottnt{T_{{\mathrm{2}}}}, z:  \&    \br{  \ottmv{l_{{\mathrm{1}}}} :  \Endpl  ,  \ottmv{l_{{\mathrm{2}}}} :  \Endpl   }   \\
      \mathbb{e}_{{\mathrm{1}}} &= \mathbb{e}_{{\mathrm{2}}} =  \case{ \ottnt{z} }{  \br{  \ottmv{l_{{\mathrm{1}}}} :  \ottmv{x_{{\mathrm{1}}}} .\,    \close{ \ottmv{x_{{\mathrm{1}}}} }   ;  \ottmv{x}  ,\,
               \ottmv{l_{{\mathrm{2}}}} :  \ottmv{x_{{\mathrm{2}}}} .\,    \close{ \ottmv{x_{{\mathrm{2}}}} }   ;  \ottmv{y}   }  }  \\
      \ottnt{T_{{\mathrm{1}}}} &=   \unit    \rightarrow _{  \lin  }   \unit   \\
      \ottnt{T_{{\mathrm{2}}}} &=   \unit    \rightarrow _{  \un  }   \unit  
    \end{align*}
(where $ \mathbb{e}_{{\mathrm{1}}}  ;  \mathbb{e}_{{\mathrm{2}}} $ stands for usual sequential composition).
Then, $ \Gamma_{{\mathrm{1}}}   \sqsubseteq   \Gamma_{{\mathrm{2}}} $, $ \mathbb{e}_{{\mathrm{1}}}   \sqsubseteq   \mathbb{e}_{{\mathrm{2}}} $, $\Gamma_{{\mathrm{1}}}  \vdashG  \mathbb{e}_{{\mathrm{1}}}  \ottsym{:}  \ottnt{T_{{\mathrm{1}}}}$, and
$\Gamma_{{\mathrm{2}}}  \vdashG  \mathbb{e}_{{\mathrm{2}}}  \ottsym{:}  \ottnt{T_{{\mathrm{2}}}}$; but $\ottnt{T_{{\mathrm{1}}}} \not \Subn \ottnt{T_{{\mathrm{2}}}}$.
\end{proof}

\newcommand{\AGTjoin}{\mathrel{\widetilde {\ddot{\vee}}}}

Although we do not state the dynamic gradual guarantee formally, we
expect at least that, if two programs $\mathbb{e}_{{\mathrm{1}}}$ and $\mathbb{e}_{{\mathrm{2}}}$ satisfy
$ \mathbb{e}_{{\mathrm{1}}}   \sqsubseteq   \mathbb{e}_{{\mathrm{2}}} $ and the execution of $\mathbb{e}_{{\mathrm{1}}}$ (after cast insertion)
terminates normally (at $\proc{()}$), then $\mathbb{e}_{{\mathrm{2}}}$ also terminates
normally.  Unfortunately, it would not be very difficult to see such
an expectation fail.  Let's consider
$$\mathbb{e}'_{{\mathrm{1}}} = \ottsym{(}   \lambda  _  \lin    \ottmv{x} {:} \ottnt{T_{{\mathrm{1}}}} .\,  \mathbb{e}_{{\mathrm{1}}}   \ottsym{)} \, \ottsym{(}   \lambda  _  \lin    \ottmv{x_{{\mathrm{1}}}} {:}  \unit  .\,  \ottmv{x_{{\mathrm{1}}}}   \ottsym{)}$$ and a more imprecise
expression $$\mathbb{e}'_{{\mathrm{2}}} = \ottsym{(}   \lambda  _  \lin    \ottmv{x} {:}  \DYN  .\,  \mathbb{e}_{{\mathrm{2}}}   \ottsym{)} \, \ottsym{(}   \lambda  _  \lin    \ottmv{x_{{\mathrm{1}}}} {:}  \unit  .\,  \ottmv{x_{{\mathrm{1}}}}   \ottsym{)}.$$ The
former will return $ \lambda  _  \lin    \ottmv{x_{{\mathrm{1}}}} {:}  \unit  .\,  \ottmv{x_{{\mathrm{1}}}} $ if $\ottmv{l_{{\mathrm{1}}}}$ is selected by
another process.  However, $\mathbb{e}'_{{\mathrm{2}}}$ shows different behavior: the
cast-inserting translation puts a cast from $ \DYN $ to
$\ottnt{T_{{\mathrm{2}}}} =   \unit    \rightarrow _{  \un  }   \unit  $ on $\ottmv{x}$ in the first branch of
$\casek{}$ in $\mathbb{e}_{{\mathrm{2}}}$ but $\ottmv{x}$ will be bound to (a reference to) a
linear function and, if $\ottmv{l_{{\mathrm{1}}}}$ is selected, the cast will fail and
raise blame.

The problem seems to stem from the fact that $ \vee $ has subtle
interaction with $\Subn$.  For typing $\casek$-expressions, we would
naturally require precision to be preserved by the join operation,
i.e., if $\ottnt{T_{{\mathrm{1}}}} \, \sqsubseteq \, \ottnt{T'_{{\mathrm{1}}}}$, then $\ottnt{T_{{\mathrm{1}}}}  \vee  \ottnt{T_{{\mathrm{2}}}} \, \sqsubseteq \, \ottnt{T'_{{\mathrm{1}}}}  \vee  \ottnt{T_{{\mathrm{2}}}}$.
However, the current definition of $ \vee $ breaks this property as
the counterexample to the static gradual guarantee above shows.  Also,
as we can see from the counterexample to the dynamic gradual guarantee,
join with a more precise type can yield a supertype---that is,
$\ottnt{T_{{\mathrm{1}}}} \, \sqsubseteq \,  \DYN $ and $ \DYN   \vee  \ottnt{T_{{\mathrm{2}}}}  \Sub  \ottnt{T_{{\mathrm{1}}}}  \vee  \ottnt{T_{{\mathrm{2}}}}$ hold.

One possible work\-around is to adapt the ``lifted join'' operation
$\AGTjoin$ of the GTFL${}_\lesssim$ language~\cite{Garcia-et-al-2016}
to \GGV.  Like $ \vee $,
$  \unit    \rightarrow _{  \lin  }   \unit   \AGTjoin   \unit    \rightarrow _{  \un  }   \unit   =   \unit    \rightarrow _{  \lin  }   \unit  $ but, unlike $ \vee $,
$ \DYN  \AGTjoin   \unit    \rightarrow _{  \un  }   \unit   =  \DYN $.  Thus, the lifted
join would perhaps recover the gradual guarantee.  However, it seems
that the lifted join is the least upper bound operation for no known
ordering between types and we would lose the minimal type property of
\GGVe{} if we used $\AGTjoin$.  Also, the lifted join has the
following property: $\ottnt{T} \AGTjoin \DYN$ is $\ottnt{T}$ only if $\ottnt{T}$
does not have nontrivial supertypes; otherwise $\ottnt{T} \AGTjoin \DYN$
is $\DYN$.  For example, $ \unit  \AGTjoin \DYN =  \unit $ and
$  \unit    \rightarrow _{  \lin  }   \unit   \AGTjoin \DYN =   \unit    \rightarrow _{  \lin  }   \unit  $ but
$  \unit    \rightarrow _{  \un  }   \unit   \AGTjoin \DYN = \DYN$ (because
$  \unit    \rightarrow _{  \un  }   \unit   \Sub   \unit    \rightarrow _{  \lin  }   \unit  $).  It means that the
standard narrowing property---if $\Gamma  \ottsym{,}  \ottmv{x}  \ottsym{:}  \ottnt{T_{{\mathrm{1}}}}  \vdashG  \mathbb{e}  \ottsym{:}  \ottnt{T}$ and
$\ottnt{T_{{\mathrm{2}}}}  \Sub  \ottnt{T_{{\mathrm{1}}}}$, then $\Gamma_{{\mathrm{1}}}  \ottsym{,}  \ottmv{x}  \ottsym{:}  \ottnt{T_{{\mathrm{2}}}}  \vdashG  \mathbb{e}  \ottsym{:}  \ottnt{T}$---does not hold.  We
leave more detailed analysis of the problem and possible remedy for
future work.

\section{Related Work}
\label{sec:related}

\subsection{Gradual Typing}
\label{sec:related-gradual-typing}

Findler and Felleisen \shortcite{Findler-Felleisen-2002} introduced two seminal ideas:
higher-order \emph{contracts} that dynamically monitor conformance to
a type discipline, and \emph{blame} to indicate whether it is the
library or the client which is at fault if the contract is violated.
Siek and Taha \shortcite{Siek-Taha-2006,DBLP:conf/ecoop/SiekT07} introduced gradual types to integrate
untyped and typed code, while Flanagan \shortcite{Flanagan-2006} introduced
hybrid types to integrate simple types with refinement types.  Both
used target languages with explicit casts and similar translations
from source to target; both exploit contracts, but neither allocates
blame.  Motivated by similarities between gradual and hybrid types,
Wadler and Findler \shortcite{Wadler-Findler-2009} introduced blame calculus, which unifies
the two by encompassing untyped, simply-typed, and refinement-typed
code. As the name indicates, it also restores blame, which enables a
proof of \emph{blame safety}: blame for type errors always lays with
less-precisely typed code---``well-typed programs can't be blamed''.

While the first investigations of gradual typing were based on 
simply-typed calculi, subsequent work has explored gradual
typing for a range of typing features. Polymorphism
\cite{DBLP:conf/popl/AhmedFSW11,DBLP:journals/pacmpl/IgarashiSI17,DBLP:journals/pacmpl/ToroLT19} has
proved to be quite tricky, with one important question about
the Jack-of-All-Trades Principle \cite{DBLP:conf/popl/AhmedFSW11} still open.
A gradual treatment of record types may be found in the paper on
Abstract Gradual Typing (AGT)~\cite{Garcia-et-al-2016}. Variant types have
proved elusive, but union types have been considered
\cite{DBLP:conf/birthday/SiekT16} along with intersection types and polymorphism as
part of a set-theoretical reevaluation of gradual principles
\cite{DBLP:journals/pacmpl/CastagnaLPS19}.

Moving towards session types, systems with gradual typestate have been
considered \cite{Wolf-et-al-2011,DBLP:journals/toplas/GarciaTWA14},
they extend an object-oriented 
language with typestate by a dynamic type and define a suitable
translation to an internal language with casts. The additional
complication is to track the current typestate at run time. 
Thiemann \shortcite{DBLP:conf/tgc/Thiemann14} describes a system with gradual types and
session types, but in it only types (and not session types) can be
gradual.

Effect systems have been gradualized based on ideas from abstract
interpretation  by Banados Schwerter and others
\shortcite{Banados-et-al-2014}. While the former work only presented a
gradualization of effects themselves a subsequent extension adds a
full treatment of types \cite{DBLP:journals/jfp/SchwerterGT16}.
Related ideas are explored by Thiemann and Fennell
\cite{DBLP:conf/esop/ThiemannF14} who present an approach to
gradualize annotated type systems, like units and security
labels. Following an earlier approach for gradual security typing for
simply-typed lambda calculus \cite{DisneyFlanagan2011}, Thiemann and
Fennell \shortcite{DBLP:conf/csfw/FennellT13} developed gradual
security for an ML core language with references and subsequently for
a Java core language LJGS with polymorphic security labels \cite{DBLP:conf/ecoop/FennellT16}. Toro and
others \shortcite{DBLP:journals/toplas/ToroGT18} developed a gradual
calculus with slightly different features from first principles using
the AGT \cite{Garcia-et-al-2016} approach. In each
of these approaches, special measures have to be taken to ensure the
key property of non-interference.
Gradual type systems related to session types also include
the run-time enforcement of affine typing of  Tov and Pucella \shortcite{Tov-Pucella-2010}.

As noted in the introduction, gradual typing may be 
important as a bridge to type systems that go beyond what is currently
available, including dependent, effect, and session types.
There is a range of gradual type systems for dependent types.
Ou and others \shortcite{Ou-et-al-2004} bridge the gap between
simply-typed lambda calculus and a calculus with indexed types.
In Flanagan's hybrid typing \cite{Flanagan-2006} subtyping judgments are
either proved or disproved statically by SMT theorem proving or
residualized as run-time checks.
Greenberg and others \shortcite{Greenberg-et-al-2010} consider
different styles of contracts in simply-typed and dependently-typed settings. 
Lehmann and Tanter
\shortcite{DBLP:conf/popl/LehmannT17} present an approach that uses
the AGT methodology to obtain a gradual system that mediates between
simple types and dependent refinement types. This work has been
augmented with type inference by Vazou and others
\shortcite{DBLP:journals/pacmpl/VazouTH18} and it has been extend to verification
\cite{DBLP:conf/vmcai/BaderAT18} where specifications may contains
unknown subformulas. Jafery and Dunfield
\cite{DBLP:conf/popl/JaferyD17} consider gradualized refinements for
sum types with the goal to control errors in pattern matching.

Gradual ownership types \cite{DBLP:conf/esop/SergeyC12} is a
gradualization of the Owners as Dominators principle of ownership. Its
theory is built with similar principles as other gradual languages,
but its flavor is different as ownership is not a semantic property,
but a structure imposed by the programmer.

Siek and others \shortcite{Siek-et-al-2015-criteria} review desirable properties of
gradually-typed languages, while Wadler \shortcite{Wadler-2015} discusses
history of the blame calculus and why blame is important.
These papers provide overviews of the field, each with many further
citations. Many of the above-cited works strive to fulfill the
properties of Siek and others, not all of them are successful, but
further discussion of the properties exceeds the scope of this survey
of related work.

TypeScript~TPD \citep{Williams-et-al-2017} applies contracts to
monitor the gradual typing of TypeScript, and evaluates the successes
and shortcomings of contracts in this context.

\subsection{Session Types}
\label{sec:related-session-types}

Session types were introduced by Honda, Vasconcelos, and Kubo
\shortcite{Honda-1993,Honda-et-al-1998}. The original system addressed
binary sessions, whereby types describe the interaction between two
partners. Binary sessions were eventually extended to the more general
setting of multiparty session types \citep{Honda-et-al-2016}.
Recent years have seen the introduction of session types in
programming languages, and software development tools. We review the
most important works.

Session types inspired the design of several programming languages.
Sing\# \cite{Fahndrich-et-al-2006} constitutes one of the first
attempts to introduce session types in programming languages. An
extension of C, Sing\# was used to implement Singularity, an operating
system based on message passing.
Gay and others \shortcite{Gay-et-al-2010} propose attaching session types to class
definitions, allowing to treat channels as objects for session-based
communication in distributed systems.
SePi \citep{franco.vasconcelos:sepi} is a concurrent, message-passing
programming language based on the pi-calculus, featuring a simple form
of refinement types.
SILL \citep{Toninho-et-al-2013,Pfenning-Griffith-2015} is a
higher-order session functional programming language, featuring
process expressions as first class objects via a linear contextual
monad.
Concurrent C0 \citep{Willsey-et-al-2016} is a type-safe C-like
programming language equipped with channel communication governed by
session types.
Links \citep{Lindley-Morris-2017} is a functional programming language
designed for tierless web applications that natively
supports binary session types.

Proposals have been made to retroactively introduce session types in
mainstream programming languages.
Session Java \citep{Hu-et-al-2008} introduces API-based session
primitives in Java, while \citep{Hu-et-al-2010} presents a Java
language extension and type discipline for session-based event-driven
programming.
Featherweight Erlang \citep{Mostrous-Vasconcelos-2011} imposes a
session-based type system to discipline message passing in Erlang.
Mungo \citep{Kouzapas-et-al-2016} is a tool for checking Java
code against session types, presented in the form of typestates.
Embedding of session types have
been proposed for Haskell
\citep{Orchard-Yoshida-2016,Pucella-Tov-2008,Sackman-Eisenbach-2008,Polakow-2015,Lindley-Morris-2016},
OCaml \citep{Padovani-2017}, Scala \citep{Scalas-Yoshida-2016}, and
Rust \citep{Jespersen-et-al-2015}.
Most of these embeddings delegate linearity checks to the run-time system.

Session types can be used in the software development process under
different forms, including languages to describe protocols,
specialised libraries to invoke session-based communication
primitives, provision for run-time monitoring against session types,
and extended type checkers.
Scribble~\citep{scribble} is a language-agnostic protocol description
formalism used in many different tools.
Multiparty Session C \citep{Ng-et-al-2012} uses Scribble, a
compiler plug-in, and a C library to validate against
session types.
Hu and Yoshida \shortcite{Hu-Yoshida-2016} generate protocol-specific Java APIs from
multiparty session types described in Scribble.
SPY \citep{Neykova-et-al-2013} generates run-time monitors for endpoint
communication from Scribble protocols.
Neykova and Yoshida \shortcite{Neykova-Yoshida-2014} designed and implemented a session actor
library in Python together with a run-time verification mechanism.
Bocchi and others \shortcite{Bocchi-et-al-2013} present a theory that incorporates both
static typing and dynamic monitoring of session types.
Fowler \shortcite{Fowler-2016b} describes a framework for monitoring Erlang
applications against multiparty session types.
Neykova and Yoshida \shortcite{Neykova-Yoshida-2017} investigate failure handling for
Erlang processes in a system that dynamically monitors session
types.

\section{Conclusions}
\label{sec:conclusions}

We presented the design of \GGV,
which combines a session-typed language \GV\
along the lines of Gay and Vasconcelos \shortcite{Gay-Vasconcelos-2010}
with a blame calculus along the lines of Wadler and Findler \shortcite{Wadler-Findler-2009},
and with dynamic enforcement of linearity along the lines of Tov and Pucella
\shortcite{Tov-Pucella-2010}.
We established expected results for such
a language, including type safety and blame safety.
Although the gradual guarantee does not hold, it seems that it is not
clear how it can be recovered without losing other good properties.

Much remains to be done; we consider just one future direction here.
The embedding of linear types in the unrestricted dynamic type relies
on an indirection through a cell in the store. In our
present work, these cells are used once and then discarded. This
\emph{one-shot policy} imposes a certain usage pattern on linear
values embedded in the unityped language. In particular, the send and receive operations on a channel
need to be chained as in $(\close{(\send{v_2}{(\send{v_1}{c})})})$.
However, one could imagine a unityped language where one
may use the channel non-linearly in an imperative style as in $(\send{v_1}{c};
\send{v_2}{c};\close{c})$, mimicking the style of network programming in
conventional languages. This style can also be supported by a variant
of {\GGV} with a
\emph{multi-shot policy} that restores an updated channel to the same
cell from which it was extracted. We leave the full formalisation of
this policy to future work.

\paragraph{Acknowledgments}
We would like to thank Alceste Scalas and Nobuko Yoshida for comments
and pointing out errors in the definition of subtyping rules, Kaede
Kobayashi for pointing out subtle errors in the operational semantics,
and anonymous reviewers for constructive comments.  We are also
grateful to Hannes Saffrich for implementing a type checker for the
calculus in this paper. This work was supported in part by the JSPS
KAKENHI Grant Number JP17H01723 (Igarashi), by FCT through the LASIGE
Research Unit ref.\ UID/CEC/00408/2019 and project Confident ref.\
PTDC/EEI-CTP/4503/2014 (Vasconcelos), and by EPSRC programme grant
EP/K034413/1 (Wadler).

\bibliographystyle{jfp}
\bibliography{biblio}

\begin{thebibliography}{}

\bibitem[\protect\citename{Ahmed {\em et~al.}\relax,
  }2011]{DBLP:conf/popl/AhmedFSW11}
Ahmed, Amal, Findler, Robert~Bruce, Siek, Jeremy~G., \& Wadler, Philip. (2011).
\newblock Blame for all.
\newblock {\em Pages  201--214 of:} Ball, Thomas, \& Sagiv, Mooly (eds), {\em
  Proceedings of the 38th {ACM} {SIGPLAN-SIGACT} Symposium on Principles of
  Programming Languages, {POPL} 2011, Austin, TX, USA, January 26-28, 2011}.
\newblock {ACM}.

\bibitem[\protect\citename{Bader {\em et~al.}\relax,
  }2018]{DBLP:conf/vmcai/BaderAT18}
Bader, Johannes, Aldrich, Jonathan, \& Tanter, {\'{E}}ric. (2018).
\newblock Gradual program verification.
\newblock {\em Pages  25--46 of:} Dillig, Isil, \& Palsberg, Jens (eds), {\em
  Verification, Model Checking, and Abstract Interpretation - 19th
  International Conference, {VMCAI} 2018, Los Angeles, CA, USA, January 7-9,
  2018, Proceedings}.
\newblock Lecture Notes in Computer Science, vol. 10747.
\newblock Springer.

\bibitem[\protect\citename{Ba{\~n}ados~Schwerter {\em et~al.}\relax,
  }2014]{Banados-et-al-2014}
Ba{\~n}ados~Schwerter, Felipe, Garcia, Ronald, \& Tanter, {\'E}ric. (2014).
\newblock A theory of gradual effect systems.
\newblock {\em Pages  283--295 of:} {\em International Conference on Functional
  Programming (ICFP)}.
\newblock ACM.

\bibitem[\protect\citename{Barendregt, }1984]{barendregt:lambda-calculus}
Barendregt, H.P.\. (1984).
\newblock {\em {The Lambda Calculus: Its Syntax and Semantics}}.
\newblock North-Holland.

\bibitem[\protect\citename{Bierman {\em et~al.}\relax,
  }2014]{Bierman-et-al-2014}
Bierman, Gavin, Abadi, Mart{\'\i}n, \& Torgersen, Mads. (2014).
\newblock Understanding {T}ype{S}cript.
\newblock {\em Pages  257--281 of:} {\em European Conference on Object-Oriented
  Programming (ECOOP)}.
\newblock LNCS, vol. 8586.
\newblock Springer.

\bibitem[\protect\citename{Bierman {\em et~al.}\relax,
  }2010]{Bierman-et-al-2010}
Bierman, Gavin~M., Meijer, Erik, \& Torgersen, Mads. (2010).
\newblock Adding dynamic types to {C\#}.
\newblock {\em Pages  76--100 of:} {\em European Conference on Object-Oriented
  Programming (ECOOP)}.
\newblock LNCS.
\newblock Springer.

\bibitem[\protect\citename{Bocchi {\em et~al.}\relax,
  }2013]{Bocchi-et-al-2013-conference}
Bocchi, Laura, Chen, Tzu-Chun, Demangeon, Romain, Honda, Kohei, \& Yoshida,
  Nobuko. (2013).
\newblock Monitoring networks through multiparty session types.
\newblock {\em Pages  50--65 of:} {\em Formal Techniques for Distributed
  Systems}.
\newblock Springer.

\bibitem[\protect\citename{Bocchi {\em et~al.}\relax, }2017]{Bocchi-et-al-2013}
Bocchi, Laura, Chen, Tzu-Chun, Demangeon, Romain, Honda, Kohei, \& Yoshida,
  Nobuko. (2017).
\newblock Monitoring networks through multiparty session types.
\newblock {\em Theoretical computer science}, {\bf 699}, 33--58.

\bibitem[\protect\citename{Brady, }2013]{Brady-2013}
Brady, Edwin. (2013).
\newblock {Idris}, a general-purpose dependently typed programming language:
  Design and implementation.
\newblock {\em Journal of functional programming}, {\bf 23}(05), 552--593.

\bibitem[\protect\citename{Caires \& Pfenning, }2010]{Caires-Pfenning-2010}
Caires, Lu{\'{\i}}s, \& Pfenning, Frank. (2010).
\newblock Session types as intuitionistic linear propositions.
\newblock {\em Pages  222--236 of:} {\em International Conference on
  Concurrency Theory (CONCUR)}.
\newblock LNCS.
\newblock Springer.

\bibitem[\protect\citename{Caires {\em et~al.}\relax, }2014]{Caires-et-al-2014}
Caires, Luis, Pfenning, Frank, \& Toninho, Bernardo. (2014).
\newblock Linear logic propositions as session types.
\newblock {\em Mathematical structures in computer science}, {\bf 26}(03),
  367--423.

\bibitem[\protect\citename{Castagna {\em et~al.}\relax,
  }2019]{DBLP:journals/pacmpl/CastagnaLPS19}
Castagna, Giuseppe, Lanvin, Victor, Petrucciani, Tommaso, \& Siek, Jeremy~G.
  (2019).
\newblock Gradual typing: A new perspective.
\newblock {\em {PACMPL}}, {\bf 3}({POPL}), 16:1--16:32.

\bibitem[\protect\citename{Chaudhuri {\em et~al.}\relax,
  }2017]{DBLP:journals/pacmpl/ChaudhuriVGRL17}
Chaudhuri, Avik, Vekris, Panagiotis, Goldman, Sam, Roch, Marshall, \& Levi,
  Gabriel. (2017).
\newblock Fast and precise type checking for {JavaScript}.
\newblock {\em {PACMPL}}, {\bf 1}({OOPSLA}), 48:1--48:30.

\bibitem[\protect\citename{Cimini \& Siek, }2016]{DBLP:conf/popl/CiminiS16}
Cimini, Matteo, \& Siek, Jeremy~G. (2016).
\newblock The {Gradualizer}: A methodology and algorithm for generating gradual
  type systems.
\newblock {\em Pages  443--455 of:} {\em Proceedings of the 43rd Annual {ACM}
  {SIGPLAN-SIGACT} Symposium on Principles of Programming Languages, {POPL}
  2016, St. Petersburg, FL, USA, January 20 - 22, 2016}.

\bibitem[\protect\citename{Cooper {\em et~al.}\relax, }2007]{Cooper-et-al-2007}
Cooper, Ezra, Lindley, Sam, Wadler, Philip, \& Yallop, Jeremy. (2007).
\newblock {Links}: Web programming without tiers.
\newblock {\em Pages  266--296 of:} {\em Formal Methods for Components and
  Objects}.
\newblock Springer.

\bibitem[\protect\citename{Demangeon {\em et~al.}\relax,
  }2015]{DBLP:journals/fmsd/DemangeonHHNY15}
Demangeon, Romain, Honda, Kohei, Hu, Raymond, Neykova, Rumyana, \& Yoshida,
  Nobuko. (2015).
\newblock Practical interruptible conversations: distributed dynamic
  verification with multiparty session types and python.
\newblock {\em Formal methods in system design}, {\bf 46}(3), 197--225.

\bibitem[\protect\citename{Disney \& Flanagan, }2011]{DisneyFlanagan2011}
Disney, Tim, \& Flanagan, Cormac. (2011).
\newblock Gradual information flow typing.
\newblock  {\em Workshop on From Scripts to Program (STOP)}.

\bibitem[\protect\citename{Ernst {\em et~al.}\relax,
  }2017]{DBLP:journals/scp/ErnstMSS17}
Ernst, Erik, M{\o}ller, Anders, Schwarz, Mathias, \& Strocco, Fabio. (2017).
\newblock Message safety in {Dart}.
\newblock {\em Sci. comput. program.}, {\bf 133}, 51--73.

\bibitem[\protect\citename{F{\"a}hndrich {\em et~al.}\relax,
  }2006]{Fahndrich-et-al-2006}
F{\"a}hndrich, Manuel, Aiken, Mark, Hawblitzel, Chris, Hodson, Orion, Hunt,
  Galen~C., Larus, James~R., \& Levi, Steven. (2006).
\newblock Language support for fast and reliable message-based communication in
  {S}ingularity {OS}.
\newblock {\em Pages  177--190 of:} {\em European Conference on Computer
  Systems (EuroSys)}.
\newblock ACM.

\bibitem[\protect\citename{Fennell \& Thiemann,
  }2012]{DBLP:conf/sfp/FennellT12}
Fennell, Luminous, \& Thiemann, Peter. (2012).
\newblock The blame theorem for a linear lambda calculus with type dynamic.
\newblock {\em Pages  37--52 of:} {\em Trends in Functional Programming}.
\newblock LNCS, vol. 7829.
\newblock Springer.

\bibitem[\protect\citename{Fennell \& Thiemann,
  }2013]{DBLP:conf/csfw/FennellT13}
Fennell, Luminous, \& Thiemann, Peter. (2013).
\newblock Gradual security typing with references.
\newblock {\em Pages  224--239 of:} {\em 2013 {IEEE} 26th Computer Security
  Foundations Symposium, New Orleans, LA, USA, June 26-28, 2013}.
\newblock {IEEE} Computer Society.

\bibitem[\protect\citename{Fennell \& Thiemann,
  }2016]{DBLP:conf/ecoop/FennellT16}
Fennell, Luminous, \& Thiemann, Peter. (2016).
\newblock {LJGS:} gradual security types for object-oriented languages.
\newblock {\em Pages  9:1--9:26 of:} Krishnamurthi, Shriram, \& Lerner,
  Benjamin~S. (eds), {\em 30th European Conference on Object-Oriented
  Programming, {ECOOP} 2016, July 18-22, 2016, Rome, Italy}.
\newblock LIPIcs, vol. 56.
\newblock Schloss Dagstuhl - Leibniz-Zentrum fuer Informatik.

\bibitem[\protect\citename{Findler \& Felleisen, }2002]{Findler-Felleisen-2002}
Findler, Robert~Bruce, \& Felleisen, Matthias. (2002).
\newblock Contracts for higher-order functions.
\newblock {\em Pages  48--59 of:} {\em International Conference on Functional
  Programming (ICFP)}.
\newblock ACM.

\bibitem[\protect\citename{Flanagan, }2006]{Flanagan-2006}
Flanagan, Cormac. (2006).
\newblock Hybrid type checking.
\newblock {\em Pages  245--256 of:} {\em Principles of Programming Languages
  (POPL)}.
\newblock ACM.

\bibitem[\protect\citename{Fowler, }2016]{Fowler-2016b}
Fowler, Simon. (2016).
\newblock An {Erlang} implementation of multiparty session actors.
\newblock {\em Pages  36--50 of:} {\em Interaction and Concurrency Experience}.

\bibitem[\protect\citename{Franco \& Vasconcelos,
  }2013]{franco.vasconcelos:sepi}
Franco, Juliana, \& Vasconcelos, Vasco~Thudichum. (2013).
\newblock A concurrent programming language with refined session types.
\newblock {\em Pages  15--28 of:} {\em SEFM}.
\newblock LNCS, vol. 8368.
\newblock Springer.

\bibitem[\protect\citename{Garcia {\em et~al.}\relax,
  }2014]{DBLP:journals/toplas/GarciaTWA14}
Garcia, Ronald, Tanter, {\'{E}}ric, Wolff, Roger, \& Aldrich, Jonathan. (2014).
\newblock Foundations of typestate-oriented programming.
\newblock {\em {ACM} trans. program. lang. syst.}, {\bf 36}(4), 12:1--12:44.

\bibitem[\protect\citename{Garcia {\em et~al.}\relax, }2016]{Garcia-et-al-2016}
Garcia, Ronald, Clark, Alison~M., \& Tanter, {\'{E}}ric. (2016).
\newblock Abstracting gradual typing.
\newblock {\em Pages  429--442 of:} {\em Principles of Programming Languages
  (POPL)}.
\newblock ACM.

\bibitem[\protect\citename{Gay \& Hole, }2005]{Gay-Hole-2005}
Gay, Simon, \& Hole, Malcolm. (2005).
\newblock Subtyping for session types in the pi calculus.
\newblock {\em Acta informatica}, {\bf 42}(2-3), 191--225.

\bibitem[\protect\citename{Gay \& Vasconcelos, }2010]{Gay-Vasconcelos-2010}
Gay, Simon, \& Vasconcelos, Vasco. (2010).
\newblock Linear type theory for asynchronous session types.
\newblock {\em Journal of functional programming}, {\bf 20}(01), 19--50.

\bibitem[\protect\citename{Gay {\em et~al.}\relax, }2010]{Gay-et-al-2010}
Gay, Simon~J., Vasconcelos, Vasco~Thudichum, Ravara, Ant{\'{o}}nio, Gesbert,
  Nils, \& Caldeira, Alexandre~Z. (2010).
\newblock Modular session types for distributed object-oriented programming.
\newblock {\em Pages  299--312 of:} {\em Principles of Programming Languages
  (POPL)}.
\newblock ACM.

\bibitem[\protect\citename{Girard, }1987]{Girard-1987}
Girard, Jean-Yves. (1987).
\newblock Linear logic.
\newblock {\em Theoretical computer science}, {\bf 50}(1), 1--101.

\bibitem[\protect\citename{Gommerstadt {\em et~al.}\relax,
  }2018]{DBLP:conf/esop/GommerstadtJP18}
Gommerstadt, Hannah, Jia, Limin, \& Pfenning, Frank. (2018).
\newblock Session-typed concurrent contracts.
\newblock {\em Pages  771--798 of:} Ahmed, Amal (ed), {\em Programming
  Languages and Systems - 27th European Symposium on Programming, {ESOP} 2018,
  Held as Part of the European Joint Conferences on Theory and Practice of
  Software, {ETAPS} 2018, Thessaloniki, Greece, April 14-20, 2018,
  Proceedings}.
\newblock Lecture Notes in Computer Science, vol. 10801.
\newblock Springer.

\bibitem[\protect\citename{Greenberg {\em et~al.}\relax,
  }2010]{Greenberg-et-al-2010}
Greenberg, Michael, Pierce, Benjamin~C., \& Weirich, Stephanie. (2010).
\newblock Contracts made manifest.
\newblock {\em Pages  353--364 of:} {\em Principles of Programming Languages
  (POPL)}.
\newblock ACM.

\bibitem[\protect\citename{Honda {\em et~al.}\relax, }2011]{scribble}
Honda, K., Mukhamedov, A., Brown, G., Chen, T., \& Yoshida, N. (2011).
\newblock Scribbling interactions with a formal foundation.
\newblock {\em Pages  55--75 of:} {\em ICDCIT}.
\newblock LNCS, vol. 6536.
\newblock Springer.

\bibitem[\protect\citename{Honda, }1993]{Honda-1993}
Honda, Kohei. (1993).
\newblock Types for dyadic interaction.
\newblock {\em Pages  509--523 of:} {\em International Conference on
  Concurrency Theory (CONCUR)}.
\newblock LNCS, vol. 715.
\newblock Springer.

\bibitem[\protect\citename{Honda {\em et~al.}\relax, }1998]{Honda-et-al-1998}
Honda, Kohei, Vasconcelos, Vasco, \& Kubo, Makoto. (1998).
\newblock Language primitives and type discipline for structured
  communication-based programming.
\newblock {\em Pages  122--138 of:} {\em European Symposium on Programming
  (ESOP)}.
\newblock LNCS.
\newblock Springer.

\bibitem[\protect\citename{Honda {\em et~al.}\relax, }2008]{Honda-et-al-2008}
Honda, Kohei, Yoshida, Nobuko, \& Carbone, Marco. (2008).
\newblock Multiparty asynchronous session types.
\newblock {\em Pages  273--284 of:} {\em Principles of Programming Languages
  (POPL)}.
\newblock ACM.

\bibitem[\protect\citename{Honda {\em et~al.}\relax, }2016]{Honda-et-al-2016}
Honda, Kohei, Yoshida, Nobuko, \& Carbone, Marco. (2016).
\newblock Multiparty asynchronous session types.
\newblock {\em Journal of the acm}, {\bf 63}(1), 9.

\bibitem[\protect\citename{Hu \& Yoshida, }2016]{Hu-Yoshida-2016}
Hu, Raymond, \& Yoshida, Nobuko. (2016).
\newblock Hybrid session verification through endpoint {API} generation.
\newblock {\em Pages  401--418 of:} {\em Fundamental Approaches to Software
  Engineering (FASE)}.
\newblock LNCS, vol. 9633.
\newblock Springer.

\bibitem[\protect\citename{Hu {\em et~al.}\relax, }2008]{Hu-et-al-2008}
Hu, Raymond, Yoshida, Nobuko, \& Honda, Kohei. (2008).
\newblock Session-based distributed programming in {Java}.
\newblock {\em Pages  516--541 of:} {\em European Conference on Object-Oriented
  Programming (ECOOP)}.
\newblock LNCS, vol. 5142.
\newblock Springer.

\bibitem[\protect\citename{Hu {\em et~al.}\relax, }2010]{Hu-et-al-2010}
Hu, Raymond, Kouzapas, Dimitrios, Pernet, Olivier, Yoshida, Nobuko, \& Honda,
  Kohei. (2010).
\newblock Type-safe eventful sessions in {Java}.
\newblock {\em Pages  329--353 of:} {\em European Conference on Object-Oriented
  Programming (ECOOP)}.
\newblock LNCS, vol. 6183.
\newblock Springer.

\bibitem[\protect\citename{Igarashi {\em et~al.}\relax,
  }2017a]{DBLP:journals/pacmpl/IgarashiTVW17}
Igarashi, Atsushi, Thiemann, Peter, Vasconcelos, Vasco~T., \& Wadler, Philip.
  (2017a).
\newblock Gradual session types.
\newblock {\em {PACMPL}}, {\bf 1}({ICFP}), 38:1--38:28.

\bibitem[\protect\citename{Igarashi {\em et~al.}\relax,
  }2017b]{DBLP:journals/pacmpl/IgarashiSI17}
Igarashi, Yuu, Sekiyama, Taro, \& Igarashi, Atsushi. (2017b).
\newblock On polymorphic gradual typing.
\newblock {\em {PACMPL}}, {\bf 1}({ICFP}), 40:1--40:29.

\bibitem[\protect\citename{Jafery \& Dunfield, }2017]{DBLP:conf/popl/JaferyD17}
Jafery, Khurram~A., \& Dunfield, Joshua. (2017).
\newblock Sums of uncertainty: Refinements go gradual.
\newblock {\em Pages  804--817 of:} Castagna, Giuseppe, \& Gordon, Andrew~D.
  (eds), {\em Proceedings of the 44th {ACM} {SIGPLAN} Symposium on Principles
  of Programming Languages, {POPL} 2017, Paris, France, January 18-20, 2017}.
\newblock {ACM}.

\bibitem[\protect\citename{Jespersen {\em et~al.}\relax,
  }2015]{Jespersen-et-al-2015}
Jespersen, Thomas Bracht~Laumann, Munksgaard, Philip, \& Larsen, Ken~Friis.
  (2015).
\newblock Session types for {Rust}.
\newblock {\em Pages  13--22 of:} {\em Workshop on Generic Programming (WGP)}.
\newblock ACM.

\bibitem[\protect\citename{Jia {\em et~al.}\relax,
  }2016]{DBLP:conf/popl/JiaGP16}
Jia, Limin, Gommerstadt, Hannah, \& Pfenning, Frank. (2016).
\newblock Monitors and blame assignment for higher-order session types.
\newblock {\em Pages  582--594 of:} Bod{\'{\i}}k, Rastislav, \& Majumdar, Rupak
  (eds), {\em Proceedings of the 43rd Annual {ACM} {SIGPLAN-SIGACT} Symposium
  on Principles of Programming Languages, {POPL} 2016, St. Petersburg, FL, USA,
  January 20 - 22, 2016}.
\newblock {ACM}.

\bibitem[\protect\citename{Kobayashi {\em et~al.}\relax,
  }1999]{DBLP:journals/toplas/KobayashiPT99}
Kobayashi, Naoki, Pierce, Benjamin~C., \& Turner, David~N. (1999).
\newblock Linearity and the pi-calculus.
\newblock {\em {ACM} trans. program. lang. syst.}, {\bf 21}(5), 914--947.

\bibitem[\protect\citename{Kouzapas {\em et~al.}\relax,
  }2016]{Kouzapas-et-al-2016}
Kouzapas, Dimitrios, Dardha, Ornela, Perera, Roly, \& Gay, Simon~J. (2016).
\newblock Typechecking protocols with {Mungo} and {StMungo}.
\newblock {\em Pages  146--159 of:} {\em Principles and Practice of Declarative
  Programming (PPDP)}.
\newblock ACM.

\bibitem[\protect\citename{Lehmann \& Tanter, }2017]{DBLP:conf/popl/LehmannT17}
Lehmann, Nico, \& Tanter, {\'{E}}ric. (2017).
\newblock Gradual refinement types.
\newblock {\em Pages  775--788 of:} Castagna, Giuseppe, \& Gordon, Andrew~D.
  (eds), {\em Proceedings of the 44th {ACM} {SIGPLAN} Symposium on Principles
  of Programming Languages, {POPL} 2017, Paris, France, January 18-20, 2017}.
\newblock {ACM}.

\bibitem[\protect\citename{Lindley \& Morris, }2016a]{Lindley-Morris-2016}
Lindley, Sam, \& Morris, J.~Garrett. (2016a).
\newblock Embedding session types in {Haskell}.
\newblock {\em Pages  133--145 of:} {\em Symposium on Haskell}.
\newblock ACM.

\bibitem[\protect\citename{Lindley \& Morris, }2016b]{Lindley-and-Morris-2016}
Lindley, Sam, \& Morris, J~Garrett. (2016b).
\newblock Talking bananas: Structural recursion for session types.
\newblock {\em Pages  434--447 of:} {\em International Conference on Functional
  Programming (ICFP)}.
\newblock ACM.

\bibitem[\protect\citename{Lindley \& Morris, }2017]{Lindley-Morris-2017}
Lindley, Sam, \& Morris, J.~Garrett. (2017).
\newblock {\em Behavioural types: from theory to tools}.
\newblock River Publishers.
\newblock Chap. Lightweight functional session types.

\bibitem[\protect\citename{Melgratti \& Padovani,
  }2017]{DBLP:journals/pacmpl/MelgrattiP17}
Melgratti, Hern{\'{a}}n~C., \& Padovani, Luca. (2017).
\newblock Chaperone contracts for higher-order sessions.
\newblock {\em {PACMPL}}, {\bf 1}({ICFP}), 35:1--35:29.

\bibitem[\protect\citename{Milner {\em et~al.}\relax, }1992]{Milner-et-al-1992}
Milner, Robin, Parrow, Joachim, \& Walker, David. (1992).
\newblock A calculus of mobile processes, {I}.
\newblock {\em Information and computation}, {\bf 100}(1), 1--40.

\bibitem[\protect\citename{Mostrous \& Vasconcelos,
  }2011]{Mostrous-Vasconcelos-2011}
Mostrous, Dimitris, \& Vasconcelos, Vasco~T. (2011).
\newblock Session typing for a featherweight {Erlang}.
\newblock {\em Pages  95--109 of:} {\em Coordination Models and Languages
  (COORDINATION)}.
\newblock LNCS, vol. 6721.
\newblock Springer.

\bibitem[\protect\citename{Neykova \& Yoshida, }2014]{Neykova-Yoshida-2014}
Neykova, Rumyana, \& Yoshida, Nobuko. (2014).
\newblock Multiparty session actors.
\newblock {\em Pages  131--146 of:} {\em Coordination Models and Languages
  (COORDINATION)}.
\newblock LNCS, vol. 8459.
\newblock Springer.

\bibitem[\protect\citename{Neykova \& Yoshida, }2017]{Neykova-Yoshida-2017}
Neykova, Rumyana, \& Yoshida, Nobuko. (2017).
\newblock Let it recover: Multiparty protocol-induced recovery.
\newblock {\em Pages  98--108 of:} {\em International Conference on Compiler
  Construction (CC)}.
\newblock ACM.

\bibitem[\protect\citename{Neykova {\em et~al.}\relax,
  }2013]{Neykova-et-al-2013}
Neykova, Rumyana, Yoshida, Nobuko, \& Hu, Raymond. (2013).
\newblock {SPY:} local verification of global protocols.
\newblock {\em Pages  358--363 of:} {\em International Conference on Runtime
  Verification (RV)}.
\newblock LNCS, vol. 8174.
\newblock Springer.

\bibitem[\protect\citename{Ng {\em et~al.}\relax, }2012]{Ng-et-al-2012}
Ng, Nicholas, Yoshida, Nobuko, \& Honda, Kohei. (2012).
\newblock {Multiparty Session C}: Safe parallel programming with message
  optimisation.
\newblock {\em Pages  202--218 of:} {\em International Conference on Modelling
  Techniques and Tools for Computer Performance Evaluation (TOOLS)}.
\newblock LNCS, vol. 7304.
\newblock Springer.

\bibitem[\protect\citename{Norell, }2009]{Norell-2009}
Norell, Ulf. (2009).
\newblock Dependently typed programming in {Agda}.
\newblock {\em Pages  1--2 of:} {\em Proceedings of the 4th International
  Workshop on Types in Language Design and Implementation}.
\newblock TLDI '09.
\newblock ACM.

\bibitem[\protect\citename{Orchard \& Yoshida, }2016]{Orchard-Yoshida-2016}
Orchard, Dominic, \& Yoshida, Nobuko. (2016).
\newblock Effects as sessions, sessions as effects.
\newblock {\em Pages  568--581 of:} {\em Principles of Programming Languages
  (POPL)}.
\newblock ACM.

\bibitem[\protect\citename{Ou {\em et~al.}\relax, }2004]{Ou-et-al-2004}
Ou, Xinming, Tan, Gang, Mandelbaum, Yitzhak, \& Walker, David. (2004).
\newblock Dynamic typing with dependent types.
\newblock {\em Pages  437--450 of:} {\em IFIP International Conference on
  Theoretical Computer Science},  vol. 155.
\newblock Springer.

\bibitem[\protect\citename{Padovani, }2017]{Padovani-2017}
Padovani, Luca. (2017).
\newblock A simple library implementation of binary sessions.
\newblock {\em Journal of functional programming}, {\bf 27}, e4.

\bibitem[\protect\citename{Pfenning \& Griffith, }2015]{Pfenning-Griffith-2015}
Pfenning, Frank, \& Griffith, Dennis. (2015).
\newblock Polarized substructural session types.
\newblock {\em Pages  3--22 of:} {\em International Conference on Foundations
  of Software Science and Computation Structures}.
\newblock LNCS, vol. 9034.
\newblock Springer.

\bibitem[\protect\citename{Pierce, }2002]{Pierce2002-tpl}
Pierce, Benjamin~C. (2002).
\newblock {\em Types and programming languages}.
\newblock MIT Press.

\bibitem[\protect\citename{Polakow, }2015]{Polakow-2015}
Polakow, Jeff. (2015).
\newblock Embedding a full linear lambda calculus in {Haskell}.
\newblock {\em Pages  177--188 of:} {\em Symposium on Haskell}.
\newblock ACM.

\bibitem[\protect\citename{Pucella \& Tov, }2008]{Pucella-Tov-2008}
Pucella, Riccardo, \& Tov, Jesse~A. (2008).
\newblock {Haskell} session types with (almost) no class.
\newblock {\em Pages  25--36 of:} {\em Symposium on Haskell}.
\newblock ACM.

\bibitem[\protect\citename{Sackman \& Eisenbach, }2008]{Sackman-Eisenbach-2008}
Sackman, Matthew, \& Eisenbach, Susan. (2008).
\newblock {\em Session types in {Haskell}: Updating message passing for the
  21st century}.

\bibitem[\protect\citename{Scalas \& Yoshida, }2016]{Scalas-Yoshida-2016}
Scalas, Alceste, \& Yoshida, Nobuko. (2016).
\newblock Lightweight session programming in {Scala}.
\newblock {\em Pages  21:1--21:28 of:} {\em European Conference on
  Object-Oriented Programming (ECOOP)}.
\newblock LIPIcs.
\newblock Schloss Dagstuhl.

\bibitem[\protect\citename{Schwerter {\em et~al.}\relax,
  }2016]{DBLP:journals/jfp/SchwerterGT16}
Schwerter, Felipe~Ba{\~{n}}ados, Garcia, Ronald, \& Tanter, {\'{E}}ric. (2016).
\newblock Gradual type-and-effect systems.
\newblock {\em J. funct. program.}, {\bf 26}, e19.

\bibitem[\protect\citename{Sergey \& Clarke, }2012]{DBLP:conf/esop/SergeyC12}
Sergey, Ilya, \& Clarke, Dave. (2012).
\newblock Gradual ownership types.
\newblock {\em Pages  579--599 of:} Seidl, Helmut (ed), {\em Programming
  Languages and Systems - 21st European Symposium on Programming, {ESOP} 2012,
  Held as Part of the European Joint Conferences on Theory and Practice of
  Software, {ETAPS} 2012, Tallinn, Estonia, March 24 - April 1, 2012.
  Proceedings}.
\newblock Lecture Notes in Computer Science, vol. 7211.
\newblock Springer.

\bibitem[\protect\citename{Siek {\em et~al.}\relax,
  }2015a]{Siek-et-al-2015-coercion}
Siek, Jeremy, Thiemann, Peter, \& Wadler, Philip. (2015a).
\newblock Blame and coercion: Together again for the first time.
\newblock {\em Pages  425--435 of:} {\em Programming Language Design and
  Implementation (PLDI)}.

\bibitem[\protect\citename{Siek \& Taha, }2006]{Siek-Taha-2006}
Siek, Jeremy~G., \& Taha, Walid. 2006 (Sept.).
\newblock Gradual typing for functional languages.
\newblock {\em Pages  81--92 of:} {\em Scheme and Functional Programming
  Workshop (Scheme)}.

\bibitem[\protect\citename{Siek \& Taha, }2007]{DBLP:conf/ecoop/SiekT07}
Siek, Jeremy~G., \& Taha, Walid. (2007).
\newblock Gradual typing for objects.
\newblock {\em Pages  2--27 of:} {\em {ECOOP}}.
\newblock Lecture Notes in Computer Science, vol. 4609.
\newblock Springer.

\bibitem[\protect\citename{Siek \& Tobin{-}Hochstadt,
  }2016]{DBLP:conf/birthday/SiekT16}
Siek, Jeremy~G., \& Tobin{-}Hochstadt, Sam. (2016).
\newblock The recursive union of some gradual types.
\newblock {\em Pages  388--410 of:} Lindley, Sam, McBride, Conor, Trinder,
  Philip~W., \& Sannella, Donald (eds), {\em A List of Successes That Can
  Change the World - Essays Dedicated to Philip Wadler on the Occasion of His
  60th Birthday}.
\newblock Lecture Notes in Computer Science, vol. 9600.
\newblock Springer.

\bibitem[\protect\citename{Siek {\em et~al.}\relax,
  }2015b]{Siek-et-al-2015-criteria}
Siek, Jeremy~G., Vitousek, Michael~M., Cimini, Matteo, \& Boyland, John~T.
  (2015b).
\newblock Refined criteria for gradual typing.
\newblock {\em Pages  274--293 of:} {\em Summit on Advances in Programming
  Languages (SNAPL)}.
\newblock LIPIcs, vol. 32.
\newblock Schloss Dagstuhl.

\bibitem[\protect\citename{{The Coq Development Team},
  }2019]{the_coq_development_team_2019_2554024}
{The Coq Development Team}. 2019 (Jan.).
\newblock {\em The coq proof assistant, version 8.9.0}.

\bibitem[\protect\citename{{The Dart Team}, }2014]{Dart-2014}
{The Dart Team}. (2014).
\newblock {\em {Dart} programming language specification}.
\newblock Google, 1.2 edition.

\bibitem[\protect\citename{Thiemann, }2014]{DBLP:conf/tgc/Thiemann14}
Thiemann, Peter. (2014).
\newblock Session types with gradual typing.
\newblock {\em Pages  144--158 of:} {\em {TGC}}.
\newblock LNCS, vol. 8902.
\newblock Springer.

\bibitem[\protect\citename{Thiemann \& Fennell,
  }2014]{DBLP:conf/esop/ThiemannF14}
Thiemann, Peter, \& Fennell, Luminous. (2014).
\newblock Gradual typing for annotated type systems.
\newblock {\em Pages  47--66 of:} Shao, Zhong (ed), {\em Programming Languages
  and Systems - 23rd European Symposium on Programming, {ESOP} 2014, Held as
  Part of the European Joint Conferences on Theory and Practice of Software,
  {ETAPS} 2014, Grenoble, France, April 5-13, 2014, Proceedings}.
\newblock Lecture Notes in Computer Science, vol. 8410.
\newblock Springer.

\bibitem[\protect\citename{Tobin{-}Hochstadt \& Felleisen,
  }2008]{DBLP:conf/popl/Tobin-HochstadtF08}
Tobin{-}Hochstadt, Sam, \& Felleisen, Matthias. (2008).
\newblock The design and implementation of typed {Scheme}.
\newblock {\em Pages  395--406 of:} Necula, George~C., \& Wadler, Philip (eds),
  {\em Proceedings of the 35th {ACM} {SIGPLAN-SIGACT} Symposium on Principles
  of Programming Languages, {POPL} 2008, San Francisco, California, USA,
  January 7-12, 2008}.
\newblock {ACM}.

\bibitem[\protect\citename{Toninho \& Yoshida,
  }2018]{DBLP:conf/fossacs/ToninhoY18}
Toninho, Bernardo, \& Yoshida, Nobuko. (2018).
\newblock Depending on session-typed processes.
\newblock {\em Pages  128--145 of:} {\em FoSSaCS}.
\newblock Lecture Notes in Computer Science, vol. 10803.
\newblock Springer.

\bibitem[\protect\citename{Toninho {\em et~al.}\relax,
  }2013]{Toninho-et-al-2013}
Toninho, Bernardo, Caires, Lu{\'\i}s, \& Pfenning, Frank. (2013).
\newblock Higher-order processes, functions, and sessions: A monadic
  integration.
\newblock {\em Pages  350--369 of:} {\em European Symposium on Programming
  (ESOP)}.
\newblock LNCS, vol. 7792.
\newblock Springer.

\bibitem[\protect\citename{Toro {\em et~al.}\relax,
  }2018]{DBLP:journals/toplas/ToroGT18}
Toro, Mat{\'{\i}}as, Garcia, Ronald, \& Tanter, {\'{E}}ric. (2018).
\newblock Type-driven gradual security with references.
\newblock {\em {ACM} trans. program. lang. syst.}, {\bf 40}(4), 16:1--16:55.

\bibitem[\protect\citename{Toro {\em et~al.}\relax,
  }2019]{DBLP:journals/pacmpl/ToroLT19}
Toro, Mat{\'{\i}}as, Labrada, Elizabeth, \& Tanter, {\'{E}}ric. (2019).
\newblock Gradual parametricity, revisited.
\newblock {\em {PACMPL}}, {\bf 3}({POPL}), 17:1--17:30.

\bibitem[\protect\citename{Tov \& Pucella, }2010]{Tov-Pucella-2010}
Tov, Jesse~A., \& Pucella, Riccardo. (2010).
\newblock Stateful contracts for affine types.
\newblock {\em Pages  550--569 of:} {\em European Symposium on Programming
  (ESOP)}.
\newblock LNCS, vol. 6012.
\newblock Springer.

\bibitem[\protect\citename{Vasconcelos,
  }2012]{vasconcelos:fundamental-sessions}
Vasconcelos, Vasco~Thudichum. (2012).
\newblock Fundamentals of session types.
\newblock {\em Information and computation}, {\bf 217}, 52--70.

\bibitem[\protect\citename{Vazou {\em et~al.}\relax,
  }2018]{DBLP:journals/pacmpl/VazouTH18}
Vazou, Niki, Tanter, {\'{E}}ric, \& Horn, David~Van. (2018).
\newblock Gradual liquid type inference.
\newblock {\em {PACMPL}}, {\bf 2}({OOPSLA}), 132:1--132:25.

\bibitem[\protect\citename{Verlaguet, }2013]{Verlaguet-2013}
Verlaguet, Julien. (2013).
\newblock Facebook: Analysing {PHP} statically.
\newblock  {\em Workshop on Commercial Uses of Functional Programming (CUFP)}.

\bibitem[\protect\citename{Vitousek {\em et~al.}\relax,
  }2017]{DBLP:conf/popl/VitousekSS17}
Vitousek, Michael~M., Swords, Cameron, \& Siek, Jeremy~G. (2017).
\newblock Big types in little runtime: Open-world soundness and collaborative
  blame for gradual type systems.
\newblock {\em Pages  762--774 of:} Castagna, Giuseppe, \& Gordon, Andrew~D.
  (eds), {\em Proceedings of the 44th {ACM} {SIGPLAN} Symposium on Principles
  of Programming Languages, {POPL} 2017, Paris, France, January 18-20, 2017}.
\newblock {ACM}.

\bibitem[\protect\citename{Wadler, }2012]{Wadler-2012}
Wadler, Philip. (2012).
\newblock Propositions as sessions.
\newblock {\em Pages  273--286 of:} {\em International Conference on Functional
  Programming (ICFP)}.
\newblock ACM.

\bibitem[\protect\citename{Wadler, }2014]{Wadler-2014}
Wadler, Philip. (2014).
\newblock Propositions as sessions.
\newblock {\em Journal of functional programming}, {\bf 24}(2-3), 384--418.

\bibitem[\protect\citename{Wadler, }2015]{Wadler-2015}
Wadler, Philip. (2015).
\newblock A complement to blame.
\newblock {\em Pages  309--320 of:} {\em 1st Summit on Advances in Programming
  Languages ({SNAPL})}.
\newblock LIPIcs, vol. 32.
\newblock Schloss Dagstuhl.

\bibitem[\protect\citename{Wadler \& Findler, }2009]{Wadler-Findler-2009}
Wadler, Philip, \& Findler, Robert~Bruce. (2009).
\newblock Well-typed programs can't be blamed.
\newblock {\em Pages  1--16 of:} {\em European Symposium on Programming
  (ESOP)}.
\newblock LNCS, vol. 5502.
\newblock Springer.

\bibitem[\protect\citename{Walker, }2005]{walker:substructural-type-systems}
Walker, David. (2005).
\newblock {\em Advanced topics in types and programming languages}.
\newblock MIT Press.
\newblock Chap. Substructural Type Systems, pages  3--43.

\bibitem[\protect\citename{Williams {\em et~al.}\relax,
  }2017]{Williams-et-al-2017}
Williams, Jack, Morris, J.~Garrett, Wadler, Philip, \& Zalewski, Jakub. (2017).
\newblock Mixed messages: Measuring conformance and non-interference in
  {TypeScript}.
\newblock {\em Pages  28:1--28:29 of:} {\em European Conference on
  Object-Oriented Programming (ECOOP)}.
\newblock LIPIcs, vol. 74.
\newblock Dagstuhl, Germany: Schloss Dagstuhl.

\bibitem[\protect\citename{Willsey {\em et~al.}\relax,
  }2017]{Willsey-et-al-2016}
Willsey, Max, Prabhu, Rokhini, \& Pfenning, Frank. (2017).
\newblock Design and implementation of concurrent {C0}.
\newblock {\em Pages  73--82 of:} {\em International Workshop on Linearity}.
\newblock EPTCS, vol. 238.

\bibitem[\protect\citename{Wolff {\em et~al.}\relax, }2011]{Wolf-et-al-2011}
Wolff, Roger, Garcia, Ronald, Tanter, {\'E}ric, \& Aldrich, Jonathan. (2011).
\newblock Gradual typestate.
\newblock {\em Pages  459--483 of:} {\em European Conference on Object-Oriented
  Programming (ECOOP)}.
\newblock LNCS, vol. 6813.
\newblock Springer.

\bibitem[\protect\citename{Yoshida \& Vasconcelos,
  }2007]{Yoshida-Vasconcelos-2007}
Yoshida, Nobuko, \& Vasconcelos, Vasco. (2007).
\newblock Language primitives and type discipline for structured
  communication-based programming revisited: Two systems for higher-order
  session communication.
\newblock {\em Entcs}, {\bf 171}(4), 73--93.

\bibitem[\protect\citename{Yoshida {\em et~al.}\relax,
  }2014]{Yoshida-et-al-2013}
Yoshida, Nobuko, Hu, Raymond, Neykova, Rumyana, \& Ng, Nicholas. (2014).
\newblock The {Scribble} protocol language.
\newblock {\em Pages  22--41 of:} {\em International Symposium on Trustworthy
  Global Computing}.
\newblock LNCS, vol. 8358.
\newblock Springer.

\end{thebibliography}

\clearpage
\appendix
\section{Typechecking Algorithm for the External Language}
\label{sec:typech-algor-extern}

We give a typechecking algorithm for \GGVe{} and show that it is correct.
The typechecking algorithm is slightly involved due to linearity:
$ \textsc{\tyexp}( \Gamma ,  \mathbb{e} ) $ outputs a pair of type $T$ and a
set $X$ of variables, containing the linear variables occurring free in $\mathbb{e}$.

%% Typechecking algorithm for expressions

\begin{algorithmic}
%%  \Require {A type environment $\Gamma$ and an expression $e$}
%%  \Ensure {The type $T$ and the set of used linear variables in typing}
%%  assertion fails if not typable

  %% type-check-expr
  \Function{\tyexp}{$\Gamma, \mathbb{e}$}
    \Case{$\mathbb{e}$}
      \CaseItem{$z$}
        \State \Assert $z \in \dom(\Gamma)$
        \State $T \asgn \Gamma(z)$
        \If {$\lin(T)$} \Return $T$, $\set{z}$
        \Else \Return $T$, $\emptyset$
        \EndIf
      \EndCaseItem

      \CaseItem{$()$}
        \Return $ \unit $, $\emptyset$
      \EndCaseItem

      %% Abs
      \CaseItem{$ \lambda  _ \ottnt{m}   \ottmv{x} {:} \ottnt{T_{{\mathrm{1}}}} .\,  \mathbb{e}_{{\mathrm{1}}} $}
        \State $T_2, Y \asgn$ \Call{\tyexp}{$\ottsym{(}  \Gamma  \ottsym{,}  \ottmv{x}  \ottsym{:}  \ottnt{T_{{\mathrm{1}}}}  \ottsym{)}, \mathbb{e}_{{\mathrm{1}}}$}

        \If {$\lin(T_1)$ and $m = \un$}
          \State \Assert $Y = \set{x}$
          \State \Return $ \ottnt{T_{{\mathrm{1}}}}   \rightarrow _{  \un  }  \ottnt{T_{{\mathrm{2}}}} $, $\emptyset$

        \ElsIf {$\lin(T_1)$ and $m = \lin$}
          \State \Assert $x \in Y$
          \State \Return $ \ottnt{T_{{\mathrm{1}}}}   \rightarrow _{  \lin  }  \ottnt{T_{{\mathrm{2}}}} $, $Y \setminus \set{x}$

        \ElsIf {$\un(T_1)$ and $m = \un$}
          \State \Assert $Y = \emptyset$
          \State \Return $ \ottnt{T_{{\mathrm{1}}}}   \rightarrow _{  \un  }  \ottnt{T_{{\mathrm{2}}}} $, $\emptyset$

        \Else \Return $ \ottnt{T_{{\mathrm{1}}}}   \rightarrow _{  \lin  }  \ottnt{T_{{\mathrm{2}}}} $, $Y$
        \EndIf
      \EndCaseItem

      %% App
      \CaseItem{$\mathbb{e}_{{\mathrm{1}}} \, \mathbb{e}_{{\mathrm{2}}}$}
        \State $T_1, X \asgn$ \Call{\tyexp}{$\Gamma, \mathbb{e}_{{\mathrm{1}}}$};\,
               $T_2, Y \asgn$ \Call{\tyexp}{$\Gamma, \mathbb{e}_{{\mathrm{2}}}$}
        %% ASSERT_DISJOINT
        \State \Assert $X \cap Y = \emptyset$

        \State $ \ottnt{T_{{\mathrm{11}}}}   \rightarrow _{ \ottnt{m} }  \ottnt{T_{{\mathrm{12}}}}  \asgn$ \Call{MatchingFun}{$T_1$}
        \State \Assert $\ottnt{T_{{\mathrm{2}}}}  \lesssim  \ottnt{T_{{\mathrm{11}}}}$
        \State \Return $T_{12}$, $X \cup Y$
      \EndCaseItem

      %% PairCons
      \CaseItem{$ (  \mathbb{e}_{{\mathrm{1}}} ,\,  \mathbb{e}_{{\mathrm{2}}}  )_ \ottnt{m} $}
        \State $T_1, X \asgn$ \Call{\tyexp}{$\Gamma, \mathbb{e}_{{\mathrm{1}}}$};\,
               $T_2, Y \asgn$ \Call{\tyexp}{$\Gamma, \mathbb{e}_{{\mathrm{2}}}$}
        %% ASSERT_DISJOINT
        \State \Assert $X \cap Y = \emptyset$

        \If {$m = \un$}
          \Assert $\un(T_1)$ and $\un(T_2)$
        \EndIf
        \State \Return $ \ottnt{T_{{\mathrm{1}}}}   \times _{ \ottnt{m} }  \ottnt{T_{{\mathrm{2}}}} $, $X \cup Y$
      \EndCaseItem

      %% PairDest
      \CaseItem {$ \letin{  \ottmv{x_{{\mathrm{1}}}} ,  \ottmv{x_{{\mathrm{2}}}}  }{ \mathbb{e}_{{\mathrm{1}}} }{ \mathbb{e}_{{\mathrm{2}}} } $}
        \State $T, Y \asgn$ \Call{\tyexp}{$\Gamma, \mathbb{e}_{{\mathrm{1}}}$}
        \State $ \ottnt{T_{{\mathrm{1}}}}   \times _{ \ottnt{m} }  \ottnt{T_{{\mathrm{2}}}}  \asgn$ \Call{MatchingProd}{$T$}

        \State $U, Z \asgn$ \Call{\tyexp}{$\ottsym{(}  \Gamma  \ottsym{,}  \ottmv{x}  \ottsym{:}  \ottnt{T_{{\mathrm{1}}}}  \ottsym{,}  \ottmv{y}  \ottsym{:}  \ottnt{T_{{\mathrm{2}}}}  \ottsym{)}, \mathbb{e}_{{\mathrm{2}}}$}

        %% If you introduced linear variables, you have to use them.
        \If {$\lin(T_1)$}
          \State \Assert $x_1 \in Z$
          %% Z is overwritten here.
          \State $Z \asgn Z \setminus \set{x_1}$
        \EndIf
        \If {$\lin(T_2)$}
          \State \Assert $x_2 \in Z$
          \State $Z \asgn Z \setminus \set{x_2}$
        \EndIf

        %% ASSERT_DISJOINT
        \State \Assert $Y \cap Z = \emptyset$
        \State \Return $U$, $Y \cup Z$
      \EndCaseItem

      \CaseItem {$ \fork{ \mathbb{e}_{{\mathrm{1}}} } $}
        \State $T, X \asgn$ \Call{\tyexp}{$\Gamma, \mathbb{e}_{{\mathrm{1}}}$}
        \State \Assert $\ottnt{T}  \sim   \unit $
        \State \Return $\unit$, $X$
      \EndCaseItem

      %% S is added in syntax
      \CaseItem {$ \newk\, \ottnt{S} $}
        \Return $ \ottnt{S}   \times _{  \lin  }   \dual{  \ottnt{S}  }  $, $\emptyset$
      \EndCaseItem

      %% Send
      \CaseItem {$ \send{ \mathbb{e}_{{\mathrm{1}}} }{ \mathbb{e}_{{\mathrm{2}}} } $}
        \State $T_1, X \asgn$ \Call{\tyexp}{$\Gamma, \mathbb{e}_{{\mathrm{1}}}$};\,
               $T_2, Y \asgn$ \Call{\tyexp}{$\Gamma, \mathbb{e}_{{\mathrm{2}}}$}
        %% ASSERT_DISJOINT
        \State \Assert $X \cap Y = \emptyset$

        \State $ \pl{ \ottnt{T_{{\mathrm{3}}}} }  \ottnt{S}  \asgn$ \Call{MatchingSend}{$T_2$}
        \State \Assert $\ottnt{T_{{\mathrm{1}}}}  \lesssim  \ottnt{T_{{\mathrm{3}}}}$
        \State \Return $S$, $X \cup Y$
      \EndCaseItem

      %% Receive
      \CaseItem {$ \recv{ \mathbb{e}_{{\mathrm{1}}} } $}
        \State $T_1, X \asgn$ \Call{\tyexp}{$\Gamma, \mathbb{e}_{{\mathrm{1}}}$}
        \State $ \qu{ \ottnt{T_{{\mathrm{2}}}} }  \ottnt{S}  \asgn$ \Call{MatchingReceive}{$T_1$}
        \State \Return $ \ottnt{T_{{\mathrm{2}}}}   \times _{  \lin  }  \ottnt{S} $, $X$
      \EndCaseItem

      %% Select
      \CaseItem {$ \select{ \ottmv{l_{\ottmv{j}}} }{ \mathbb{e}_{{\mathrm{1}}} } $}
        \State $T, X \asgn$ \Call{\tyexp}{$\Gamma, \mathbb{e}_{{\mathrm{1}}}$}
        \State $ \oplus    \br{  \ottmv{l_{\ottmv{j}}}  :  \ottnt{S_{\ottmv{j}}}  }   \asgn$ \Call{MatchingSelect}{$T, \ottmv{l_{\ottmv{j}}}$}
        %\State \Assert $j \in I$
        \State \Return $S_j$, $X$
      \EndCaseItem

      %% Case
      %% J = {1,2, ..., k} && k = |J|
      \CaseItem {$ \case{ \mathbb{e}_{{\mathrm{0}}} }{  \br{  \ottmv{l_{{\mathrm{1}}}} :  \ottmv{x_{{\mathrm{1}}}} .\,  \mathbb{e}_{{\mathrm{1}}} , \dots,
               \ottmv{l_{\ottmv{k}}} :  \ottmv{x_{\ottmv{k}}} .\,  \mathbb{e}_{\ottmv{k}}  }  } $}
        \State $T, X \asgn$ \Call{\tyexp}{$\Gamma, \mathbb{e}_{{\mathrm{0}}}$}
        \State $ \&    \br{  \ottmv{l_{{\mathrm{1}}}} : \ottnt{R_{{\mathrm{1}}}} , \dots ,  \ottmv{l_{\ottmv{k}}} : \ottnt{R_{\ottmv{k}}}  }   \asgn$ \Call{MatchingCase}{$\ottnt{T}, \br{\ottmv{l_{{\mathrm{1}}}},\ldots, \ottmv{l_{\ottmv{k}}}}$}

        %% Typing each branch ej
        \For {$j = 1$ to $k$}
          \State $U_j, Y_j \asgn$ \Call{\tyexp}{$\ottsym{(}  \Gamma  \ottsym{,}  \ottmv{x_{\ottmv{j}}}  \ottsym{:}  \ottnt{R_{\ottmv{j}}}  \ottsym{)}, \mathbb{e}_{\ottmv{j}}$}
          \State \Assert $x_j \in Y_j$
          \State $Y_j \asgn Y_j \setminus \set{x_j}$
        \EndFor

        \State \Assert $Y_1 = \dots = Y_k (\mathrel{=:} Y)$
        %% Compute bigjoin of {Ui}. Should I write in for-loop?
        \State $U \asgn U_1 \join \dots \join U_k$

        \State \Assert $X \cap Y = \emptyset$
        \State \Return $U$, $X \cup Y$
      \EndCaseItem

      %% Close
      \CaseItem {$ \close{ \mathbb{e}_{{\mathrm{1}}} } $}
        \State $T, X \asgn$ \Call{\tyexp}{$\Gamma, \mathbb{e}_{{\mathrm{1}}}$}
        \State \Assert $\ottnt{T}  \sim   \Endpl $
        \State \Return $\unit$, $X$
      \EndCaseItem

      %% Wait
      \CaseItem {$ \wait{ \mathbb{e}_{{\mathrm{1}}} } $}
        \State $T, X \asgn$ \Call{\tyexp}{$\Gamma, \mathbb{e}_{{\mathrm{1}}}$}
        \State \Assert $\ottnt{T}  \sim   \Endqu $
        \State \Return $\unit$, $X$
      \EndCaseItem

    \EndCase
  \EndFunction
\end{algorithmic}

%% Auxiliary functions
%% matching functions

\begin{algorithmic}
  %% Fun
  \Function{MatchingFun}{$T$}
    \Case{$T$}
      \CaseItem {$ \ottnt{T_{{\mathrm{1}}}}   \rightarrow _{ \ottnt{m} }  \ottnt{T_{{\mathrm{2}}}} $}
        \Return $ \ottnt{T_{{\mathrm{1}}}}   \rightarrow _{ \ottnt{m} }  \ottnt{T_{{\mathrm{2}}}} $
      \EndCaseItem
      \CaseItem {$\DYN$}
        \Return $  \DYN    \rightarrow _{  \lin  }   \DYN  $
      \EndCaseItem
      \CaseItem {$\WC$} \Error
      \EndCaseItem
    \EndCase
  \EndFunction

  %% Prod
  \Function{MatchingProd}{$T$}
    \Case{$T$}
      \CaseItem {$ \ottnt{T_{{\mathrm{1}}}}   \times _{ \ottnt{m} }  \ottnt{T_{{\mathrm{2}}}} $}
        \Return $ \ottnt{T_{{\mathrm{1}}}}   \times _{ \ottnt{m} }  \ottnt{T_{{\mathrm{2}}}} $
      \EndCaseItem
      \CaseItem {$\DYN$}
        \Return $  \DYN    \times _{  \lin  }   \DYN  $
      \EndCaseItem
      \CaseItem {$\WC$} \Error
      \EndCaseItem
    \EndCase
  \EndFunction

  %% Send
  \Function{MatchingSend}{$T$}
    \Case{$T$}
      \CaseItem {$ \pl{ \ottnt{T'} }  \ottnt{S} $}
        \Return $ \pl{ \ottnt{T'} }  \ottnt{S} $
      \EndCaseItem
      \CaseItem {$\DC \mid \DYN$}
        \Return $ \pl{  \DYN  }   \DC  $
      \EndCaseItem
      \CaseItem {$\WC$} \Error
      \EndCaseItem
    \EndCase
  \EndFunction

  %% Receive
  \Function{MatchingReceive}{$T$}
    \Case{$T$}
      \CaseItem {$ \pl{ \ottnt{T'} }  \ottnt{S} $}
        \Return $ \pl{ \ottnt{T'} }  \ottnt{S} $
      \EndCaseItem
      \CaseItem {$\DC \mid \DYN$}
        \Return $ \pl{  \DYN  }   \DC  $
      \EndCaseItem
      \CaseItem {$\WC$} \Error
      \EndCaseItem
    \EndCase
  \EndFunction

  %% Select
  \Function{MatchingSelect}{$T, \ottmv{l_{\ottmv{j}}}$}
    \Case{$T$}
      \CaseItem {$ \oplus    \br{  \ottmv{l_{\ottmv{i}}}  :  \ottnt{S_{\ottmv{i}}}  }_{  \ottmv{i}  \in  \ottnt{I}  }  $}
        \State\Assert $j \in I$
        \State\Return $ \oplus    \br{  \ottmv{l_{\ottmv{j}}}  :  \ottnt{S_{\ottmv{j}}}  }  $
      \EndCaseItem
      \CaseItem {$\DC \mid \DYN$}
        \Return $ \oplus    \br{  \ottmv{l_{\ottmv{j}}}  :   \DC   }  $
      \EndCaseItem
      \CaseItem {$\WC$} \Error
      \EndCaseItem
    \EndCase
  \EndFunction

  %% Case
  \Function{MatchingCase}{$T, \set{\ottmv{l_{\ottmv{j}}}}_{j\in J}$}
    \Case{$T$}
      \CaseItem {$ \&    \br{  \ottmv{l_{\ottmv{i}}}  :  \ottnt{S_{\ottmv{i}}}  }_{  \ottmv{i}  \in  \ottnt{I}  }  $}
        %% Assert: I \subseteq J
        %% \If {$I = J$}
        %%   \State\Return $ \&    \br{  \ottmv{l_{\ottmv{i}}}  :  \ottnt{S_{\ottmv{i}}}  }_{  \ottmv{i}  \in  \ottnt{I}  }  $
        \If {$I \subseteq J$}
          \Return $\&  \br{  \ottmv{l_{\ottmv{i}}}  :  \ottnt{S_{\ottmv{i}}}  }_{  \ottmv{i}  \in  \ottnt{I}  }  \cup  \br{  \ottmv{l_{\ottmv{j}}}  :   \DC   }_{  \ottmv{j}  \in   \ottnt{J}  \setminus  \ottnt{I}   } $
        \Else\Error
        \EndIf
      \EndCaseItem
      \CaseItem {$\DC \mid \DYN$}
        \Return $ \&    \br{  \ottmv{l_{\ottmv{j}}}  :   \DC   }_{  \ottmv{j}  \in  \ottnt{J}  }  $
      \EndCaseItem
      \CaseItem {$\WC$} \Error
      \EndCaseItem
    \EndCase
  \EndFunction
\end{algorithmic}

Theorem~\ref{thm:tycheck-soundness} states soundness of the typechecking algorithm.
A few lemmas are required in preparation.
Let $ \remove( \Gamma ,  \ottnt{X} ) $ denote
the operation that removes variables $\ottnt{X}$
from a type environment $\Gamma$.

\begin{lemma} \label{lem:tycheck-remove-unused} % Strengthening
  Suppose $y\colon U \in \Gamma$ with $\lin(U)$.
  \begin{center}
    If $ \textsc{\tyexp}( \Gamma ,  \mathbb{e} ) =  \ottnt{T} ,  \ottnt{X} $ and $y \notin X$, then
    $ \textsc{\tyexp}(  \remove( \Gamma ,  \ottsym{\{}  \ottmv{y}  \ottsym{\}} )  ,  \mathbb{e} ) =  \ottnt{T} ,  \ottnt{X} $.
  \end{center}
\end{lemma}
\begin{proof}
  By induction on $\mathbb{e}$.
  We show one important case below.
  \begin{description}
  \item[Case] $\mathbb{e}  \ottsym{=}  \mathbb{e}_{{\mathrm{1}}} \, \mathbb{e}_{{\mathrm{2}}}$: We are given
    \[
      \Gamma(y) = U \gap
      \lin(U) \gap
       \textsc{\tyexp}( \Gamma ,  \mathbb{e}_{{\mathrm{1}}} \, \mathbb{e}_{{\mathrm{2}}} ) =  \ottnt{T} ,  \ottnt{X}  \gap
      y \notin X.
    \]
    %
    %Thus, % By $ \textsc{\tyexp}( \Gamma ,  \mathbb{e}_{{\mathrm{1}}} \, \mathbb{e}_{{\mathrm{2}}} ) =  \ottnt{T} ,  \ottnt{X} $, we have
    By the definition of the algorithm,
    $ \textsc{\tyexp}( \Gamma ,  \mathbb{e}_{\ottmv{i}} ) =  \ottnt{T_{\ottmv{i}}} ,  \ottnt{X_{\ottmv{i}}} $ for $i = 1, 2$ and
    \[
       \matching{  \ottnt{T_{{\mathrm{1}}}}  }{   \ottnt{T_{{\mathrm{11}}}}   \rightarrow _{ \ottnt{m} }  \ottnt{T_{{\mathrm{12}}}}   }  \gap
      \ottnt{T_{{\mathrm{2}}}}  \lesssim  \ottnt{T_{{\mathrm{11}}}} \gap
      \ottnt{T}  \ottsym{=}  \ottnt{T_{{\mathrm{12}}}} \gap
      X = X_1 \uplus X_2.
    \]
    Since $y \notin X$, we have $y \notin X_1$ and $y \notin X_2$.
    By $ \textsc{\tyexp}( \Gamma ,  \mathbb{e}_{\ottmv{i}} ) =  \ottnt{T_{\ottmv{i}}} ,  \ottnt{X_{\ottmv{i}}} $ and $y \notin X_i$ and the IH
    for $i = 1, 2$, we have
    \[
       \textsc{\tyexp}(  \remove( \Gamma ,  \ottsym{\{}  \ottmv{y}  \ottsym{\}} )  ,  \mathbb{e}_{{\mathrm{1}}} ) =  \ottnt{T_{{\mathrm{1}}}} ,  \ottnt{X_{{\mathrm{1}}}}  \gap
       \textsc{\tyexp}(  \remove( \Gamma ,  \ottsym{\{}  \ottmv{y}  \ottsym{\}} )  ,  \mathbb{e}_{{\mathrm{2}}} ) =  \ottnt{T_{{\mathrm{2}}}} ,  \ottnt{X_{{\mathrm{2}}}} .
    \]
    Thus, by the definition of the algorithm,
    $ \textsc{\tyexp}(  \remove( \Gamma ,  \ottsym{\{}  \ottmv{y}  \ottsym{\}} )  ,  \mathbb{e}_{{\mathrm{1}}} \, \mathbb{e}_{{\mathrm{2}}} ) =  \ottnt{T} ,  \ottnt{X} $.
  \end{description}
\end{proof}

\begin{lemma} \label{lem:flv-split}
  Suppose $ \flv{ \Gamma }  = X_1 \uplus X_2$.
  If $\Gamma_{{\mathrm{1}}}  \ottsym{=}   \remove( \Gamma ,  \ottnt{X_{{\mathrm{2}}}} ) $ and $\Gamma_{{\mathrm{2}}}  \ottsym{=}   \remove( \Gamma ,  \ottnt{X_{{\mathrm{1}}}} ) $, then
  $\Gamma  \ottsym{=}   \Gamma_{{\mathrm{1}}}  \circ  \Gamma_{{\mathrm{2}}} $ and
  $ \flv{ \Gamma_{{\mathrm{1}}} }   \ottsym{=}  \ottnt{X_{{\mathrm{1}}}}$ and $ \flv{ \Gamma_{{\mathrm{2}}} }   \ottsym{=}  \ottnt{X_{{\mathrm{2}}}}$.
\end{lemma}
\begin{proof}
  By induction on $\Gamma$.
\end{proof}

\begin{theorem}[Soundness of the typechecking algorithm] \label{thm:tycheck-soundness}
  If $ \textsc{\tyexp}( \Gamma ,  \mathbb{e} ) =  \ottnt{T} ,  \ottnt{X} $ and $ \flv{ \Gamma }   \ottsym{=}  \ottnt{X}$, then
  $\Gamma  \vdashG  \mathbb{e}  \ottsym{:}  \ottnt{T}$.
\end{theorem}
\begin{proof}
  By induction on $\mathbb{e}$.
  We show main cases below.
  \begin{description}
  %% Abs
  \item[Case] $\mathbb{e}  \ottsym{=}   \lambda  _ \ottnt{m}   \ottmv{x} {:} \ottnt{T_{{\mathrm{1}}}} .\,  \mathbb{e}_{{\mathrm{1}}} $: We are given
    \[
       \textsc{\tyexp}( \Gamma ,   \lambda  _ \ottnt{m}   \ottmv{x} {:} \ottnt{T_{{\mathrm{1}}}} .\,  \mathbb{e}_{{\mathrm{1}}}  ) =  \ottnt{T} ,  \ottnt{X}  \gap
       \flv{ \Gamma }   \ottsym{=}  \ottnt{X}.
    \]
    We consider only when $m = \un$ and $ \lin   \ottsym{(}  \ottnt{T_{{\mathrm{1}}}}  \ottsym{)}$.
    By the definition of the algorithm,
    \[
       \textsc{\tyexp}( \ottsym{(}  \Gamma  \ottsym{,}  \ottmv{x}  \ottsym{:}  \ottnt{T_{{\mathrm{1}}}}  \ottsym{)} ,  \mathbb{e}_{{\mathrm{1}}} ) =  \ottnt{T_{{\mathrm{2}}}} ,  \ottsym{\{}  \ottmv{x}  \ottsym{\}}  \gap
      X = \emptyset.
    \]
    % By $ \flv{ \Gamma }   \ottsym{=}  \ottnt{X}$,
    So, $ \flv{ \Gamma }   \ottsym{=}  \ottnt{X} = \emptyset$.
    By $ \lin   \ottsym{(}  \ottnt{T_{{\mathrm{1}}}}  \ottsym{)}$, we have $ \flv{ \Gamma  \ottsym{,}  \ottmv{x}  \ottsym{:}  \ottnt{T_{{\mathrm{1}}}} }   \ottsym{=}  \ottsym{\{}  \ottmv{x}  \ottsym{\}}$.
    %$ \textsc{\tyexp}( \ottsym{(}  \Gamma  \ottsym{,}  \ottmv{x}  \ottsym{:}  \ottnt{T_{{\mathrm{1}}}}  \ottsym{)} ,  \mathbb{e}_{{\mathrm{1}}} ) =  \ottnt{T_{{\mathrm{2}}}} ,  \ottsym{\{}  \ottmv{x}  \ottsym{\}} $ and
    Thus, by the IH, we have $\Gamma  \ottsym{,}  \ottmv{x}  \ottsym{:}  \ottnt{T_{{\mathrm{1}}}}  \vdash  \mathbb{e}_{{\mathrm{1}}}  \ottsym{:}  \ottnt{T_{{\mathrm{2}}}}$.
    %
    %and Lemma~\ref{lem:flv-empty-un},
    Here, by $ \flv{ \Gamma }  = \emptyset$, we have $ \un   \ottsym{(}  \Gamma  \ottsym{)}$.
    So, $ \super{   \un   }{  \Gamma  } $.
    We finish by
    \[
      \Gamma  \ottsym{,}  \ottmv{x}  \ottsym{:}  \ottnt{T_{{\mathrm{1}}}}  \vdash  \mathbb{e}_{{\mathrm{1}}}  \ottsym{:}  \ottnt{T_{{\mathrm{2}}}} \gap
       \super{   \un   }{  \Gamma  } 
    \]
    and the abstraction rule. % T-Abs (GGVe)

  %% App
  \item[Case] $\mathbb{e}  \ottsym{=}  \mathbb{e}_{{\mathrm{1}}} \, \mathbb{e}_{{\mathrm{2}}}$: We are given
    \[
       \textsc{\tyexp}( \Gamma ,  \mathbb{e}_{{\mathrm{1}}} \, \mathbb{e}_{{\mathrm{2}}} ) =  \ottnt{T} ,  \ottnt{X}  \gap
       \flv{ \Gamma }   \ottsym{=}  \ottnt{X}.
    \]
    %Thus, %By $ \textsc{\tyexp}( \Gamma ,  \mathbb{e}_{{\mathrm{1}}} \, \mathbb{e}_{{\mathrm{2}}} ) =  \ottnt{T} ,  \ottnt{X} $,
    By the definition of the algorithm,
    $ \textsc{\tyexp}( \Gamma ,  \mathbb{e}_{\ottmv{i}} ) =  \ottnt{T_{\ottmv{i}}} ,  \ottnt{X_{\ottmv{i}}} $ for $i = 1,2$ and
    \[
       \matching{  \ottnt{T_{{\mathrm{1}}}}  }{   \ottnt{T_{{\mathrm{11}}}}   \rightarrow _{ \ottnt{m} }  \ottnt{T_{{\mathrm{12}}}}   }  \gap
      \ottnt{T_{{\mathrm{2}}}}  \lesssim  \ottnt{T_{{\mathrm{11}}}} \gap
      \ottnt{T}  \ottsym{=}  \ottnt{T_{{\mathrm{12}}}} \gap
      X = X_1 \uplus X_2.
    \]
    Let $\Gamma_{{\mathrm{1}}}  \ottsym{=}   \remove( \Gamma ,  \ottnt{X_{{\mathrm{2}}}} ) $ and $\Gamma_{{\mathrm{2}}}  \ottsym{=}   \remove( \Gamma ,  \ottnt{X_{{\mathrm{1}}}} ) $.
    By Lemma~\ref{lem:flv-split},
    $\Gamma  \ottsym{=}   \Gamma_{{\mathrm{1}}}  \circ  \Gamma_{{\mathrm{2}}} $ and
    $ \flv{ \Gamma_{{\mathrm{1}}} }   \ottsym{=}  \ottnt{X_{{\mathrm{1}}}}$ and $ \flv{ \Gamma_{{\mathrm{2}}} }   \ottsym{=}  \ottnt{X_{{\mathrm{2}}}}$.
    By $ \textsc{\tyexp}( \Gamma ,  \mathbb{e}_{\ottmv{i}} ) =  \ottnt{T_{\ottmv{i}}} ,  \ottnt{X_{\ottmv{i}}} $ and Lemma~\ref{lem:tycheck-remove-unused}, we have $ \textsc{\tyexp}( \Gamma_{\ottmv{i}} ,  \mathbb{e}_{\ottmv{i}} ) =  \ottnt{T_{\ottmv{i}}} ,  \ottnt{X_{\ottmv{i}}} $  for $i = 1,2$.
    By $ \flv{ \Gamma_{\ottmv{i}} }   \ottsym{=}  \ottnt{X_{\ottmv{i}}}$ and $ \textsc{\tyexp}( \Gamma_{\ottmv{i}} ,  \mathbb{e}_{\ottmv{i}} ) =  \ottnt{T_{\ottmv{i}}} ,  \ottnt{X_{\ottmv{i}}} $ and the IH, we have $\Gamma_{\ottmv{i}}  \vdash  \mathbb{e}_{\ottmv{i}}  \ottsym{:}  \ottnt{T_{\ottmv{i}}}$ for $i = 1,2$.
    We finish by
    \[
      \Gamma_{{\mathrm{1}}}  \vdash  \mathbb{e}_{{\mathrm{1}}}  \ottsym{:}  \ottnt{T_{{\mathrm{1}}}} \gap
      \Gamma_{{\mathrm{2}}}  \vdash  \mathbb{e}_{{\mathrm{2}}}  \ottsym{:}  \ottnt{T_{{\mathrm{2}}}} \gap
       \matching{  \ottnt{T_{{\mathrm{1}}}}  }{   \ottnt{T_{{\mathrm{11}}}}   \rightarrow _{ \ottnt{m} }  \ottnt{T_{{\mathrm{12}}}}   }  \gap
      \ottnt{T_{{\mathrm{2}}}}  \lesssim  \ottnt{T_{{\mathrm{11}}}}
    \]
    and the application rule. %T-App (GGVe)
  \end{description}
\end{proof}

\begin{sloppypar}
We will also show the converse of the theorem above. Completeness
states that $ \textsc{\tyexp}( \Gamma ,  \mathbb{e} ) $ computes a minimal type with respect to
negative subtyping.
\end{sloppypar}

\begin{theorem}[Completeness of the typechecking algorithm] \label{thm:tycheck-completeness}
  If $\Gamma  \vdashG  \mathbb{e}  \ottsym{:}  \ottnt{T}$, then
  $ \textsc{\tyexp}( \Gamma ,  \mathbb{e} ) =  \ottnt{T'} ,  \ottnt{X} $ and $ \ottnt{T'}   \Sub^{-}   \ottnt{T} $ and $ \flv{ \Gamma }   \ottsym{=}  \ottnt{X}$ for some $\ottnt{T'}$.
\end{theorem}

To prove this theorem, we need a stronger statement, namely Lemma~\ref{lem:stronger-tycheck-completeness}.
We define environment positive consistent subtyping, written $ \Gamma'   \Sub^{-}   \Gamma $, as
$\dom(\Gamma) \subseteq \dom(\Gamma')$ and $\Gamma(\ottmv{x})     \Sub^{-}     \Gamma'(\ottmv{x})$, for any
$\ottmv{x} \in \dom(\Gamma)$.
  Then,
the theorem follows from the fact that $    \Sub^{-}    $ on type environments
is reflexive.
We start with a few lemmas about $    \Sub^{-}    $.

\begin{lemma}\label{lem:poscon-matching}
  \begin{enumerate}
  \item If $ \ottnt{T'_{{\mathrm{1}}}}   \Sub^{-}   \ottnt{T_{{\mathrm{1}}}} $ and
    $ \matching{  \ottnt{T_{{\mathrm{1}}}}  }{   \ottnt{T_{{\mathrm{11}}}}   \rightarrow _{ \ottnt{m} }  \ottnt{T_{{\mathrm{12}}}}   } $, then there exist some
    $\ottnt{T'_{{\mathrm{11}}}}$, $\ottnt{T'_{{\mathrm{12}}}}$, and $n$ such that
    $\Call{MatchingFun}{T_1'} =  \ottnt{T'_{{\mathrm{11}}}}   \rightarrow _{ \ottnt{n} }  \ottnt{T'_{{\mathrm{12}}}} $ and
    $  \ottnt{T'_{{\mathrm{11}}}}   \rightarrow _{ \ottnt{n} }  \ottnt{T'_{{\mathrm{12}}}}    \Sub^{-}    \ottnt{T_{{\mathrm{11}}}}   \rightarrow _{ \ottnt{m} }  \ottnt{T_{{\mathrm{12}}}}  $.
  \item If $ \ottnt{T'_{{\mathrm{1}}}}   \Sub^{-}   \ottnt{T_{{\mathrm{1}}}} $ and
    $ \matching{  \ottnt{T_{{\mathrm{1}}}}  }{   \ottnt{T_{{\mathrm{11}}}}   \times _{ \ottnt{m} }  \ottnt{T_{{\mathrm{12}}}}   } $, then there exist some
    $\ottnt{T'_{{\mathrm{11}}}}$, $\ottnt{T'_{{\mathrm{12}}}}$, and $n$ such that
    $\Call{MatchingProd}{T_1'} =  \ottnt{T'_{{\mathrm{11}}}}   \times _{ \ottnt{n} }  \ottnt{T'_{{\mathrm{12}}}} $ and
    $  \ottnt{T'_{{\mathrm{11}}}}   \times _{ \ottnt{n} }  \ottnt{T'_{{\mathrm{12}}}}    \Sub^{-}    \ottnt{T_{{\mathrm{11}}}}   \times _{ \ottnt{m} }  \ottnt{T_{{\mathrm{12}}}}  $.
  \item If $ \ottnt{T'_{{\mathrm{1}}}}   \Sub^{-}   \ottnt{T_{{\mathrm{1}}}} $ and
    $ \matching{  \ottnt{T_{{\mathrm{1}}}}  }{   \pl{ \ottnt{T_{{\mathrm{11}}}} }  \ottnt{S_{{\mathrm{12}}}}   } $, then there exist some
    $\ottnt{T'_{{\mathrm{11}}}}$ and $\ottnt{S'_{{\mathrm{12}}}}$ such that \\
    $\Call{MatchingSend}{T_1'} =  \pl{ \ottnt{T'_{{\mathrm{11}}}} }  \ottnt{S'_{{\mathrm{12}}}} $ and
    $  \pl{ \ottnt{T'_{{\mathrm{11}}}} }  \ottnt{S'_{{\mathrm{12}}}}    \Sub^{-}    \pl{ \ottnt{T_{{\mathrm{11}}}} }  \ottnt{S_{{\mathrm{12}}}}  $.
  \item If $ \ottnt{T'_{{\mathrm{1}}}}   \Sub^{-}   \ottnt{T_{{\mathrm{1}}}} $ and
    $ \matching{  \ottnt{T_{{\mathrm{1}}}}  }{   \qu{ \ottnt{T_{{\mathrm{11}}}} }  \ottnt{S_{{\mathrm{12}}}}   } $, then there exist some
    $\ottnt{T'_{{\mathrm{11}}}}$ and $\ottnt{S'_{{\mathrm{12}}}}$ such that \\
    $\Call{MatchingReceive}{T_1'} =  \qu{ \ottnt{T'_{{\mathrm{11}}}} }  \ottnt{S'_{{\mathrm{12}}}} $ and
    $  \qu{ \ottnt{T'_{{\mathrm{11}}}} }  \ottnt{S'_{{\mathrm{12}}}}    \Sub^{-}    \qu{ \ottnt{T_{{\mathrm{11}}}} }  \ottnt{S_{{\mathrm{12}}}}  $.
  \item If $ \ottnt{T'_{{\mathrm{1}}}}   \Sub^{-}   \ottnt{T_{{\mathrm{1}}}} $ and
    $ \matching{  \ottnt{T_{{\mathrm{1}}}}  }{   \oplus    \br{  \ottmv{l_{\ottmv{j}}}  :  \ottnt{S_{\ottmv{j}}}  }    } $, then there exist some
    $\ottnt{S'_{\ottmv{j}}}$ such that \\
    $\Call{MatchingSelect}{T_1', l_j} =  \oplus    \br{  \ottmv{l_{\ottmv{j}}}  :  \ottnt{S'_{\ottmv{j}}}  }  $ and
    $  \oplus    \br{  \ottmv{l_{\ottmv{j}}}  :  \ottnt{S'_{\ottmv{j}}}  }     \Sub^{-}    \oplus    \br{  \ottmv{l_{\ottmv{j}}}  :  \ottnt{S_{\ottmv{j}}}  }   $.
  \item If $ \ottnt{T'_{{\mathrm{1}}}}   \Sub^{-}   \ottnt{T_{{\mathrm{1}}}} $ and
    $ \matching{  \ottnt{T_{{\mathrm{1}}}}  }{   \&    \br{  \ottmv{l_{{\mathrm{1}}}} : \ottnt{S_{{\mathrm{1}}}} , \dots ,  \ottmv{l_{\ottmv{k}}} : \ottnt{S_{\ottmv{k}}}  }    } $, then there exist some
    $\ottnt{S'_{{\mathrm{1}}}}, \ldots, \ottnt{S'_{\ottmv{k}}}$ such that
    $\Call{MatchingCase}{T_1', \br{l_1,\ldots, l_k}} =  \&    \br{  \ottmv{l_{{\mathrm{1}}}} : \ottnt{S'_{{\mathrm{1}}}} , \dots ,  \ottmv{l_{\ottmv{k}}} : \ottnt{S'_{\ottmv{k}}}  }  $ and
    $  \&    \br{  \ottmv{l_{{\mathrm{1}}}} : \ottnt{S'_{{\mathrm{1}}}} , \dots ,  \ottmv{l_{\ottmv{k}}} : \ottnt{S'_{\ottmv{k}}}  }     \Sub^{-}    \&    \br{  \ottmv{l_{{\mathrm{1}}}} : \ottnt{S_{{\mathrm{1}}}} , \dots ,  \ottmv{l_{\ottmv{k}}} : \ottnt{S_{\ottmv{k}}}  }   $.
  \end{enumerate}
\end{lemma}

\begin{proof}
  By case analysis on $ \matching{  \ottnt{T}  }{  \ottnt{U}  } $.
\end{proof}

\begin{lemma}\label{lem:poscon-un}
  If $ \ottnt{T}   \Sub^{-}   \ottnt{U} $ and $ \un   \ottsym{(}  \ottnt{U}  \ottsym{)}$, then $ \un   \ottsym{(}  \ottnt{T}  \ottsym{)}$.
\end{lemma}
\begin{proof}
  By case analysis on $ \ottnt{T}   \Sub^{-}   \ottnt{U} $.
\end{proof}

\begin{lemma}\label{lem:trans-polcon}
  \begin{enumerate}
  \item If $ \ottnt{T_{{\mathrm{1}}}}   \Sub^{-}   \ottnt{T_{{\mathrm{2}}}} $ and $\ottnt{T_{{\mathrm{2}}}}  \lesssim  \ottnt{T_{{\mathrm{3}}}}$, then $\ottnt{T_{{\mathrm{1}}}}  \lesssim  \ottnt{T_{{\mathrm{3}}}}$.
  \item If $\ottnt{T_{{\mathrm{1}}}}  \lesssim  \ottnt{T_{{\mathrm{2}}}}$ and $ \ottnt{T_{{\mathrm{2}}}}   \Sub^{+}   \ottnt{T_{{\mathrm{3}}}} $, then $\ottnt{T_{{\mathrm{1}}}}  \lesssim  \ottnt{T_{{\mathrm{3}}}}$.
  \end{enumerate}
\end{lemma}

\begin{proof}
  Both items are proved by simultaneous induction on $ \lesssim $.
\end{proof}

\begin{lemma}\label{lem:stronger-tycheck-completeness}
  If $\Gamma  \vdashG  \mathbb{e}  \ottsym{:}  \ottnt{T}$ and $ \Gamma'   \Sub^{-}   \Gamma $, then
  there exists $\ottnt{T'}$ such that
  $ \textsc{\tyexp}( \Gamma' ,  \mathbb{e} ) =  \ottnt{T'} ,  \ottnt{X} $ and $ \ottnt{T'}   \Sub^{-}   \ottnt{T} $ and $ \flv{ \Gamma }   \ottsym{=}  \ottnt{X}$.
\end{lemma}

\begin{proof}
  By induction on $\mathbb{e}$.  We show main cases below.
  \begin{description}
  \item[Case] $\mathbb{e} = \mathbb{e}_{{\mathrm{1}}} \, \mathbb{e}_{{\mathrm{2}}}$:
    By inversion of the typing relation,
    $\Gamma_{{\mathrm{1}}}  \vdashG  \mathbb{e}_{{\mathrm{1}}}  \ottsym{:}  \ottnt{T_{{\mathrm{1}}}}$ and
    $\Gamma_{{\mathrm{2}}}  \vdashG  \mathbb{e}_{{\mathrm{2}}}  \ottsym{:}  \ottnt{T_{{\mathrm{2}}}}$ and
    $\Gamma = \Gamma_{{\mathrm{1}}} \circ \Gamma_{{\mathrm{2}}}$ and
    $ \matching{  \ottnt{T_{{\mathrm{1}}}}  }{   \ottnt{T_{{\mathrm{11}}}}   \rightarrow _{ \ottnt{m} }  \ottnt{T_{{\mathrm{12}}}}   } $ and
    $\ottnt{T_{{\mathrm{2}}}}  \lesssim  \ottnt{T_{{\mathrm{11}}}}$
    for some $\Gamma_{{\mathrm{1}}}$, $\Gamma_{{\mathrm{2}}}$, $\ottnt{T_{{\mathrm{1}}}}$, $\ottnt{T_{{\mathrm{2}}}}$, $\ottnt{T_{{\mathrm{11}}}}$, $\ottnt{T_{{\mathrm{12}}}}$, and $\ottnt{m}$.
    It is easy to show $ \Gamma'   \Sub^{-}   \Gamma_{{\mathrm{1}}} $ and $ \Gamma'   \Sub^{-}   \Gamma_{{\mathrm{2}}} $.
    By the induction hypothesis, for some $\ottnt{T'_{{\mathrm{1}}}}$, $\ottnt{T'_{{\mathrm{2}}}}$, $X$, and $Y$,
    $T_1', X =$ \Call{\tyexp}{$\Gamma', \mathbb{e}_{{\mathrm{1}}}$} and $ \ottnt{T'_{{\mathrm{1}}}}   \Sub^{-}   \ottnt{T_{{\mathrm{1}}}} $ and $ \flv{ \Gamma_{{\mathrm{1}}} }   \ottsym{=}  \ottnt{X}$ and
    $T_2', Y =$ \Call{\tyexp}{$\Gamma', \mathbb{e}_{{\mathrm{2}}}$} and $ \ottnt{T'_{{\mathrm{2}}}}   \Sub^{-}   \ottnt{T_{{\mathrm{2}}}} $ and $ \flv{ \Gamma_{{\mathrm{2}}} }   \ottsym{=}  \ottnt{Y}$.
    Since $\Gamma_{{\mathrm{1}}} \circ \Gamma_{{\mathrm{2}}}$ is well defined, $X \cap Y$ must be $\emptyset$.
    By Lemma~\ref{lem:poscon-matching},
    $ \ottnt{T'_{{\mathrm{11}}}}   \rightarrow _{ \ottnt{n} }  \ottnt{T'_{{\mathrm{12}}}}  =$ \Call{MatchingFun}{$T_1'$} and
    $  \ottnt{T'_{{\mathrm{11}}}}   \rightarrow _{ \ottnt{n} }  \ottnt{T'_{{\mathrm{12}}}}    \Sub^{-}    \ottnt{T_{{\mathrm{11}}}}   \rightarrow _{ \ottnt{m} }  \ottnt{T_{{\mathrm{12}}}}  $ for some $\ottnt{T'_{{\mathrm{11}}}}$ and $\ottnt{T'_{{\mathrm{12}}}}$.
    By inversion of $    \Sub^{-}    $,
    we have $ \ottnt{T_{{\mathrm{11}}}}   \Sub^{+}   \ottnt{T'_{{\mathrm{11}}}} $ and $ \ottnt{T'_{{\mathrm{12}}}}   \Sub^{-}   \ottnt{T_{{\mathrm{12}}}} $.
    Then, $\ottnt{T'_{{\mathrm{2}}}}  \lesssim  \ottnt{T'_{{\mathrm{11}}}}$ is shown by Lemma~\ref{lem:trans-polcon}.
    It is easy to show $X \cup Y =  \flv{ \Gamma } $ because $\Gamma = \Gamma_{{\mathrm{1}}} \circ \Gamma_{{\mathrm{2}}}$.
    Finally, $ \ottnt{T'_{{\mathrm{12}}}}   \Sub^{-}   \ottnt{T_{{\mathrm{12}}}} $ finishes the case.
  \item[Case] $\mathbb{e} =  \case{ \mathbb{e}_{{\mathrm{0}}} }{  \br{  \ottmv{l_{\ottmv{j}}} :  \ottmv{x_{\ottmv{j}}} .\,  \mathbb{e}_{\ottmv{j}}   }_{  \ottmv{j}  \in  \ottnt{J}  }  } $:
    By inversion of the typing relation, we have
    $\Gamma_{{\mathrm{1}}}  \vdashG  \mathbb{e}_{{\mathrm{0}}}  \ottsym{:}  \ottnt{T_{{\mathrm{0}}}}$ and
    $ \matching{  \ottnt{T_{{\mathrm{0}}}}  }{   \&    \br{  \ottmv{l_{\ottmv{j}}}  :  \ottnt{R_{\ottmv{j}}}  }_{  \ottmv{j}  \in  \ottnt{J}  }    } $ and
    $(\Gamma_{{\mathrm{2}}}  \ottsym{,}  \ottmv{x_{\ottmv{j}}}  \ottsym{:}  \ottnt{R_{\ottmv{j}}}  \vdashG  \mathbb{e}_{\ottmv{j}}  \ottsym{:}  \ottnt{U_{\ottmv{j}}})_{\ottmv{j} \in \ottnt{J}}$ and
    $\ottnt{T} =  \vee \br{\ottnt{U_{\ottmv{j}}}}_{\ottmv{j}\in \ottnt{J}}$ and
    $\Gamma = \Gamma_{{\mathrm{1}}} \circ \Gamma_{{\mathrm{2}}}$ for some
    $\Gamma_{{\mathrm{1}}}$, $\Gamma_{{\mathrm{2}}}$, $\ottnt{R_{\ottmv{j}}}$, and $\ottnt{U_{\ottmv{j}}}$ (for $\ottmv{j} \in \ottnt{J}$).
    It is easy to show $ \Gamma'   \Sub^{-}   \Gamma_{{\mathrm{1}}} $ and $ \Gamma'   \Sub^{-}   \Gamma_{{\mathrm{2}}} $.
    By the induction hypothesis,
    $T_0', X =$ \Call{\tyexp}{$\Gamma', \mathbb{e}_{{\mathrm{0}}}$} and $ \ottnt{T'_{{\mathrm{0}}}}   \Sub^{-}   \ottnt{T_{{\mathrm{0}}}} $ and $ \flv{ \Gamma_{{\mathrm{1}}} }   \ottsym{=}  \ottnt{X}$
    for some $\ottnt{T'_{{\mathrm{0}}}}$ and $X$.
    By Lemma~\ref{lem:poscon-matching},
    $\Call{MatchingCase}{T_0', \br{l_1,\ldots, l_k}} =  \&    \br{  \ottmv{l_{{\mathrm{1}}}} : \ottnt{R'_{{\mathrm{1}}}} , \dots ,  \ottmv{l_{\ottmv{k}}} : \ottnt{R'_{\ottmv{k}}}  }  $ and
    $  \&    \br{  \ottmv{l_{\ottmv{j}}}  :  \ottnt{R'_{\ottmv{j}}}  }_{  \ottmv{j}  \in  \ottnt{J}  }     \Sub^{-}    \&    \br{  \ottmv{l_{\ottmv{j}}}  :  \ottnt{R_{\ottmv{j}}}  }_{  \ottmv{j}  \in  \ottnt{J}  }   $ for some $(\ottnt{R'_{\ottmv{j}}})_{\ottmv{j} \in \ottnt{J}}$.
    By inversion of $    \Sub^{-}    $, we have $( \ottnt{R'_{\ottmv{j}}}   \Sub^{-}   \ottnt{R_{\ottmv{j}}} )_{\ottmv{j} \in \ottnt{J}}$.
    By the induction hypothesis,  for any $\ottmv{j} \in \ottnt{J}$,
    there exist $\ottnt{U'_{\ottmv{j}}}$ and $Y_j$ such that
    $U_j', Y_j =$ \Call{\tyexp}{$(\Gamma', \ottmv{x_{\ottmv{j}}}:\ottnt{R'_{\ottmv{j}}}), \mathbb{e}_{\ottmv{j}}$} and $ \ottnt{U'_{\ottmv{j}}}   \Sub^{-}   \ottnt{U_{\ottmv{j}}} $
    and $ \flv{ \Gamma_{{\mathrm{2}}}  \ottsym{,}  \ottmv{x_{\ottmv{j}}}  \ottsym{:}  \ottnt{R_{\ottmv{j}}} }   \ottsym{=}  \ottnt{Y_{\ottmv{j}}}$.
    It is easy to show that $x_j \in Y_j$ for any $\ottmv{j} \in \ottnt{J}$
    and $Y_1 = \cdots = Y_k$ and $X \cap Y_1 = \emptyset$
    because $\ottnt{R_{\ottmv{j}}}$ is linear and $\Gamma_{{\mathrm{1}}} \circ \Gamma_{{\mathrm{2}}}$ is well defined.
    It is also easy to show $X \cup (Y_1 \setminus \br{\ottmv{x_{{\mathrm{1}}}}}) =  \flv{ \Gamma } $ because $\Gamma = \Gamma_{{\mathrm{1}}} \circ \Gamma_{{\mathrm{2}}}$.
    Finally, $ \vee \br{\ottnt{U'_{\ottmv{j}}}}_{\ottmv{j} \in \ottnt{J}}     \Sub^{-}      \vee \br{\ottnt{U_{\ottmv{j}}}}_{\ottmv{j} \in \ottnt{J}}$
    is shown by Lemma~\ref{lemma:polarised-subtyping-lub-glb}.
  \end{description}
\end{proof}

\end{document}